\documentclass{article}
\usepackage[top=1in, bottom=1in, left=1in, right=1in]{geometry}

\usepackage{savesym}
\savesymbol{Bbbk}
\usepackage{amsthm,amsmath,amssymb,amsfonts}
\usepackage{url}
\usepackage{tikz}
\usetikzlibrary{math}
\usetikzlibrary{calc}
\usepackage{mathrsfs}
\usepackage{nicefrac,bbm}
\usepackage{hyperref, cleveref}
\usepackage[boxruled]{algorithm2e} %
\usepackage{mathtools,pifont,enumitem} %
\usepackage[framemethod=TikZ]{mdframed}
\usepackage{color,fancybox,graphicx,subfigure}
\usepackage[many]{tcolorbox}

\savesymbol{st}
\usepackage{soul}
\restoresymbol{SOUL}{st}

\SetAlgorithmName{Algorithm}{algorithm}{List of Algorithms}
  \SetKwInput{kwInit}{Input}
  \SetKwFor{For}{For}{do}{end}
  \SetKwFor{ForEach}{For each}{do}{end}

\usepackage{pbox}

\newtheorem{assumption}{Assumption}
\newtheorem{theorem}{Theorem}[section]
\newtheorem{lemma}[theorem]{Lemma}
\newtheorem{definition}[theorem]{Definition}

\newtheorem{corollary}[theorem]{Corollary}

\newtheorem{observation}[theorem]{Observation}
\newtheorem{proposition}[theorem]{Proposition}
\newtheorem{fact}[theorem]{Fact}

\newtheorem*{theorem*}{Theorem}

\crefname{section}{Section}{Sections}

\crefname{theorem}{Theorem}{Theorems}
\crefname{assumption}{Assumption}{Assumptions}
\crefname{lemma}{Lemma}{Lemmas}
\crefname{definition}{Definition}{Definitions}
\crefname{conjecture}{Conjecture}{Conjectures}
\crefname{corollary}{Corollary}{Corollaries}
\crefname{construction}{Construction}{Constructions}
\crefname{claim}{Claim}{Claims}
\crefname{observation}{Observation}{Observations}
\crefname{proposition}{Proposition}{Propositions}
\crefname{fact}{Fact}{Facts}
\crefname{question}{Question}{Questions}
\crefname{problem}{Problem}{Problems}
\crefname{remark}{Remark}{Remarks}
\crefname{example}{Example}{Examples}
\crefname{appendix}{Appendix}{Appendices}
\crefname{figure}{Figure}{Figures}
\crefname{equation}{Equation}{Equations}

\newcommand{\yesnum}{\addtocounter{equation}{1}\tag{\theequation}}
\newcommand{\tagnum}[1]{\addtocounter{equation}{1}{\tag{#1)~~(\theequation}}}
\makeatletter
\newcommand{\customlabel}[2]{%
   \protected@write \@auxout {}{\string \newlabel {#1}{{#2}{\thepage}{#2}{#1}{}} }%
   \hypertarget{#1}{}
}
\makeatother

\mdfdefinestyle{FrameBox}{ linecolor=white, outerlinewidth=0pt, roundcorner=5pt,
                          innertopmargin=5pt, innerbottommargin=5pt,
                          innerrightmargin=10pt, innerleftmargin=10pt, backgroundcolor=white}

\mdfdefinestyle{FrameBox2}{ linecolor=black, outerlinewidth=0pt, roundcorner=5pt,
                          innertopmargin=5pt, innerbottommargin=5pt,
                          innerrightmargin=10pt, innerleftmargin=10pt, backgroundcolor=white}

\mdfdefinestyle{FrameQuestion2}{ linecolor=white, outerlinewidth=0pt,
                        roundcorner=5pt,
                        innertopmargin=5pt, innerbottommargin=0pt,
                        innerrightmargin=16pt, innerleftmargin=16pt, backgroundcolor=white}

\newcommand{\white}[1]{\textcolor{white}{#1}}

\colorlet{RED}{black}
\colorlet{BLUE}{black}

\newcommand{\N}{\mathbb{N}}

\newcommand{\R}{\mathbb{R}}
\newcommand{\Z}{\mathbb{Z}}

\newcommand{\cD}{\mathcal{D}}
\newcommand{\evE}{\ensuremath{\mathscr{E}}}
\newcommand{\evF}{\mathscr{F}}
\newcommand{\evG}{\mathscr{G}}
\newcommand{\evH}{\mathscr{H}}
\newcommand{\evJ}{\mathscr{J}}

\newcommand{\cU}{\mathcal{U}}

\newcommand{\wt}{\widetilde}

\newcommand{\sfrac}[2]{#1/#2}
\newcommand{\st}{\mathrm{s.t.}}

\newcommand{\supp}{\mathrm{supp}}

\newcommand{\eps}{\varepsilon}
\renewcommand{\epsilon}{\varepsilon}

\newcommand{\argmax}{\operatornamewithlimits{argmax}}
\newcommand{\Ex}{\operatornamewithlimits{\mathbb{E}}}

\def\abs#1{\left| #1 \right|}
\def\sabs#1{| #1 |}
\newcommand{\given}{\;\middle|\;}
\newcommand{\sinparen}[1]{(#1)}

\newcommand{\sinsquare}[1]{[#1]}
\newcommand{\sinangle}[1]{\langle#1\rangle}
\newcommand{\inparen}[1]{\left(#1\right)}
\newcommand{\inbrace}[1]{\left\{#1\right\}}
\newcommand{\insquare}[1]{\left[#1\right]}
\newcommand{\inangle}[1]{\left\langle#1\right\rangle}
\newcommand{\floor}[1]{\left\lfloor#1\right\rfloor}

\newcommand{\norm}[1]{\ensuremath{\left\lVert #1 \right\rVert}}
\newcommand{\snorm}[1]{\ensuremath{\lVert #1 \rVert}}

\newcommand{\zo}{\{0,1\}}

\newcommand{\negsp}{\hspace{-0.5mm}}
\newcommand{\sexp}[1]{{\hbox{\tiny$($}}#1{\hbox{\tiny$)$}}}

\newcommand{\wh}[1]{\widehat{#1}}

\newcommand{\hw}{\widehat{w}}
\newcommand{\hx}{\widehat{x}}

\newcommand{\Eqref}[1]{Equation~\eqref{#1}}

\newcommand{\sx}{\ensuremath{x^\star}}
\newcommand{\tx}{\ensuremath{\wt{x}}}
\newcommand{\ratio}{\ensuremath{\mathscr{R}}}

\newcommand{\unif}{\cU}

\newcommand{\bb}{\ensuremath{\beta_0}}

\newcommand{\Stackrel}[2]{\stackrel{\mathmakebox[\widthof{\ensuremath{#2}}]{#1}}{#2}}

\newcommand{\prog}[1]{Program~\eqref{#1}}

\newif\ifconf
\conftrue

\newcommand{\leftmarginINTERNAL}{18pt}
\ifconf
\renewcommand{\leftmarginINTERNAL}{18pt}
\else
\renewcommand{\leftmarginINTERNAL}{\leftmargin}
\fi

\ifconf

\else

\fi

\title{\bf Selection in the Presence of Implicit Bias:\\ The Advantage of Intersectional Constraints}

\author{Anay Mehrotra \\ Yale University \and Bary S. R. Pradelski \\ CNRS \and Nisheeth K. Vishnoi \\ Yale University}

\begin{document}

\maketitle
\begin{abstract}
    In selection processes such as hiring, promotion, and college admissions, implicit bias toward socially-salient attributes such as race, gender, or sexual orientation of candidates is known to produce persistent inequality and reduce aggregate utility for the decision maker. Interventions such as the Rooney Rule and its generalizations, which require the decision maker to select at least a specified number of individuals from each affected group, have been proposed to mitigate the adverse effects of implicit bias in selection. Recent works have established that such lower-bound constraints can be very effective in improving aggregate utility in the case when each individual belongs to at most one affected group. However, in several settings, individuals may belong to multiple  affected groups and, consequently, face more extreme implicit bias due to this {\em intersectionality}. We consider independently drawn utilities and show that, in the intersectional case, the aforementioned non-intersectional constraints can only recover part of the total utility achievable in the absence of implicit bias. On the other hand, we show that if one includes appropriate lower-bound constraints on the intersections, almost all the utility achievable in the absence of implicit bias can be recovered. Thus, intersectional constraints can offer a significant advantage over a reductionist dimension-by-dimension non-intersectional approach to reducing inequality.
\end{abstract}

\newpage
  \setcounter{tocdepth}{2}
  \tableofcontents
  \addtocontents{toc}{\protect\setcounter{tocdepth}{2}}

\newpage

\section{Introduction}\label{sec:intro}

   \emph{Implicit bias} is the unconscious association, belief, attitude, skewed observation, {or lack of awareness} toward any socially-salient group, which may lead to systematic disadvantages for particular -- often underprivileged -- groups \cite{greenwald2006implicit,jolls2006law,kang2011implicit, Greenwald2020}.
  The negative impact of implicit bias on certain groups of the population is well documented in many societal contexts~\cite{munoz2016big,acm2017statement}, including hiring~\cite{rooth2010automatic, ziegert2005employment, corinne2012science,bendick2012developing}, university admissions~\cite{capers2017implicit, posselt2016inside}, and healthcare~\cite{chapman2007sterotyping,green2007implicit,jimenez2010perioperative}.
Instances of implicit bias in hiring and other selection processes include
  higher salaries for men than women despite the same qualifications \cite{corinne2012science},
  biased peer-review of fellowship applications against women~\cite{wenneras2001nepotism}, and
  stricter promotion standards for women in managerial positions~\cite{lyness2006fit}.
{Affected candidates can also face biases before participating in selection processes: For instance, they can face
  implicit bias in the form of lower teacher expectations~\cite{GershensonLowExpec2016}, or harsher grading policies~\cite{staats2016understanding}, {which may further hurt their future prospects in hiring or promotion~\cite{okonofua2015two}.}
  Implicit bias not only has adverse effects on individuals, but also on decision makers who may hire/promote less qualified candidates.
  Moreover, such biases can also affect downstream algorithms and policies, {either through biased human decisions or through past-data used to inform these decisions,} giving rise to  biases that affect different groups differently  \cite{amazonRecruitingTool,term_of_stay_correlated_with_race,bleemer2021college}.

    {Policy makers, private entities, and researchers have introduced a host of measures to counter  adverse effects of implicit bias:
    affirmative action policies which increase representation of affected groups~\cite{sowell2004affirmative,obama_rr,facebook_rr,bland2017schumer,passarielloWSJ},
    structured interviews which reduce the scope for bias in evaluation criteria~\cite{BarbaraAscription2000,gawande2010checklist,wenneras2001nepotism,baldiga2014gender},
    and  anonymized evaluations that blind decision makers to the socially-salient attributes of applicants~\cite{BlindAuditions2000}.
   Significant effort has also been devoted to reduce implicit bias itself:
    training that exposes individuals to counter-stereotypical evidence opposing their implicit beliefs~\cite{zestcott2016examining,starbucks_incident2018,FitzGeraldInterventions2019},
    enhanced accountability which enables enforcement of other interventions~\cite{KruglanskiAccountabilityBias1983,lerner1999accounting,botelho2021disciplining}, and
    information campaigns that increase awareness about implicit biases~\cite{mcgregor2017race}.}

    {Of interest here are affirmative action policies that introduce lower-bound constraints for groups adversely affected by implicit bias.}
  {A popular instantiation of this strategy is the Rooney rule,} which requires the decision maker to select at least one individual from the affected group for interview.
  The hope is that during the interview, the decision maker will assess the ``true'' value of the individual {\cite{TristanQuality2017}} and this interaction will reduce their implicit bias {\cite{DasguptaOutGroupExposure2008}.}
    {In addition,} variants of the Rooney Rule have also been used for the final stage of a selection process such as that for board membership or highly-priced entry jobs to {directly} counter the effects of implicit bias \cite{Hol00,cavicchia-implicit-bias-rooney,passarielloWSJ,facebook_rr}.

    Recently, {some works have analyzed} the effectiveness of Rooney Rule type constraints for selection and ranking processes~\cite{KleinbergR18,celis2020interventions,EmelianovGGL20}.
    In particular,  \cite{KleinbergR18} study the effectiveness of the Rooney Rule for selection in the presence of implicit bias when there is a single affected group.
    Here, {there are $m$ {\em individuals},} where each {individual} $i\in \inbrace{1,2,\dots,m}$ has a {non-negative} {\em latent utility} $w_i\geq 0$, that is the {\em value} it adds to the selection, and an {\em observed utility} $\hw_i\leq w_i$, that is the decision maker's possibly biased estimate of $w_i$.
    The decision maker selects $n$ {individuals} with the maximum sum of observed utilities.
    \cite{KleinbergR18} study a  model where implicit bias acts via a multiplicative factor $0<\beta\leq 1$:  the observed utility of an individual $i$ belonging to the affected group is $\hw_i=\beta\cdot w_i$ while that of an unaffected individual $j$ is $\hw_j=w_j$.
    {They argue that this model is a reasonable approximation of the empirical findings of \cite{wenneras2001nepotism}, who find that in peer-reviewed evaluations for fellowships women's score were systematically scaled down compared to men with similar productivity.}
    \cite{KleinbergR18} study conditions on the parameters $n$, $m$, $\beta$, and the distribution of latent utility $w_i$, where the Rooney Rule increases the total latent utility of the selection.
    Under the same implicit bias model, \cite{celis2020interventions} study a generalization of the Rooney Rule, where the decision maker is constrained to select at least $L\geq 1$ individuals from the affected group.
    \cite{celis2020interventions} show that, for a single affected group, there is an $L$ for which the decision maker, constrained by the generalized Rooney Rule, achieves near-optimal latent utility.

    \emph{Intersectionality} posits that one needs to take into account the interconnected nature of the multiple socially-salient attributes, as opposed to viewing these attributes through a reductionist lens, that is, dimension-by-dimension \cite{Cre89}.
    Intersectionality can lead to the creation of overlapping and interdependent systems of discrimination or disadvantage, and there is a rich literature in social sciences and law that studies it  \cite{king1988multiple,Cre89,Col00,brewer2002complexities,collins2004black,Sen06,Pur08,elu2013earnings,williams2014double,cooper2016intersectionality,collins2016intersectionality,akerlof2017value,Car19}.
    Intersectional implicit bias also arises in selection processes.
    {For instance, \cite{wenneras2001nepotism} find significantly lower scores for women unaffiliated with the evaluation committee compared to other women and to men unaffiliated with the committee in peer-reviewed applications for a fellowship in Sweden.}
    {Thus, neither gender nor affiliation alone explain the bias faced by individuals, and to understand this bias, a combination of the two attributes must be considered.}
    Recently, intersectional bias {has also been observed in the outputs of algorithms.}
    For instance, \cite{BuolamwiniG18} audit commercial image-based gender classifiers and find intersectional bias against Black women, and \cite{TanC19} report intersectional bias in contextualized word representations.

    However, the intersectional nature of social groupings and, thus, biases has been largely overlooked when designing interventions to reduce or counter implicit bias.
    Further, data reporting, such as that by the U.S. Census Bureau, is mostly dimension-by-dimension (for example, by race or by gender) and omits intersectional data.
    {Since this data is used to inform policies, it has} inevitably led to policies that only focus on reducing inequality along one identity dimension at a time, as highlighted in a report by the European Union \cite{skjeie2015gender}.
  {In other words,} existing data reporting and, consequently, policies are {\em non-intersectional}.
  {They} specify lower-bound constraints on each affected group, but not on their intersections.
  But it is natural to expect that the individuals at the intersection of multiple affected groups face higher {and, possibly, different} implicit bias~\cite{king1988multiple,wenneras2001nepotism,BrowneIntersection2003,williams2014double}.
    {In fact, as argued by \cite{king1988multiple}, intersectional bias can be significantly higher and often compound, or multiply, the} {{biases faced by individuals in single affected groups.}
  Thus, the following question arises and is studied in this paper.}

  \medskip
    
    \begin{tcolorbox}[bottom=0.01cm,top=0.01cm,left=0.05cm,right=0.05cm]
      {\em Are non-intersectional constraints sufficient to recover the entire latent utility with intersecting affected groups or does one need to specify constraints across all intersections to achieve this?}
     \end{tcolorbox}

\subsection{Our Contributions}
  We consider the effectiveness of lower-bound constraints on selection processes in the presence of intersectional implicit bias.
 To capture the effect of intersectionality on implicit bias, we consider an extension of the aforementioned model of \cite{KleinbergR18} due to \cite{celis2020interventions}:   each individual may belong to zero, one, or more of the $p$ affected groups (such as the groups of all women or all Black people) (\cref{sec:implicit_bias}).
  For each group $\ell\in \inbrace{1,2,\dots,p}$, there is an implicit bias parameter $0<\beta_\ell\leq 1$ and
  the implicit bias experienced by an individual is the product of  parameters of each group they belong to.

  We {compare} non-intersectional and intersectional lower-bound constraints when the latent utilities are independently {and identically} distributed.
  Non-intersectional constraints specify the minimum number of individuals to be selected from each affected group (e.g., the groups of all women or all Black people).
  They do not specify the minimum number of individuals to be selected from a given intersection {(e.g., the groups of all Black women, all non-Black women, all Black non-women, or all non-Black non-women)}.
  For each of the intersectional groups,  intersectional constraints specify the minimum number of individuals to be selected from this intersection.
  To compare the relative efficacy of constraints, we consider a {\em utility ratio}, defined as
  the expected value of the ratio of the latent utility achieved under the {constraint to the latent utility absent any implicit bias (\cref{sec:utility_ratio}).}
  By definition, the utility ratio is a number between $0$ and $1$  and the goal of the policy maker is design interventions such that the corresponding utility ratio is $1$.

  {We show that under general conditions on the  distribution of latent utilities, no matter which {\em non-intersectional constraints} are deployed, the utility ratio is strictly less than $1$ (\cref{thm:ub,thm:gen_ub}).}
  In particular, our result applies to distributions such as uniform, truncated normal, and truncated power-law distributions.
 Moreover, our result gives a quantitative bound on the maximum of utility ratio: it is at most $1-\phi$, where $\phi$ is positive and depends only on the implicit bias parameters and generic parameters of the distribution family, and independent of the {number of candidates} $m$.
  Concretely,  {when the utilities are uniformly distributed on the interval $[0,1]$}, there is a family of instances such that {the maximum utility ratio achievable using only non-intersectional constraints can be as low as $8/9$} (\cref{prop:89}).
   {Further, we show that this result also holds for
    generalizations of the implicit bias model where, for instance, the implicit bias experienced by individuals in multiple groups is different than the product of the implicit bias parameters of the groups they belong to (\cref{thm:gen_ub}).}
  Thus, these results imply that, unlike the setting of a single affected group studied in \cite{KleinbergR18,celis2020interventions,EmelianovGGL20}, non-intersectional constraints may be {insufficient to completely mitigate the effects of implicit bias in the presence of intersections.}}

  {On the positive side,} we show that there are {\em intersectional} lower-bound constraints that, for any amount of implicit bias, recover utility ratio arbitrarily close to 1 (\cref{thm:fullrecovery}).
  This result extends for all continuous distributions of utility and for the generalizations of the implicit bias model considered above (\cref{coro:fullrecovery}).
  {We show that these intersectional lower-bound constraints end up being just a function of the sizes of intersections  and do not depend on the amount of implicit bias $0<\beta_1,\ldots,\beta_p\leq 1$ or the specific utility distribution.
  In fact, the constraints require at least a near-proportional number of individuals from each intersection and,
 hence, they can be employed in practice where the implicit bias parameters $\beta_1,\ldots,\beta_p$ and the distribution of utility are not known and can vary across contexts or over time.}
  Thus, a policy maker may choose {intersectional constraints in order to obtain a utility ratio arbitrarily close to 1.}

      Overall, our results imply that the advantage of intersectional constraints can be substantial and a reductionist dimension-by-dimension approach is not sufficient to mitigate the adverse effects of implicit bias.
      They provide a utilitarian reason for policy makers to choose intersectional constraints over non-intersectional constraints.

    \subsection{Related Work}

    \paragraph{Implicit bias and empirics.}
    {There are several theories about how implicit bias arises; e.g., \cite{allport1954,tversky1974judgment,haselton_buss_2009,tamar2011epistemic,greenwald1995implicit,McC1981,Pay19,agarwal2020sway}.
    Specific examples include, \cite{tversky1974judgment} who propose that humans unconsciously use heuristics to overcome their limited computing ability.
    Such heuristics can take the form of stereotypes where one divides individuals into groups and, then, extrapolates the characteristics of specific individuals from the characteristics associated with their group(s)~\cite{tamar2011epistemic,agarwal2020sway}.
    Another theory suggests that using stereotypes was evolutionarily advantageous~\cite{haselton_buss_2009,kurzban2001evolutionary}.
    One reason, according to \cite{kurzban2001evolutionary}, is that undervaluing the utility of unknown out-group individuals and, hence, avoiding contact with them reduced the risk of contracting new diseases.
    Apart from these theories, it has also been suggested that implicit bias ``is a trace of [the individual's] past experience'' \cite{greenwald1995implicit} and that prior, explicit, racism has been channeled into implicit bias \cite{McC1981}.}
  Regardless of the cause of implicit bias, studies identifying the adverse effects of implicit bias are abundant: from police shootings~\cite{sadler2012world},
  {promotion and hiring decisions~\cite{bertrand2004emily,lyness2006fit}, education~\cite{corinne2012science,van2010implicit}, to peer-review \cite{wenneras2001nepotism}.}

    \paragraph{{Models of implicit bias and decision-making in the presence of implicit bias}.}
    A growing literature is studying decision making in the presence of implicit bias, ranging from works on set selection \cite{KleinbergR18,EmelianovGGL20,Faenza2020fair},  ranking \cite{celis2020interventions}, to  classification \cite{blum2020recovering}.
    Among these, works on the set selection and ranking problems are directly related to our work.
    \cite{KleinbergR18} introduce a mathematical model of implicit bias for a single affected group and study when the Rooney Rule increases the total latent utility of the selection.
    Unlike them, we consider multiple and intersectional affected groups and also consider generalizations of the Rooney Rule.
    \cite{celis2020interventions} study the ranking problem, where the selected individuals also need to be ordered.
    Specializing their work to set selection: they consider the setting with a single affected group where the decision maker must select at least $L\geq 1$ individuals from the affected group and show that there are constraints which achieve near-optimal latent utility in expectation.
    \cite{celis2020interventions} also extended the model of implicit bias due to \cite{KleinbergR18} to multiple and intersectional groups.
    For this model, \cite{celis2020interventions} show that for any set of utilities and amount of bias, there are utility-dependent non-intersectional constraints that achieve optimal latent utility for the ranking problem.
    However, since their constraints are a function of the latent utilities, which are not observed, these constraints cannot be determined in practice.
    While we also consider the model of implicit bias \cite{celis2020interventions} introduced, the intersectional constraints we propose are different and do not depend on the, unknown, latent utilities.
    \cite{EmelianovGGL20} study selection under a different model of bias, where the decision maker's observed utility has higher than average noise for individuals in the affected group.
    They consider a family of constraints and show that, for a single affected group, these constraints increase the latent utility.
    Unlike them, we consider multiple and intersectional groups and study a different model of bias.
    Finally, unlike these prior works, we also study the maximum utility achievable by using  non-intersectional constraints.

    \paragraph{Intersectionality as a source of bias.}
  The discussion of being subjected to multiple biases has originally focused on the experience of Black women versus that of non-Black women and Black men.
  The ``double jeopardy'' {and ``multiple jeopardy'' hypotheses posit} that belonging to more than one affected group -- as is the case for {Black} women -- disproportionately increases the experienced bias  \cite{Eps73,king1988multiple}.
  Empirically, for example, \cite{elu2013earnings} show that returns to schooling in sub-Saharan Africa depend both gender and ethnicity {and \cite{wenneras2001nepotism} find that peer-reviewed scores for post-doctoral fellowships were a function of both the candidate's gender and their affiliation with reviewers.}
  {Since \cite{Eps73,king1988multiple}, several works have proposed extensions of this theory beyond two socially-salient attributes;  e.g.,  \cite{Cre89,Pur08,Col00,brewer2002complexities,cooper2016intersectionality,collins2016intersectionality}.}
    The  implicit bias models that we consider in this work can be viewed as motivated \mbox{by these works, and in particular, by the multiple jeopardy model \cite{king1988multiple}.}

    \paragraph{Intersectionality vs. non-intersectionality.}
  To the best of our knowledge, there are only few examples of mathematical studies that analyze how the belonging to intersectional groups interacts with policies.
    \cite{akerlof2017value} study how individuals suppress or foster different dimensions of their identity to increase economic reward.
    \cite{Car19} study the effect of intersectional vs. non-intersectional interventions  on the share of different groups among the selected individuals over time.
    Our focus here is  understanding the advantages of {intersectional constraints} over {dimension-by-dimension} non-intersectional constraints  -- albeit in the very different setup of selection under implicit bias.

\section{Model}\label{sec:model}
    {\em Notation.}
    For a number $n\in \N$, $[n]$ denotes the set $\{1,2,\dots,n\}$.
    $\unif$ denotes the uniform distribution over $[0, 1]$.
    {We use $w \sim \mathcal{D}$ notation to denote that $w$ is an independent sample from distribution $\mathcal{D}$.}
    For a distribution $\cD$ over $\R$, we use $\mu_\cD\colon \R\to \R_{\geq 0}$ to denote its probability density function and  $F_\cD\colon \R\to [0,1]$ to denote its cumulative distribution function.
    {We say a distribution $\cD$ over $\R$ is continuous if $\mu_\cD$ exists and is finite at all points in $\R$.}
    {The support of a continuous distribution $\cD$ over $\R$ is the set $\inbrace{x\in \R\colon \mu_\cD(x)>0}$, and is denoted by $\supp(\cD)$.}
    Given two vectors $x,y\in \R^m$, we use $\inangle{x,y}$  to denote their inner product $\sum_{i=1}^m x_i y_i$.

    \subsection{Selection Problem}
      {The task of selecting a subset of individuals from a pool of applicants or employees arises in many contexts such as hiring, college admission, and selection for fellowships or board of directors.
      In these settings, the basic mathematical problem is as follows:}
      {Given a number $m$ and for each of the $m$ individuals, or more generally {\em items}, $i\in [m]$ a non-negative {\em latent utility} $w_i\geq 0$,}
      the set selection problem asks to find a subset of $n$ items that has the maximum sum of latent utilities.
      If we represent a subset by a binary vector $x\in \zo^m$, {where $x_i=1$ indicates that $i$ is in the subset and $x_i=0$ indicates otherwise,}
      {the goal is to find $x \in \{0,1\}^m$ such that $\inangle{x,w}$ is maximized subject to $\sum_{i=1}^m x_i=n$.}

    \subsection{Affected Groups, Intersections, and a Model of Implicit Bias}\label{sec:implicit_bias}
      We consider the setting with $p$ affected groups (henceforth referred to as just groups)  $G_1,G_2,\dots,G_p\subseteq [m]$.
      {Each of the $m$ items may belong to one or more of the $p$ groups that face implicit bias, or may belong to none of these groups, i.e., in $[m]\setminus{} \inparen{G_1\cup G_2\cup\dots\cup G_p},$ and hence, not face any implicit bias.}
      These groups can intersect arbitrarily.
      We use the following notation to capture all the intersections that can arise from $p$ groups:
      For a set $S\negsp \subseteq\negsp  [p]$ of groups, let
      $\sigma \in\zo^p$ denote the corresponding indicator vector, i.e., for all $\ell\in [p]$, $\sigma_\ell=1$ if $\ell \in S$ and $\sigma_\ell=0$ otherwise.
      Let $I_\sigma\negsp \subseteq\negsp  [m]$
      denote the set of elements that belong
      to every group in $S$ and none of the groups not in $S$.
      {Formally, when $S \neq \emptyset$ and, hence, $\sigma$ is not the all $0$s vector (denoted by $0$), let}
      \begin{align*}
             I_\sigma\coloneqq { \inparen{\bigcap\nolimits_{\ell:\sigma_\ell=1}G_\ell}} \setminus \inparen{\bigcup\nolimits_{\ell:\sigma_\ell=0} G_\ell}.
     \end{align*}
      For $\sigma =0$, let {$I_0 \coloneqq [m] \setminus \inparen{G_1\cup G_2\cup\dots\cup G_p}$ denote the items in none of the $p$ groups.}
      {Where with some abuse of notation we used $I_0$ to denote $I_{00}$ for $p=2$, $I_{000}$ for $p=3$, and so on.}
      Thus, the sets $\inbrace{I_\sigma}_{\sigma\in\zo^p}$ partition the set of items $[m]$.
      {\cref{fig:groups:a} illustrates this with two groups $G_1$ and $G_2$ (i.e., $p=2$) which divide the set of items into four disjoint intersections $I_{11}$, $I_{10}$, $I_{01}$, and $I_{00}$, where $I_{11}=G_1\cap G_2$, $I_{10}=G_1\setminus{} G_2$, $I_{01}=G_2\setminus{} G_1$, and $I_{00}=[m]\setminus{}\inparen{G_1 \cup G_2}$.}

      {We focus on the extension of the implicit bias model of \cite{KleinbergR18} presented in \cite{celis2020interventions}.
      Later, in \cref{sec:results:extensions}, we also consider further generalizations of this model.}
      In this model, the decision maker does not observe the latent utilities  of the items.
      Instead, for each item $i$, they see an observed utility $\hw_i$, which is their possibly biased estimate of $w_i$.
      In particular, the decision maker might perceive that items belonging to certain groups have a lower observed utility: $\wh{w}_i<w_i$.
      For each {group} $\ell\in [p]$, there is an implicit bias parameter $0<\beta_\ell\leq 1$ that captures the relative implicit bias faced by items in $G_\ell$ compared to items not in $G_\ell$.
      The total implicit bias experienced by an item $i$ is assumed to be the product of the implicit bias parameters of all groups it belongs to: ${\prod\nolimits_{\ell\in [p]\colon G_\ell\ni i}\beta_\ell}$.
      Thus, for a given latent utility $w_i$ of item $i$, the observed utility is
       \begin{align}
        \hw_i\coloneqq \inparen{\prod\nolimits_{\ell\in [p]\colon G_\ell\ni i}\beta_\ell}\cdot w_{i}.\tagnum{Implicit bias}\customlabel{eq:mult_bias}{\theequation}
      \end{align}
      \noindent
        The property that individuals belonging to multiple groups are subject to {more acute} implicit bias is motivated by similar observations in the real world, which have been reproduced across many contexts such as hiring in industry, promotions in industry and academia, and peer-review in academia~\cite{wenneras2001nepotism,RosetteFailure2012,deo2017intersectional,derous2019gender}
      {It also aligns with \cite{king1988multiple}, which proposes a multiplicative-model where individuals face the compounded effect of the biases of groups they belong to.} %

      {As an illustration of this intersectional implicit bias model, consider two groups where the implicit bias parameter of the first group is $\beta_1$ and of the second group is $\beta_2$.
      An item which belongs to both $G_1$ and $G_2$ (i.e., in $I_{11}$) experiences an implicit bias $\beta_1\beta_2$.
      Whereas items in $G_1$ and not $G_2$ (i.e., in $I_{10}$) experience an implicit bias $\beta_1$ and items in $G_2$ but not $G_1$ (i.e., in $I_{01}$) experience an implicit bias $\beta_2$.
      Items neither $G_1$ nor $G_2$ (i.e., in $I_{00}$) do not face implicit bias.
      As a numerical example, if $\beta_1=0.9$ and $\beta_2=0.8$, then an item $i\in I_{11}$ experiences an implicit bias $0.72$, which is more acute than the implicit bias experienced by items in $I_{10}$ or in $I_{10}$, which experience biases $0.9$ and $0.8$ respectively.}

      {In \cref{sec:results}, we consider a generalization of this model, where for each intersection $I_\sigma$, there is an increasing function $b_\sigma$, and the observed utility of an item $i$ in intersection $I_\sigma$ is $\hw_i \negsp =\negsp b_\sigma(w_i)$.
      This generalization captures Equation~\eqref{eq:mult_bias} when {$b_\sigma(x)\negsp \coloneqq\negsp x \cdot\negsp  {\prod\nolimits_{\ell\in [p]\colon \sigma_\ell = 1}\beta_\ell}$ for all intersections $\sigma$ and $x\negsp\geq\negsp 0$.}}
      
      \begin{figure*}[b!]
        \centering
        \vspace{-2mm}
        \hspace{-16mm}
        \subfigure[\scriptsize Intersections ($p=2$) \label{fig:groups:a}]{
        \resizebox{136.27544pt}{75.24838999999999pt}
        {\begin{tikzpicture}
          \tikzmath{\t = 0.9;}
          \tikzmath{\s = 0.85;}
          \tikzmath{\mvx = 1.25;}
          \tikzmath{\mvy = 0.75;}
          \tikzmath{\mvxx = 0.6;}
          \tikzmath{\side = 3;}

          \node[] at (-3.3,-1) {\white{$G$}};
          \node[] at (0,-2.25) {\white{$G$}};

          \draw[color=cyan!0,fill=cyan!0] (\mvx+1,1-\mvy) rectangle ++(0.3,0.3);
          \node[] at (\mvx+1.75,1.15-\mvy) {\white{$G_1$}};

          \draw[color=black,thick,fill=none] (-\side/2-0.3,-\side/3*2-0.6) rectangle ++(\side*1.2,0.2+\side/6*5+1.05);

          \draw[rotate=50,color=cyan!60,fill=cyan,very thick,fill opacity=0.1] (-\s,0) ellipse (\t*1.618cm and \t*1cm);

          \draw[rotate=-50,color=red!60, fill=red, very thick,fill opacity=0.1] (+\s,0) ellipse (\t*1.618cm and \t*1cm);

          \node at (1.1,-1.4) {\Large $I_{01}$};
          \node at (-1.1,-1.4) {\Large $I_{10}$};
          \node at (0,0) {\Large $I_{11}$};

          \node at (14.0:1.65*\s) {\Large $I_{00}$};

        \end{tikzpicture}
        }
        }
        \hspace{-10mm}
        \subfigure[\scriptsize Implicit bias parameters ($p=2$) \label{fig:groups:b}]{
        \resizebox{136.27544pt}{75.24838999999999pt}
        {\begin{tikzpicture}
          \tikzmath{\t = 0.9;}
          \tikzmath{\s = 0.85;}
          \tikzmath{\mvx = 1.25;}
          \tikzmath{\mvy = 0.75;}
          \tikzmath{\mvxx = 0.6;}
          \tikzmath{\side = 3;}

          \node[] at (-3.3,1) {\white{$|$}};
          \node[] at (0,-2.25) {\white{$G$}};

          \draw[color=cyan!60,fill=cyan!60] (\mvx+1,1-\mvy) rectangle ++(0.3,0.3);
          \node[] at (\mvx+1.75,1.15-\mvy) {$G_1$};

          \draw[color=red!60,fill=red!60] (\mvx+1,1-0.5-\mvy) rectangle ++(0.3,0.3);
          \node[] at (\mvx+1.75,1.15-0.5-\mvy) {$G_2$};

          \draw[color=black,very thick,fill=none] (\mvx+1+0.01,1-1-\mvy) rectangle ++(0.28,0.28);
          \node[] at (\mvx+1.75,1.15-1-\mvy) {$[m]$};

          \draw[color=black,thick,fill=none] (-\side/2-0.3,-\side/3*2-0.6) rectangle ++(\side*1.2,0.2+\side/6*5+1.05);

          \draw[rotate=50,color=cyan!60,fill=cyan,very thick,fill opacity=0.1] (-\s,0) ellipse (\t*1.618cm and \t*1cm);

          \draw[rotate=-50,color=red!60, fill=red, very thick,fill opacity=0.1] (+\s,0) ellipse (\t*1.618cm and \t*1cm);

          \node at (1.1,-1.4) {\Large $\beta_2$};
          \node at (-1.1,-1.4) {\Large $\beta_1$};
          \node at (0,0) {\Large $\beta_1\beta_2$};

          \node at (14.0:1.65*\s) {\Large $1$};

        \end{tikzpicture}
        }
        }
        \hspace{-7mm}
        \subfigure[\scriptsize Intersections ($p=3$)\label{fig:groups:c}]{
        \resizebox{123.19498905pt}{74.5993689pt}
        {\begin{tikzpicture} %
          \tikzmath{\r = 1.2;}
          \tikzmath{\s = 0.75;}
          \tikzmath{\mvx = 1.25;}
          \tikzmath{\mvy = -0.5;}
          \tikzmath{\mvxx = 1;}
          \tikzmath{\side = 3.75;}
          \node[] at (-2.9,0) {\white{$G$}};
          \draw[color=cyan!0,fill=cyan!0] (\mvx+1,1-\mvy) rectangle ++(0.3,0.3);
          \node[] at (\mvx+1.75,1.15-\mvy) {\white{$G_1$}};
          \draw[color=black,thick,fill=none] (-\side/2-0.125,-\side/2+0.175) rectangle ++(\side+0.25,\side);
          \draw[rotate=0,color=black!60,fill=black,very thick,fill opacity=0.1] (0,+\s) ellipse (\r cm and \r cm);
          \draw[rotate=0,color=cyan!60,fill=cyan,very thick,fill opacity=0.1] (-\s*1.732/2,-\s/2) ellipse (\r cm and \r cm);
          \draw[rotate=0,color=red!60, fill=red, very thick,fill opacity=0.1] (+\s*1.732/2,-\s/2) ellipse (\r cm and \r cm);
          \node at (30:\r*0.71) {\Large $I_{011}$};
          \node at (150:\r*0.71) {\Large $I_{101}$};
          \node at (-90:\r*0.71) {\Large $I_{110}$};
          \node at (45:1.75*\r) {\Large $I_{000}$};
          \node at (-150:\r) {\Large $I_{100}$};
          \node at (-30:\r) {\Large $I_{010}$};
          \node at (90:\r) {\Large $I_{001}$};
          \node at (0,0) {\Large $I_{111}$};
        \end{tikzpicture}
        }
        }
        \hspace{-7mm}
        \subfigure[\scriptsize Implicit bias parameters ($p=3$)\label{fig:groups:d}]{
        \resizebox{123.19498905pt}{74.5993689pt}
        {\begin{tikzpicture} %
          \tikzmath{\r = 1.2;}
          \tikzmath{\s = 0.75;}
          \tikzmath{\mvx = 1.25;}
          \tikzmath{\mvy = -0.5;}
          \tikzmath{\mvxx = 1;}
          \tikzmath{\side = 3.75;}
          \node[] at (-2.9,-1) {\white{$|$}};
          \draw[color=cyan!60,fill=cyan!60] (\mvx+1,1-\mvy) rectangle ++(0.3,0.3);
          \node[] at (\mvx+1.75,1.15-\mvy) {$G_1$};
          \draw[color=red!60,fill=red!60] (\mvx+1,1-0.5-\mvy) rectangle ++(0.3,0.3);
          \node[] at (\mvx+1.75,1.15-0.5-\mvy) {$G_2$};
          \draw[color=black!60,fill=black!60] (\mvx+1,1-1-\mvy) rectangle ++(0.3,0.3);
          \node[] at (\mvx+1.75,1.15-1-\mvy) {$G_3$};
          \draw[color=black,very thick,fill=none] (\mvx+1,1-1.5-\mvy) rectangle ++(0.3,0.3);
          \node[] at (\mvx+1.75,1.15-1.5-\mvy) {$[m]$};
          \draw[color=black,thick,fill=none] (-\side/2-0.125,-\side/2+0.175) rectangle ++(\side+0.25,\side);
          \draw[rotate=0,color=black!60,fill=black,very thick,fill opacity=0.1] (0,+\s) ellipse (\r cm and \r cm);
          \draw[rotate=0,color=cyan!60,fill=cyan,very thick,fill opacity=0.1] (-\s*1.732/2,-\s/2) ellipse (\r cm and \r cm);
          \draw[rotate=0,color=red!60, fill=red, very thick,fill opacity=0.1] (+\s*1.732/2,-\s/2) ellipse (\r cm and \r cm);

          \node at (30:\r*0.71) {$\beta_2\beta_3$};
          \node at (150:\r*0.71) {$\beta_1\beta_3$};
          \node at (-90:\r*0.71) {$\beta_1\beta_2$};
          \node at (45:1.75*\r) {$1$};
          \node at (-150:\r) {$\beta_1$};
          \node at (-30:\r) {$\beta_2$};
          \node at (90:\r) {$\beta_3$};
          \node at (0,0) {$\beta_1\beta_2\beta_3$};
        \end{tikzpicture}
        }
        }
        \hspace{-7mm}
        \vspace{-2mm}
        \caption{\small {\em Model of affected groups and implicit bias:}
        The first two sub-figures consider two overlapping affected groups.
        \cref{fig:groups:a} presents the four intersections created by these groups and \cref{fig:groups:b} gives the implicit bias experienced by items in each intersection.
        Here, items in both $G_1$ and $G_2$, i.e., those in $I_{11}$, face the most acute implicit bias, $\beta_1\beta_2$: this is the compounded effect of the implicit bias faced by each items each group $\beta_1<1$ and $\beta_2<1$. %
        The last two sub-figures this illustration to three overlapping groups.
        }
        \label{fig:groups}
    \end{figure*}
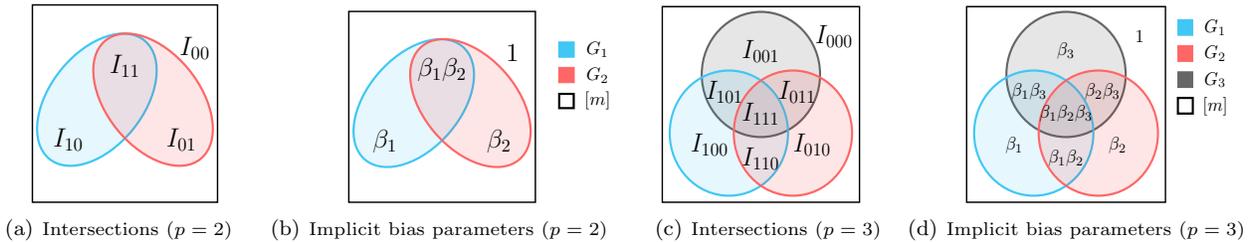

    \subsection{Selection Under Implicit Bias: Intersectional and Non-intersectional Constraints}\label{sec:selection_with_bias}
      {The decision maker would like to pick an $\sx{} \in \{0,1\}^m$ with $n$ items, i.e., $\sum_{i=1}^m \sx_i=n$, that maximizes latent utility:} %
        \begin{align*}
            \sx \coloneqq \argmax\nolimits_{x\in \zo^m} \inangle{x,w},\quad \text{s.t.,}\quad \sum_{i=1}^m x_i=n.\yesnum\label{eq:def_xstar}
        \end{align*}
      {However, due to their implicit bias, the decision maker instead maximizes the observed utility:
      They choose a selection $\hx$ with $n$ items that maximizes the observed utility,}
      \begin{align*}
            \hx\coloneqq \argmax\nolimits_{x\in \zo^m} \inangle{x,\hw},\quad \text{s.t.,}\quad \sum_{i=1}^m x_i=n.\yesnum
        \end{align*}
      Since $\hx$ maximizes a different objective than $\sx$, it could be very different from $\sx$.
      Hence, it may have a much smaller latent utility: $\inangle{\hx,w} \ll \sinangle{\sx,w}$.
      {To see this, consider two items where the latent utility of the first item is $w_1=1$ and of the second item is $w_2=0.1$.
      Suppose the decision maker has to select one item (i.e., $n=1$) and there are two groups. %
      Since $\sx$ maximizes the latent utility, it selects the first item. %
      However, if the first item is in $I_{11}$ and the second item is in $I_{00}$, then their observed utilities are $\hw_1=\beta_1\beta_2$ and $\hw_2=0.1$ respectively.
      If $\beta_1\beta_2<0.1$, then $\hx$ selects the second item.
      Hence, it has a latent utility $\sinangle{\hx,w}=0.1$, which is {significantly smaller than the latent utility of $\sx$: $\sinangle{\sx,w}=1$.}}
      Apart from this, $\hx$ also selected fewer items facing implicit bias, i.e., items $i$ for which $\hw_i < w_i$; adversely affecting such items.

    {Affirmative action policies broadly seek to reduce and counter implicit bias and other systematic biases.}
    {There are several types of affirmative action policies~\cite{sowell2004affirmative,baswana2019centralized}.}
    Here, we consider lower-bound constraints, such as the Rooney Rule, which require the decision maker to select at least a specified number of candidates from each group.
        Recent works \cite{KleinbergR18,celis2020interventions,EmelianovGGL20} on decision-making in the presence of implicit bias have shown that in the setting with a single group, such constraints can serve an additional purpose: to improve the latent utility of the subset selected by the decision maker.
      Motivated by this, we study the efficacy of these constraints to improve the latent utility in {the} setting with multiple and intersectional groups, as discussed in \cref{sec:implicit_bias}.
      Specifically, we consider two types of lower-bound constraints: non-intersectional and intersectional:
        \begin{itemize}[leftmargin=\leftmarginINTERNAL{}]
          \item {\em Non-intersectional constraints} are specified by lower bounds $L_1,L_2,\dots,L_p\in \Z_{\geq 0}$, and require that, for each $\ell\in [p]$, the decision maker include at least $L_\ell$ items from group  $G_\ell$.
          \item {\em Intersectional constraints} are specified by $2^p$ lower bounds, $L_\sigma \in \Z_{\geq 0}$ for each $\sigma\in \zo^p$, and require that, for each $\sigma\in \zo^p$, the decision maker include at least $L_\sigma$ items from intersection $I_\sigma$.
        \end{itemize}
        \noindent When the context is clear, we  use $L$ to denote the vector of lower bounds in either intersectional or non-intersectional constraints.
        Given a vector $L$ defining a lower-bound constraint, either non-intersectional or intersectional, let $C(L)\subseteq \zo^m$ be the set of all subsets {with $n$ items} that satisfy the constraints defined by $L$. %
        Given a constraint $L$, the {\em constrained} decision maker chooses the selection $\tx$ with the highest observed utility in $C(L)$:
      \begin{equation}\label{eq:css}
        \tx \coloneqq \argmax\nolimits_{x\in C(L)} \inangle{x,\hw}.
      \end{equation}

      \noindent Our goal is to understand how close to the optimal latent utility, $\sinangle{\sx,w}$, can {the decision maker (who selects) get when a policy maker (who decides the lower-bound constraints) imposes intersectional vs. non-intersectional constraints on them.}
      In particular, given groups $G_1,\ldots, G_p$ and some unknown implicit bias parameters $\beta_1,\ldots,\beta_p$, which lower bound vectors $L$ have the property that the latent utility of the selection $\tx$, $\langle \tx,w \rangle$, is close to $\sinangle{\sx,w}.$

    \subsection{Utility Ratio}\label{sec:utility_ratio}
        As in \cite{KleinbergR18,celis2020interventions,EmelianovGGL20}, we study the setting where latent utilities of all items are drawn independently from some continuous distribution $\cD$ with a non-negative support.
        {We begin with the uniform distribution and later consider other distributions in \cref{sec:results:extensions}.}
      To measure the relative efficacy of different constraints, we consider a {\em utility ratio} defined as the expected value of the ratio of $\inangle{\tx, w}$ to $\sinangle{\sx, w}$. 
      {Here, $\sx$ is the selection with the highest latent utility (as defined in Equation~\eqref{eq:def_xstar}) and $\tx$ is the selection picked by the constrained decision maker
      (as defined in Equation~\eqref{eq:css}).
      For a draw of latent utilities $w$, $\frac{ \sinangle{\tx,w} }{ \sinangle{\sx,w} }$ is the fraction of the optimal latent utility obtained by the constrained decision maker.
      Hence, the utility ratio measures the expected fraction of the optimal utility obtained by the constrained decision maker.}
      Formally, given lower bounds $L$ and implicit bias parameters $\beta$, the utility ratio is:
      \begin{align*}
        \ratio_{\cD}(L,\beta)\coloneqq \Ex\nolimits_{w\sim \cD}\insquare{ \frac{ \sinangle{\tx,w} }{ \sinangle{\sx,w} }}.
        \tagnum{Utility ratio}\customlabel{eq:def_ratio}{\theequation}
      \end{align*}
      {Where $\sx$ is a function of the utilities $w$ and $\tx$ is a function of utilities $w$, implicit bias parameters $\beta$, and lower bounds $L$.}
    It can be shown that the utility ratio is invariant to scaling {of utilities} (\cref{claim:invariance_of_r}), and its range is invariant across all distributions with the same mean and support (\cref{claim:range_of_r}).
    In particular, for the uniform distribution on $[0,C]$, for any $C>0$, the range is $\frac{1}{2}$ to 1 (i.e., $\smash{\frac{1}{2}} \leq  \ratio_\cU(L,\beta) \leq 1$).
    {{When discussing the uniform distribution on $[0,1]$ we drop the subscript of $\ratio$ and use $\ratio(L,\beta)$ to denote $\ratio_\cU(L,\beta)$.}}

\section{Results}\label{sec:results}

  \subsection{Sub-Optimal Latent Utility for Any Non-intersectional Constraints}\label{sec:upper_bound}
        {Our first result establishes} an upper bound on the maximum utility ratio (Equation~\eqref{eq:def_ratio}) that {a policy maker can secure} in the presence of implicit bias by using non-intersectional constraints (\cref{sec:selection_with_bias}).
        For simplicity, we first consider the case where utilities are distributed according to the uniform distribution on $[0,1]$ and, later, consider generalizations to other distributions in \cref{sec:results:extensions}.
        By definition, the utility ratio $\ratio(L,\beta)$ is at most 1.
        We show that for two groups (i.e., $p=2$), if the utilities are independently drawn from the uniform distribution, then no matter which non-intersectional lower bounds $L_1\geq 0$ and $L_2\geq 0$ the {policy maker chooses, $\ratio(L,\beta)$ is bounded away from 1 whenever  $\beta_1 <1$ and $\beta_2 <1$.}
    \begin{theorem}[\textbf{Non-intersectional constraints cannot recover full utility}]\label{thm:ub}
      {Suppose the latent utilities are uniformly distributed on $[0,1]$, %
      the fraction of candidates selected is between $\eta$ and $1-\eta$ for some constant $\eta>0$ (i.e., $\eta<\frac{n}{m}<1-\eta$), and the size of each intersection is greater than $\rho m$ (i.e., for all $\sigma\in \zo^2$, $\abs{I_\sigma}>\rho m$).
      For all implicit bias parameters $0<\beta_1,
      \beta_2<1$ and constants $\eta>0$ and $\rho>0$, there is a threshold $m_0\in \N$ such that if the number of candidates is more than this threshold, $m\geq m_0$, then for any non-intersectional lower {bounds $L_1,L_2\geq 0$ the utility ratio is strictly smaller than one, where the difference between $1$ and the utility ratio depends up on $\eta,\rho,\beta_1$, and $\beta_2$ as follows:}
      \begin{align*}
        \ratio(L,\beta)\leq
        1 - \inparen{\frac{\rho}{3}\cdot \min\inbrace{\eta,1-\eta}\cdot (1-\beta_1) \cdot(1-\beta_2)}^2.
        \yesnum\label{eq:upperbound_uniform_two_groups}
      \end{align*}}
    \end{theorem}
      \noindent Thus, \cref{thm:ub} establishes that for uniformly distributed utilities, a policy maker cannot recover the full utility ratio with  non-intersectional constraints for any bias parameters $0<\beta_1,\beta_2 < 1$.
      As an example, suppose $\eta=\frac{1}{2}$ and $\rho=\frac{1}{4}$, then \cref{eq:upperbound_uniform_two_groups} says that $\ratio(L,\beta)\leq 1 - \inparen{\frac{1}{24}(1-\beta_1) \cdot(1-\beta_2)}^2$, which is strictly smaller than 1 for any $\beta_1<1$ and $\beta_2<1$ and is independent of the number of candidates $m$.
    This is in contrast to the setting with a single group, where non-intersectional constraints can recover a utility ratio arbitrarily close to 1 as the number of candidates $m$ increase~\cite{KleinbergR18,celis2020interventions,EmelianovGGL20}.

    We emphasize that the upper bound in \cref{thm:ub} does not depend on the specific groups $G_1$ and $G_2$ or the specific value of $n$:
    it only requires that the intersection sizes are not too small and $n$ is not too close to 0 or $m$.
    For instance, in several admissions and hiring contexts, one can expect $\eta$ non-vanishing as even the most selective undergraduate programs in the US select more than 5\% of the total applicants, which corresponds to $\eta\geq\frac{1}{20}$~\cite{pewCollegeRates2019}.
    Moreover, a majority of US undergraduate programs select between 0.4 and 0.8 fraction of the {total applicants, which corresponds to $\frac{2}{5}\leq \eta\leq \frac{4}{5}$~\cite{pewCollegeRates2019}.}

    \cref{thm:ub} immediately implies the same result for $p>2$, for instance, by adding empty groups or repeating the family of groups.
    In summary, \cref{thm:ub} shows that non-intersectional constraints are not sufficient to recover the entire latent utility in the presence of multiple and intersectional groups.
    The dependence on implicit bias parameters $\beta$, and bounds on intersection size $\rho$ and  selection rate $\eta$ are also natural and we discuss these below.

    {\em Dependence on $\beta$.}
            {One can verify that as $(\beta_1,\beta_2)\to (1,1)$, the upper bound in \cref{thm:ub} goes to 1.
        In particular, it becomes vacuous at $\beta_1=\beta_2=1$.
        This is expected as when $\beta_1=\beta_2=1$, no item faces implicit bias: for all items $i$, $w_i = \hw_i$.
        Further, the upper bound in \cref{thm:ub} also goes to 1, as $\beta_1\to 1$ while $\beta_2\in (0,1)$ is fixed (and vice versa).
        This holds as when $\beta_1=1$, only one group faces implicit bias.
        Thus, one can use the constraints given by \cite{celis2020interventions}, for the $p=1$ case, to recover the near-optimal utility.}

      {\em Dependence on $\rho$.}
        The upper bound also goes to 1 as $\rho\to 0$.
        In particular, it becomes vacuous when $\rho=0$.
        This is expected as when $\rho=0$, $G_1$ and $G_2$ may not intersect.
        If the intersection is empty, then non-intersectional and intersectional constraints are the same, and one can use intersectional constraints in \cref{thm:fullrecovery} to recover the near-optimal utility.

      {\em Dependence on $\eta$.}
        The upper bound in \cref{thm:ub} also goes to 1 as $\eta$ approaches either 0 or 1, while $\rho\in (0,1)$ is fixed.
        {This is because when $\eta=1$, then both $\tx$ and $\sx$ select all $m$ items and, hence, $\tx=\sx$. Therefore, when $\eta=1$, the utility ratio is 1.}
        When $\eta$ is close to 0, $n$ is significantly smaller than $m$ and, hence, the best $n$ items in each intersection have latent utility very close to 1 with high probability.
        In this case, even if a decision maker is extremely biased, say {they only select} items from one intersection (e.g., White men), they would select $n$ items whose latent utility is very close to $n$, {which is the maximum value of the latent utility when utilities are drawn from the uniform distribution.}

        The basic property used to prove \cref{thm:ub} is that: Any selection which deviates from proportional representation has a latent utility smaller than optimal.
        At a high level, \cref{thm:ub} holds because the decision maker can alter their selection by selecting fewer (respectively more) items from $I_{11}\coloneqq G_1\cap G_2$ and more (respectively fewer) items from $I_{01}$ and $I_{01}$, while keeping the number of selections from $G_1$ and from $G_2$ invariant.
        In the proof, we show that for any $L_1$ and $L_2$, the decision maker selects a significantly less-than-proportional number of candidates from at least one of the four intersections.
        The proof of \cref{thm:ub} appears in \cref{sec:proofof:thm:ub} and an overview of the proof appears in Section~\ref{sec:proofoverviewof:thm:ub}.
        
        \begin{figure*}[t!]
            \vspace{-2mm}
            \centering
            \subfigure[Maximum utility ratio attained by non-intersectional constraints as a function of $\beta$, where the implicit bias parameters are $\beta_1\coloneqq \beta$ and $\beta_2\coloneqq \beta$.
            \label{fig:2d_beta_vs_util}
            ]{
                {\includegraphics[width=0.47\linewidth, trim={-1cm 0.25cm 0.5cm 1.75cm},clip]{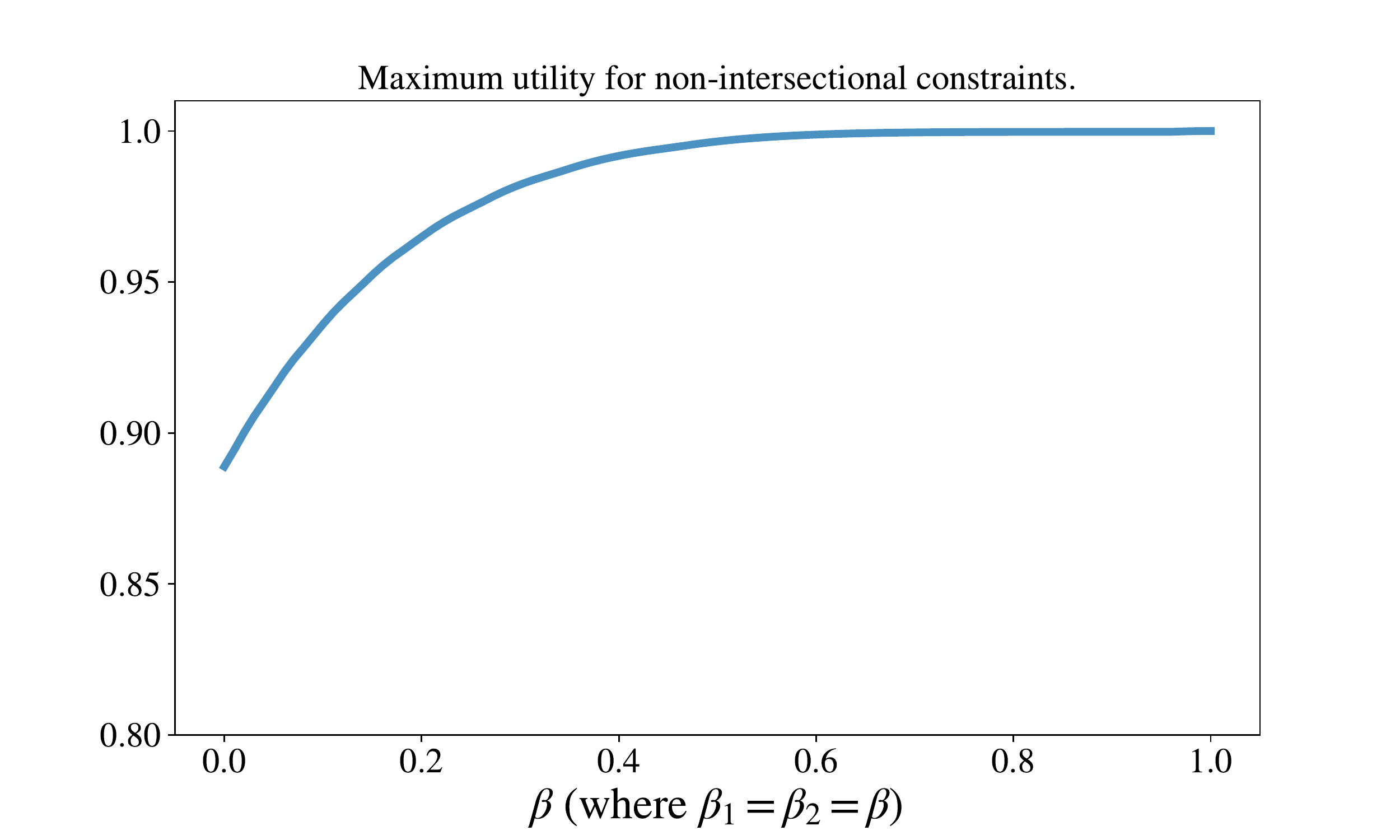}}
            }
            \hfill
            \subfigure[Maximum utility ratio attained by non-intersectional constraints as a function of implicit bias parameters $\beta_1$ and $\beta_2$.
            \label{fig:3d_beta_vs_util}
            ]{
                \hspace{3mm}
                {\includegraphics[width=0.38\linewidth, trim={5cm 2cm 5cm 3cm},clip]{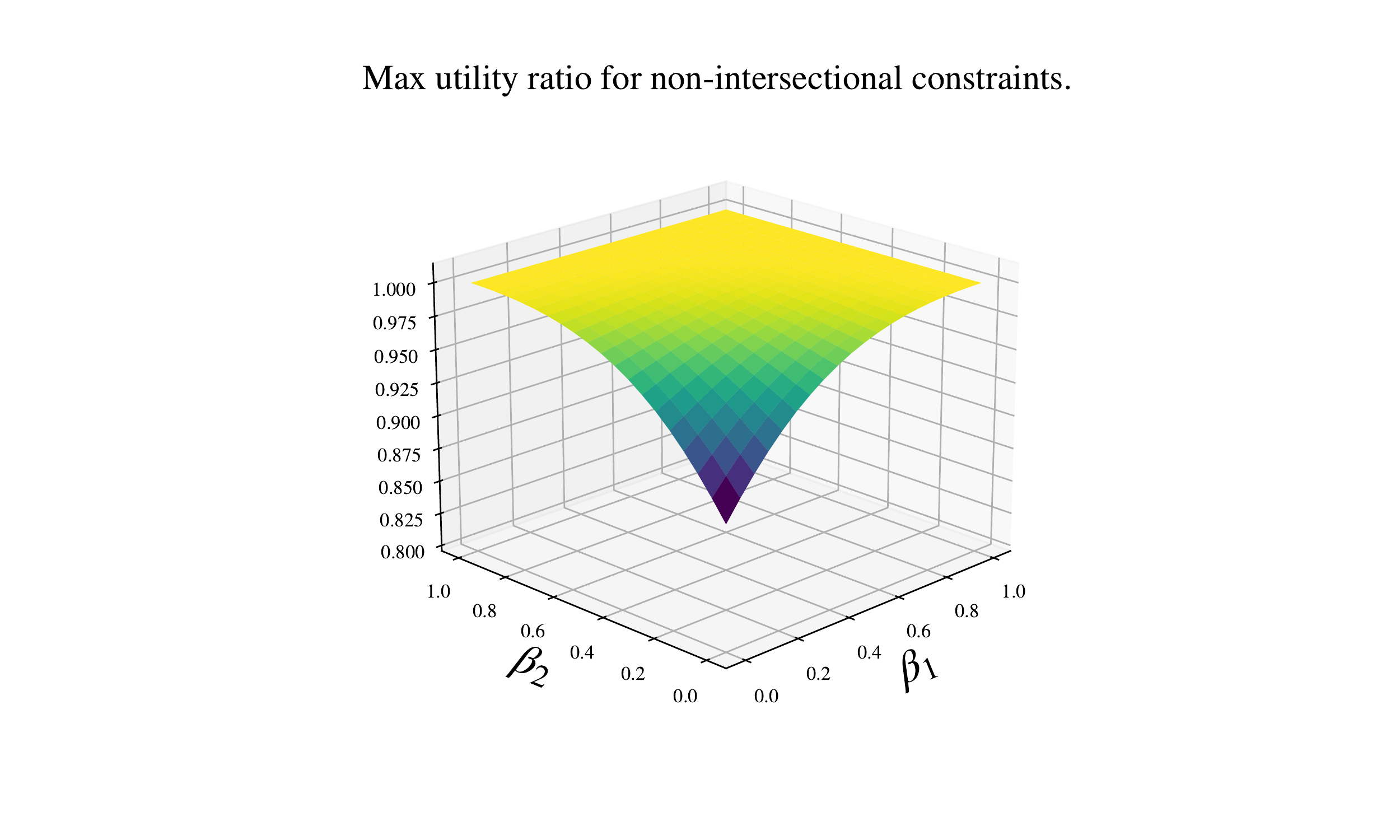}}
                \hspace{3mm}
            }
            \vspace{-4mm}
            \caption{\small 
            \cref{fig:2d_beta_vs_util,fig:3d_beta_vs_util} plot the maximum utility ratio attained by non-intersectional constraints as a function of (1) $\beta$, where $\beta_1\coloneqq \beta$ and $\beta_2\coloneqq \beta$ and (2) as a function of $\beta_1$ and $\beta_2$ respectively.
            In these plots, we fix a specific instance of the selection problem where half of the candidates are selected (i.e., $n=\frac{m}{2}$) and all intersections have size $\frac{m}{4}$ (i.e., $\forall\sigma\in \zo^2$, $\abs{I_\sigma}=\frac{m}{4}$).
            For this instance, these figures plot the value of the utility ratio in limit that $m\to\infty$.
            One can observe that the maximum utility ratio attained is strictly less than 1 for all $0<\beta_1,\beta_2<1$ and approaches 1 as $\max\inbrace{\beta_1,\beta_2}\to 1$.
            Moreover, when the bias is extreme, i.e., $\beta_1$ and $\beta_2$ are close to 0, then the maximum utility ratio, attained by non-intersectional constraints, approaches $\frac{8}{9}$ (see \cref{prop:89}).
            In contrast, \cref{thm:fullrecovery} shows that intersectional constraints achieve a utility ratio of 1, in the limit that $m\to\infty$, for all $0<\beta_1,\beta_2<1$.
            }
            \label{fig:beta_vs_utility_non-intersectional}
        \end{figure*}

  \subsection{Optimal Latent Utility With Intersectional Constraints}\label{sec:positive_result}

        Complementing \cref{thm:ub}, we show that, if {the policy maker is} allowed to place constraints on intersections, then {they} can recover all but a vanishing (with $m$) fraction of the total latent utility.
        In particular, for any {desired constant} $0<\eps<1$, number of items $m$, and groups  $G_1,\dots,G_p$, we give {lower bounds $L_\sigma\geq 0$, for each intersection $\sigma$,}
        such that for {\em any} implicit bias parameters $0<\beta_1, \ldots, \beta_p\leq 1$, constrained selection leads to a utility ratio more  than $1-\eps$.
    \begin{theorem}[\textbf{Intersectional constraints can (asymptotically) recover  full utility}]\label{thm:fullrecovery}
      Suppose that the fraction of selected candidates is greater than $\eta$ for some constant $\eta>0$ (i.e., $\frac{n}{m}>\eta\cdot m$).
      For all constants $0<\eps<1$, $\eta>0$, and number of groups $p\in \N$, there exists a threshold $m_0\in \N$,
      such that if the number of candidates is more than this threshold, $m\geq m_0$, then for any groups $G_1,\dots,G_p\subseteq [m]$,
      there exist intersectional constraints, $L_\sigma\geq 0$ for each intersection $\sigma\in \zo^p$, which for any
      implicit bias parameters $0<\beta_1,\dots,\beta_p\leq 1$ and any continuous distribution $\cD$ with non-negative support, have a utility ratio at least $1-\eps$, i.e.,
      $\ratio_\cD(L,\beta)\geq 1-\eps.$
    \end{theorem}
    \noindent Thus, \cref{thm:fullrecovery} shows that a policy maker using intersectional constraints can recover utility ratio arbitrarily close to 1 for any bias parameters $0<\beta_1,\beta_2 < 1$.
    In contrast to \cref{thm:ub}, here, the difference between the utility ratio and 1 approaches 0 as the number of candidates increase.
    {Moreover, \cref{thm:fullrecovery} holds for a larger choice of parameters than \cref{thm:ub}:
    (1) Any constant $\eta$ and sizes of the intersections ($\forall\sigma\in \zo^m,\ \abs{I_\sigma}$) that satisfy the conditions in \cref{thm:ub} also satisfy the conditions in \cref{thm:fullrecovery}, and (2) while $\cD$ is fixed to be the uniform distribution on $[0,1]$ in \cref{thm:ub}, $\cD$ can be the uniform distribution on $[0,1]$ or any other continuous distribution in \cref{thm:fullrecovery}.}
    
    We note that the intersectional constraints promised in \cref{thm:fullrecovery} do not depend on the values of the implicit bias parameters $\beta_1,\dots,\beta_p$ and the distribution of latent utility $\smash{\cD}$ and only depend on the intersection sizes, i.e., $\abs{I_\sigma}$ for all $\sigma\in \zo^p$.
    Thus, they are applicable in practice when the amount of implicit bias and utility distributions are not know.
    Moreover, independence from $\beta_1,\dots,\beta_p$ allows the same constraints to be used for different decision makers -- who may have different implicit biases --
    and independence from $\cD$ allows them to be used across selection tasks with different utility distributions.
    {Finally, the independence also allows the constraints to be stable over time, while the implicit bias~\cite{lawrence1985racial,patterns2019charlesworth} and utility distributions may change \cite{dorans2002recentering}.}
    The proof of \cref{thm:fullrecovery} appears in \cref{sec:proofof:thm:fullrecovery} and its overview appears in Section~\ref{sec:proofoverviewof:thm:fullrecovery}.

    \indent {\em Independence from $\beta$ and $\cD$.}
    The independence from $\beta_1,\dots,\beta_p$ {relies on the fact that  if the optimal constrained selection vector $\tx$ picks $k\in \N$ items from an intersection, then these are the $k$ items with the highest latent utility in the intersection, irrespective of the value of $\beta_1,\dots,\beta_p$.}
    The independence from $\cD$ relies on the fact that the sets of items selected by $\sx{}$ and $\tx$ from an intersection only depend on the order of the latent utility of items and not {the actual values of latent utilities of items.}
    Note, however, that while the lower bound on the utility ratio is independent of $\cD$, the actual utility achieved by the non-intersectional constraints depends on $\cD.$
    
    \indent {\em Computational Complexity.}
        Since each candidate belongs to exactly one intersection and there are $m$ candidates, there are at most $m$ non-empty intersections.
        Hence, even when the total number of intersections $2^p$ is larger than $m$, the decision maker has to consider at most $m$ (non-empty) intersections (as they cannot select candidates from empty intersections).
        Using this, we can show that the decision maker can find $\wt{x}$ in $O(m\log{m})$ arithmetic operations.

  \subsection{Extensions of \cref{thm:ub} and \cref{thm:fullrecovery}}\label{sec:results:extensions}
  {\em Extensions of \cref{thm:ub}.}
  A natural question raised by Theorem \ref{thm:ub} is:
  {How small can the upper bound on $\ratio(L,\beta)$ be?}
  We show that  $\ratio(L,\beta) \leq \smash{\frac{8}{9} + \frac{3}{2}}\cdot \max(\beta_1,\beta_2)$.
  While this bound becomes vacuous when either $\beta_1\geq \smash{\frac{2}{27}}$ or $\beta_2\geq \smash{\frac{2}{27}}$,
  {in the regime of large implicit bias, i.e., when $\beta_1$ and $\beta_2$ are close to $0$, this bound approaches $\frac89$.}

    \begin{proposition}\label{prop:89}
        Suppose the latent utilities are uniformly distributed on $[0,1]$, the fraction of candidates selected is $\frac{1}{2}$ (i.e., $\frac{n}{m}=\frac{1}{2}$), and $G_1$ and $G_2$ are such that all intersections have size $\frac{m}{4}$ (i.e., for all $\sigma\in \zo^2$, $\abs{I_\sigma}=\frac{m}{4}$).
        For all implicit bias parameters $0<\beta_1,\beta_2\leq 1$ and for all non-intersectional lower-bound constraint $L_1,L_2\geq 0$,
        there exists an threshold $m_0\in \Z$, such that for all $m\geq m_0$, the utility ratio is upper bounded by a quantity that approaches $\frac{8}{9}$ as $\beta_1$ and $\beta_2$ go to 0, where the specific dependence of the upper bound on $\beta_1$ and $\beta_2$ is as follows:\  \
        $\ratio(L,\beta)\leq \frac{8}{9} + \frac{3}{2}\cdot \max(\beta_1,\beta_2)$.
    \end{proposition}

    \noindent
    We do not know if the bound in \cref{prop:89} is tight but conjecture that it cannot go below $\frac89$ for the uniform distribution.
    The proof of \cref{prop:89} appears in \cref{sec:proofof:prop:89}.

  Another question is if \cref{thm:ub} is specific to the uniform distribution of utilities or to the implicit bias model in Equation~\eqref{eq:mult_bias}.
  Other relevant distribution families include power-law and truncated normal.
  These families arise in peer-review and university admissions as distributions of  citations~\cite{clauset2009power} and test scores, e.g., SAT~\cite{dorans2002recentering}, respectively.
  As for the implicit bias model, Equation~\eqref{eq:mult_bias} assumes that implicit bias experienced by an item $i$ in two groups, say, $G_1$ and $G_2$ is exactly the product of the implicit bias parameters $\beta_1$ and $\beta_2$ of these groups, i.e., $\beta_1\beta_2$.
  But depending on the specific context, the implicit bias could be more severe, e.g., $\inparen{\beta_1\beta_2}^2$, or less severe, e.g.,  $\inparen{\beta_1\beta_2}^{\frac{1}{2}}$.

  We show that a version of \cref{thm:ub} holds for truncated power-law and normal distributions and the above implicit bias models.
  We consider a family of models where the implicit bias at each intersection $\sigma$ is given by a strictly increasing function $b_\sigma\colon \R_{\geq 0}\to \R_{\geq 0}$,
  such that, the observed utility of each item $i\in I_\sigma$ in this intersection is:
  \begin{align*}
      \hw_i\coloneqq b_{\sigma}(w_{i}).
      \yesnum\label{eq:gen_model}
  \end{align*}
  \noindent This model captures the model in \Eqref{eq:mult_bias} when $b_\sigma(x)\coloneqq x \cdot {\prod\nolimits_{\ell\in [p]\colon \sigma_\ell = 1}\beta_\ell}$.
  Further, when $p=2$, changing $b_{11}(x)\coloneqq x\cdot \inparen{\beta_1\beta_2}^2$ or $b_{11}(x)\coloneqq x\cdot \inparen{\beta_1\beta_2}^{\frac{1}{2}}$ in the previous sentence captures the models mentioned above.
  We prove a version of \cref{thm:ub} under the following assumption: %

  \begin{assumption}\label{asmp:2}
        There is are constants $c,d>0$ such that
        the probability density function of $\cD$, $\mu_\cD$, is differentiable and for all $x$ in the support of $\cD$ and takes values between $c$ and $\frac{1}{c}$ (i.e., $c \leq \mu_\cD(x)\leq \frac{1}{c}$), and
        the bias function for each intersection (i.e., $b_\sigma$ for each $\sigma\in \zo^p$) is  differentiable, takes value at most  $\frac{1}{c}$ (i.e., $b_\sigma(x)\leq \frac{1}{c}$), and has a derivative between $d$ and $\frac{1}{c}$ (i.e., $d\leq b_\sigma'(x)\leq \frac{1}{c}$) for each $x$ in the support of $\cD$.
  \end{assumption}
  \noindent
  The assumption on $b_\sigma'$ ensures that two items with similar latent utilities also have similar observed utilities:
  if $\abs{w_i-w_j}=\eps$ then $d\eps\leq \abs{\hw_i-\hw_j}\leq \frac{\eps}{c}$.
  The functions $b_\sigma$ show up in the generalization of \cref{thm:ub}, and the upper bound on $b_\sigma$ is needed to re-scale the right-hand side of \cref{eq:thm_gen_ub} to be non-negative.
  {The assumption on $\cD$ is satisfied, for instance, by the truncated normal and power-law distributions:
  This is because their probability density functions are bounded between two positive constants at any point in the support.
  As opposed to them, the {\em un-truncated} normal and power-law distributions do not satisfy \cref{asmp:2}, this is because their probability density function, $\mu$, approaches 0 for large values (i.e., $\mu(x)\to 0$ as $x\to\infty$) and, hence, does not satisfy the lower bound.}
  Because of this, for these distributions, the constant $c$ decreases as the upper-end of the truncation interval increases.
  Further, the extensions of the model in \Eqref{eq:mult_bias} discussed above
  also satisfy \cref{asmp:2}:
  In particular, if $\cD$ is the uniform distribution on $[0,1]$ then \cref{asmp:2} holds for $c=1$ and for $d=\inparen{\beta_1\beta_2}^2$ (and $d=\inparen{\beta_1\beta_2}^\frac{1}{2}$ respectively).
    \begin{theorem}\label{thm:gen_ub}
    Suppose the fraction of candidates selected is between $\eta$ and $1-\eta$ for some constant $\eta>0$ (i.e., $\eta<\frac{n}{m}<1-\eta$), and the size of each intersection is at least $\rho m$ (i.e., for all $\sigma\in \zo^2$, $\abs{I_\sigma}>\rho m$).
    Under the implicit bias model in Equation~\eqref{eq:gen_model}, for any continuous distribution $\cD$ with non-negative support and set of strictly increasing functions $b_\sigma\colon \R_{\geq 0}\to \R_{\geq 0}$, where there is one function for each intersection $\sigma\in\zo^2$, if \cref{asmp:2} holds,
    then there is a threshold $m_0\in \N$ such that when the number of candidates is more than this threshold, $m\geq m_0$,
    any non-intersectional lower-bound constraint $L_1,L_2\geq 0$
    the utility ratio   $\ratio_\cD(L,\beta)$ is at-most
    \begin{align*}
  \footnotesize    1 - \inparen{c^4 \rho\min\inbrace{\eta,1-\eta}  (b_{00}-b_{10}-b_{01}+b_{11})\circ \inparen{F_\cD^{-1}\inparen{1-\frac{n}{m}}}}^2.\footnotemark
      \yesnum\label{eq:thm_gen_ub}
    \end{align*}
    \end{theorem}
    \footnotetext{Given functions $f_1,f_2,f_3,$ and $f_4$, and a number $x\in \R$, we use  $(f_1+f_2+f_3+f_4)\circ(x)$ to denote $f_1(x)+f_2(x)+f_3(x)+f_4(x)$.}
    \noindent
    Thus, for any bias functions and distributions for which $(b_{00}-b_{10}-b_{01}+b_{11})\circ \inparen{F_\cD^{-1}\inparen{1-\frac{n}{m}}}$ is non-zero, \cref{thm:gen_ub} shows that the utility ratio is strictly smaller than 1 for any non-intersectional lower bounds $L_1$ and $L_2$.
    Complementing this, we show that when this specific additive quantity is 0, there are non-intersectional constraints which recover near-optimal latent utility (\cref{prop:robustness_ub_complement}).
    We present the proof of \cref{thm:gen_ub} in {\cref{sec:proofof:thm:gen_ub} and discuss how it differs from the proof of \cref{thm:ub} in \cref{sec:proofoverviewof:thm:ub}.}

    {\em Extensions of \cref{thm:fullrecovery}.} The proof of \cref{thm:fullrecovery} only needs the following property of the implicit bias model:
          For two items in the same intersection, the order of their latent utilities is the same as the order of their observed utilities.
          Abstracting this leads us to the generalization of \cref{thm:fullrecovery} to the implicit bias in  Equation~\eqref{eq:gen_model}, which was discussed in \cref{sec:upper_bound}.
          This, in particular, shows that intersectional constraints can have an arbitrarily high utility ratio for the generalization of the model of implicit bias considered in \cref{thm:gen_ub}.
        \begin{corollary}\label{coro:fullrecovery}
          Suppose that the fraction of candidates selected is equal to some constant $\eta>0$ (i.e., $\frac{n}{m}=\eta$).
              Given $0<\eps<1$ and $\eta>0$, consider the threshold $m_0$ in \cref{thm:fullrecovery},
              for all $m\geq m_0$ and group structures $G_1,\dots,G_p\subseteq [m]$, the intersectional lower bound $L$ from \cref{thm:fullrecovery} are such that, for any set of strictly increasing functions $b_\sigma\colon \R_{\geq 0}\to \R_{\geq 0}$, where there is one function for each intersection $\sigma\in\zo^p$,
              they satisfy $\ratio_\cD(L,\beta)\geq 1-\eps.$
            \end{corollary}
    \noindent Finally, we state a desirable fairness property of the constraints in \cref{thm:fullrecovery}.
    For each intersection $\sigma\in \zo^{p}$, the constraints in \cref{thm:fullrecovery} satisfy $L_\sigma\geq (1-\eps)\cdot n\cdot  \frac{\abs{I_\sigma}}{m}$.
    Hence, for small $\eps>0$, the constrained decision maker must select at least a near-proportional number of items from each intersection.
    Thus, these constraints can be seen as a form of affirmative action.

\section{Overview of Proofs}\label{sec:overview_of_proofs}
    \subsection{Proof Overview of \cref{thm:ub}}\label{sec:proofoverviewof:thm:ub}
        In this section, we explain the key ideas in the proof of \cref{thm:ub} under some simplifying assumptions.
        Recall that utility ratio is $\ratio(L,\beta)\coloneqq \Ex_{w}\insquare{ \frac{ \sinangle{\tx,w} }{ \sinangle{\sx,w} } },$
        where $\tx$ is the selection picked by the constrained decision maker (as defined in Equation~\eqref{eq:css}) and $\sx$ is the selection with the highest latent utility (as defined in Equation~\eqref{eq:def_xstar}).
        We are given the implicit bias parameters $0<\beta_1,\beta_2\leq 1$ and the guarantee that for some $\eta>0$ and $\rho>0$,
        \begin{align*}
            \forall \sigma \in \zo^2,\quad
            \rho <\frac{\abs{I_\sigma}}{m} < 1-\rho
            \quad\text{and}\quad 
            \eta<\frac{n}{m}<1-\eta.
            \yesnum\label{asmp:1}
        \end{align*}
        For $\eta, \rho>0$, let  %
            $\phi(\beta_1,\beta_2) \coloneqq  \inparen{\frac{\rho}{3} \min\inbrace{\eta,1-\eta} (1-\beta_1) (1-\beta_2)}^2.$
        Our goal is to show that
        \begin{align*}
      \max_{L_1,L_2\in \Z_{\geq 0}}\ratio(L,\beta)\leq
      1-\phi(\beta).\yesnum
    \end{align*}
        \noindent In this overview, our first simplifying assumption is that for all $L$ and $\beta$:
        $\ratio(L,\beta)\coloneqq \Ex\nolimits_{w}\insquare{ \frac{ \sinangle{\tx,w} }{ \sinangle{\sx,w} }} = \frac{ \Ex\nolimits_{w}\insquare{\sinangle{\tx,w}} }{ \Ex\nolimits_{w}\insquare{\sinangle{\sx,w}} }.$
        This assumption is {\em not true}. In the proof, we remove it by showing that the distribution of $\sinangle{\sx,w}$ is concentrated around its mean.
        Since $\sx$ is not a function of $L$, under the above assumption, it suffices to show that
        \begin{align*}
             \max_{L_1,L_2\in \Z_{\geq 0}} \Ex\nolimits_{w}\insquare{\sinangle{\tx,w}}
             \leq
             \inparen{1-\phi(\beta)} \cdot \Ex\nolimits_{w}\insquare{\sinangle{\sx,w} }.
             \yesnum\label{eq:to_show}
        \end{align*}
\noindent
        {\bf Challenges in computing $\Ex\nolimits_{w}\insquare{\sinangle{\tx,w}}$.}
        One approach to prove \cref{eq:to_show} could be to compute $\Ex\nolimits_{w}\insquare{\sinangle{\tx,w}}$ explicitly as a function of $L$ and then verify \cref{eq:to_show}.
        This is the approach that \cite{celis2020interventions} take for a single group.
        In their case, there is a simple iterative algorithm that, given $L_1$ and observed utilities, computes $\tx$.
        The algorithm is as follows:
        Pick $L_1$ items with the highest observed utility from $G_1$ and from the remaining items pick $n-L_1$ items that have the highest observed utility.
        \cite{celis2020interventions} analyze this algorithm to compute $\Ex\nolimits_{w}\insquare{\sinangle{\tx,w}}$ as a function of $L_1$.

        This algorithm and analysis straightforwardly extends to multiple {\em non-overlapping} groups:
        For each $\ell \in [p]$, pick $L_\ell$ items with the highest observed utility from $G_\ell$, then from the remaining items pick $n-L_1-L_2-\dots-L_p$ items that have the highest observed utility.
        However, this algorithm breaks down when groups overlap.
        This is because items at intersections of multiple groups can satisfy multiple lower bounds.
        Moreover, at least in the case where the number of groups, $p$, is non-constant we do not expect there to be a simple algorithm which computes $\tx$:
        This is because the NP-complete hitting set problem reduces to checking if $\tx$ exists, i.e., if the specified lower bounds are satisfiable.

        Instead of computing $\Ex\nolimits_{w}\insquare{\sinangle{\tx,w}}$ as a function of $L$,
        we express $\Ex\nolimits_{w}\sinsquare{\sinangle{\tx,w}}$ and $\Ex\nolimits_{w}\sinsquare{\sinangle{\sx,w}}$ as solution to two optimization programs. %
        Then, we directly upper bound the ratio $\frac{\Ex\nolimits_{w}\insquare{\sinangle{\tx,w}}}{\Ex\nolimits_{w}\insquare{\sinangle{\sx,w}}}$ by analyzing the optimization programs.

        \paragraph{Step 1: Reduce computing $\sx$ and $\tx$ to a small number of variables.}
        A property of both $\sx$ and $\tx$ is that if they pick $k$ items from $I_\sigma$, then these are the $k$ items with the highest latent utility, or {\em equivalently} the highest observed utility, in $I_\sigma$ (see \cref{obs:wtx_picks_top_latent_util}).
        Hence, determining $\sx$ and $\tx$ reduces to computing, the following quantities for each $\sigma\in \zo^2$ 
        \begin{align*}
            K^\star_\sigma\coloneqq \sum_{i\in I_\sigma} \sx_i \quad \text{and}\quad \wt{K}_\sigma\coloneqq \sum_{i\in I_\sigma} \tx_i.
        \end{align*}

        \paragraph{Step 2: Express $K^\star$ and $\wt{K}$ as solutions to different optimization problems.}
        Since $\sx$ and $\tx$ are functions of randomly generated utilities, they and, hence, $K^\star$ and $\wt{K}$ are random variables.
        Under Assumption~\eqref{asmp:1}, we show that $K^\star$ and $\wt{K}$ are concentrated around the optimizers of optimization problems \eqref{prog:1} and \eqref{prog:2} respectively; where
        for \mbox{any $\gamma\in [0,1]^4$ and $k\in \R_{\geq 0}^4$}
        \begin{align*}
                f_\gamma(k)
                &\coloneqq \sum_\sigma \gamma_\sigma \abs{I_\sigma} \cdot \int_{1-\frac{k_\sigma}{\abs{I_\sigma}}}^1  x dx.
                \yesnum\label{eq:def_g}
        \end{align*}
              \begin{align*}
                    &\argmax\nolimits_k \quad f_1(k),\yesnum\label{prog:1}\\
                    &\qquad \st,  \quad\sum\nolimits_\sigma k_\sigma=n,\\
                    &\qquad\qquad\ \ \forall\sigma\in \zo^2,\quad  0\leq k_\sigma \leq \abs{I_\sigma}.
                \end{align*}
              \begin{align*}
                    &\argmax\nolimits_k \quad f_{\beta}(k),\yesnum\label{prog:2}\\
                    &\qquad \st,\quad k_{10}+k_{11}\geq L_1 \text{ and }  k_{01}+k_{11}\geq L_2,\\
                    &\qquad\qquad\ \ \sum\nolimits_\sigma k_\sigma=n,\\
                    &\qquad\qquad\ \ \forall\sigma\in \zo^2,\  0\leq\negsp{} k_\sigma \negsp{}\leq\negsp{} \abs{I_\sigma}.\negsp\negsp
                    \yesnum\label{eq:prog:2:grp_sz_constraint}
                \end{align*}
        \noindent While the integral in \Eqref{eq:def_g} can be computed exactly, we use the integral-form because it generalizes to other distributions and bias functions, where the resulting integral may not be possible to compute.
        Let $x(k)$ be the selection that picks $k_\sigma$ items with the highest latent utility from $I_\sigma$ for all $\sigma\in \zo^2$.
        Suppose $k$ is feasible for \prog{prog:1} and \prog{prog:2}.
        The constraints in \prog{prog:1} ensure $x(k)$ selects a total of $n$ candidates, and the additional constraints in \prog{prog:2} ensure $x(k)$ picks at least $L_1$ and $L_2$ candidates from $G_1$ and $G_2$.
        $f_\beta(k)$ roughly measures the expected {\em observed} utility of $x(k)$ and
        {$f_1(k)$ roughly measures the expected {\em latent} utility $x(k)$:}
        \begin{align*}
            f_\beta(k)=\Ex\insquare{\inangle{x(k),\hw}}\pm O(m^{-1})\quad \text{and}\quad  f_1(k)=\Ex\insquare{\inangle{x(k),w}}\pm O(m^{-1}).
            \yesnum\label{eq:def_f}
        \end{align*}
        \noindent $f_1$ and $f_\beta$ can be shown to be strongly concave for all $\gamma$, and hence, Programs~\eqref{prog:1} and \eqref{prog:2} have unique solutions.
        Formally, we prove the following concentration bound on $K^\star$ and $\wt{K}$.
        \begin{lemma}\label{lem:conc_of_k}
            Let $s^\star$ and $\wt{s}$ be the optimizers of Programs~\eqref{prog:1} and \eqref{prog:2} respectively.
            With probability at least $1-O\sinparen{m^{-\frac{1}{4}}}$, %
            \begin{align*}
                \forall\sigma\in \zo^2,\quad
                \sabs{K^\star_\sigma - s^\star_\sigma} \leq O\sinparen{nm^{-\frac{1}{4}}}
                \quad \text{and}\quad
                \sabs{\wt{K}_\sigma - \wt{s}_\sigma} \leq O\sinparen{nm^{-\frac{1}{4}}}.
            \end{align*}
        \end{lemma}
        \noindent Suppose that $K^\star$ and $\wt{K}$ are equal to minimizers of Programs~\eqref{prog:1} and \eqref{prog:2} with probability 1.
        Then, from Equation~\eqref{eq:def_f},
        \begin{align*}
            \Ex\insquare{\sinangle{\sx, w}}=f_1(K^\star)\pm O(m^{-1})
             \quad \text{and}\quad
            \Ex\insquare{\inangle{\tx,w}}=f_1(\wt{K})\pm O(m^{-1}).
        \end{align*}
        \noindent Suppose these hold with equality.
        Then because the feasible region of \prog{prog:1} is a superset of the feasible region of \prog{prog:2} and $K^\star$ maximizes $f_1$ over the feasible region of \prog{prog:1}, it follows that
        $f_1(\wt{K})=\Ex\insquare{\inangle{\tx,w}}\leq \Ex\insquare{\sinangle{\sx, w}}=f_1(K^\star).$
        \noindent The question is:
        Is $f_1(\hat{K})$ significantly smaller than $f_1(K^\star)$? Does \cref{eq:to_show_2} (below) hold?
        \begin{align*}
            f_1(\wt{K})
            \leq\inparen{1-\phi(\beta)}\cdot f_1(K^\star).
            \yesnum\label{eq:to_show_2}
        \end{align*}

        \paragraph{Step 3: Prove  Equation~\eqref{eq:to_show_2} (Main argument).}
            Our first observation is that both $f_1$ is $\frac{1}{(1-\rho)m}$-strongly concave and $f_\beta$ is $\frac{1}{\rho m}$-Lipschitz continuous.
            Using the gradient test, we can analytically solve \prog{prog:1} to get $$K^\star\coloneqq \inbrace{\abs{I_\sigma}\cdot \frac{n}{m}}_{\sigma\in \zo^2}.$$
            The claim follows because if $\wt{K}$ satisfies any inequality in  Equation~\eqref{eq:prog:2:grp_sz_constraint} with equality, then $\snorm{K^\star-\wt{K}}_2$ is large.
            Namely,
            \begin{align*}
                \snorm{K^\star-\wt{K}}_2^2\geq 2nm(1-\rho)\cdot \phi(\beta).
                \yesnum\label{eq:lb_on_difference}
            \end{align*}
            Hence, by the $\frac{1}{(1-\rho)m}$-strong convexity of $f_1$ and the fact that $f(K^\star)\leq n$, Equation~\eqref{eq:to_show_2} holds.
            Otherwise if $\wt{K}$ satisfies all inequalities in  Equation~\eqref{eq:prog:2:grp_sz_constraint} with strict inequality, then because \prog{prog:2} has only three other constraints, while $\wt{K}$ has four coordinates, it follows that there is some constant $t_0>0$ and a vector, namely $v\coloneqq (1,-1,-1,1)$ such that for all $-t_0\leq t\leq t_0$, $\wt{K}+tv$ is feasible for \prog{prog:2}.
            Since $\wt{K}$ is the optimal solution of \prog{prog:2}, this implies that
            \begin{align*}
                \sinangle{\nabla f_\beta(\wt{K}), v} = 0.
                \yesnum\label{eq:zero_grad}
            \end{align*}
            Using the value of $K^\star$, we have
            \begin{align*}
                \inangle{\nabla f_\beta(K^\star), v}
                =\inparen{1-\frac{n}{m}}\cdot \inparen{1-\beta_1}\cdot \inparen{1-\beta_2}.
                \yesnum\label{eq:exp_grad}
            \end{align*}
            This is sufficient to show that, in this case, $\snorm{\nabla f_\beta(\wt{K}) - \nabla f_\beta(K^\star)}_2$ is large:
            \begin{align*}
                \snorm{\nabla f_\beta(\wt{K}) - \nabla f_\beta(K^\star)}_2
                &\geq \frac{1}{\norm{v}_2}\cdot \abs{\inangle{\nabla f_\beta(\wt{K}) - \nabla f_\beta(K^\star), v}}
                \\ 
                &  \Stackrel{\eqref{eq:zero_grad},\eqref{eq:exp_grad}}{=}\quad\  \frac{1}{2}\inparen{1-\frac{n}{m}}\cdot \inparen{1-\beta_1}\cdot \inparen{1-\beta_2}.
            \end{align*}
            Combined with the fact that $f_\beta$ is $\frac{1}{\rho m}$-Lipschitz continuous, this implies that $$\snorm{K^\star-\wt{K}}_2^2\geq \frac{\rho^2 m^2}{4}\cdot \inparen{1-\frac{n}{m}}^2\cdot \inparen{1-\beta_1}^2\cdot \inparen{1-\beta_2}^2 \geq 2mn(1-\rho)\cdot \phi(\beta).$$
            Thus, in this case also Equation~\eqref{eq:lb_on_difference} follows from  $\frac{1}{(1-\rho)m}$-strong convexity of $f_1$.

            \paragraph{Generalization to other bias functions and distributions.}
                The proof of \cref{thm:gen_ub} is analogous to that of \cref{thm:ub}.
                Let $F_\cD$ be the cumulative distribution function of $\cD$ and $\mu_\cD$ be the probability density function of $\cD$.
                The main difference is that $f_1$ and $f_\beta$ change to
                \begin{align*}
                     f_1(k)
                    & \coloneqq \sum_\sigma \abs{I_\sigma} \int_{z_\sigma(k)}^{z_\sigma(0)}  x d\mu_\cD(x),
                \quad\text{and}  \quad                   f_b(k)
                    \coloneqq \sum_\sigma \abs{I_\sigma} \int_{z_\sigma(k)}^{z_\sigma(0)}  b_\sigma(x) d\mu_\cD (x)
                    \yesnum\label{eq:def_g_general}
                \end{align*}
                where $z_\sigma(k)\coloneqq F^{-1}_\cD\inparen{1-\frac{k_\sigma}{\abs{I_\sigma}}}.$
                We choose this definition of $f$ because of a similar reason:
                Under the general bias model and distribution, $f_b(k)$, roughly, measures the expected {\em observed} utility of $x(k)$ and
                $f_1(k)$, roughly, measures the expected {\em latent} utility $x(k)$.
                Next, we prove \cref{lem:conc_of_k} for the new definition of $f$.
                The rest of the proof follows analogously once we prove that $f_1$ is   $\frac{\Omega(1)}{(1-\rho)m}$-strongly concave and $f_b$ is $\frac{O(1)}{\rho m}$-Lipschitz continuous.
                
        \subsection{Proof Overview of \cref{thm:fullrecovery}}\label{sec:proofoverviewof:thm:fullrecovery}

        In the proof we show that with high probability, for each intersection $\sigma\in \zo^p$: $\frac{\sum_{i\in I_\sigma} \tx_i w_i}{\sum_{i\in I_\sigma} \sx_i w_i} > 1-\frac{\eps}{2}.$
    This implies:
    \begin{align*}
      \frac{\sinangle{\tx,w}}{\sinangle{\sx,w}} = \frac{\sum_{\sigma}\sum_{i\in I_\sigma} \tx_i w_i}{\sum_{\sigma}\sum_{i\in I_\sigma} \sx_i w_i}
      \negsp{}\Stackrel{}{>}\negsp{} \frac{\sum_{\sigma}
      \sum_{i\in I_\sigma} \sx_i w_i}{\sum_{\sigma}\sum_{i\in I_\sigma} \sx_i w_i} \inparen{1-\frac{\eps}{2}}
      > 1-\frac{\eps}{2}. \yesnum
    \end{align*}
    \noindent The claimed result follows by taking the expectation of the above quantity.
    At a high level, our strategy is to find constraints such that $\tx$ selects a similar number of items from each intersection as $\sx{}$.
    This suffices to prove the result due to a property of the implicit bias model:
    If $\tx$ picks $K$ items from any intersection $I_\sigma$, then these are the $K$ items with the highest latent utility in $I_\sigma$ (\cref{obs:wtx_picks_top_latent_util}).

    To see why this suffices, let $v_1\leq v_2\leq \dots\leq v_{|I_\sigma|}$ be the latent utilities of items in $I_\sigma$ in non-increasing order.
    Let $N$ and $\wt{N}$ be the random variables counting the number of items $\sx$ and $\tx$ select from $I_\sigma$.
    Using the above observation, we know that $\sx{}$ and $\tx$ select the items $v_1,v_2,\dots,v_{N}$ and $v_1,v_2,\dots,v_{\wt{N}}$ respectively.
    Hence,
    \begin{align*}
      \frac{\sum_{i\in I_\sigma} \tx_i w_i}{\sum_{i\in I_\sigma} \sx_i w_i}
      &\negsp{}\geq\negsp{} \frac{\sum_{j=1}^{\wt{N}}v_j  }{ \sum_{i\in I_\sigma} \sx_i w_i } %
      \negsp{}=\negsp{} \frac{\sum_{j=1}^{\wt{N}}v_j  }{ \sum_{j=1}^{N}v_j}
      \negsp{}=\negsp{} \frac{\sum_{j=1}^{\wt{N}}v_j  }{  \sum_{j=1}^{\wt{N}}v_j + \sum_{j=\wt{N} + 1}^{N}v_j }.
    \end{align*}
    \noindent Using that for all $c\geq 0$, $\frac{x}{c+x}$ is an increasing function of $x$, we have
    \begin{align*}
      \textstyle \frac{\sum_{i\in I_\sigma} \tx_i w_i}{\sum_{i\in I_\sigma} \sx_i w_i} &\geq \frac{ \wt{N}\cdot v_{{\wt{N}}}  }
      {  \wt{N}\cdot v_{{\wt{N}}}  + \sum_{j=\wt{N} + 1}^{N}v_j }%
      \geq \frac{ \wt{N}\cdot v_{{\wt{N}}}  }{  \wt{N}\cdot v_{{\wt{N}}}  + (N-\wt{N})\cdot v_{{\wt{N}}} }
      = \frac{ \wt{N}}{N}.\yesnum
    \end{align*}
    Thus, if $\frac{ \wt{N}}{N}>1-\frac{\eps}{2}$ with high probability for all $\sigma\in \zo^p$, then for all $\sigma\in \zo^2,$ it holds that $$\frac{\sum_{i\in I_\sigma} \tx_i w_i}{\sum_{i\in I_\sigma} \sx_i w_i} > 1-\frac{\eps}{2}.$$
    Towards this, we prove that for all $\sigma\in \zo^p$, $N_\sigma$ is concentrated around $|I_\sigma|\cdot \frac{n}{m}$:
    \begin{lemma}\label{lem:overview:1}
      For any fixed $\sigma\in \zo^p$ and $\Delta\geq 2$, it holds that 
     \begin{align*}
          \text{$\Ex[N_\sigma ] = \abs{I_\sigma}\cdot \frac{n}{m}$ and $\Pr\nolimits_{w }\insquare{ N_{\sigma} > \Ex[N_\sigma] + \Delta }\leq e^{-\frac{\Delta^2}{n}}.$}
      \end{align*}
    \end{lemma}
    \noindent Given \cref{lem:overview:1}, an obvious strategy is to set $L_\sigma=|I_\sigma|\cdot \frac{n}{m}$ for all $\sigma\in \zo^p$.
    However, this does not work because if $|I_\sigma|$ is small, then the concentration bound in \cref{lem:overview:1} is weak.
    We overcome this by setting slightly larger bounds for small intersections and slightly smaller bounds for big intersections:
    \begin{align*}
      \text{for all } \sigma \in \zo^p,\qquad L_\sigma \coloneqq \frac{\abs{I_\sigma}}{m}\cdot n\cdot (1-\eps)+2^{-p}{n\eps},\yesnum
    \end{align*}
    \noindent where if RHS is larger than $|I_\sigma|$, then we set $L_\sigma = |I_\sigma|$.
    Using \cref{lem:overview:1}, one can show that these constraints suffice.

\section{Proofs}
In this section, we present the proofs of our results.
For completeness, we restate some lemmas from \cref{sec:overview_of_proofs}.

\subsection{Preliminaries}\label{sec:prelim}

\noindent For reals $x,y,z\in \R$, let $x\in y\pm z$ denote $x\in [y-z,y+z].$
For a constant $c>0$ and parameter $x\in \R$, we use $O_c(x)$ and $\Omega_c(x)$ to denote $O(c^{-1}x)$ and $\Omega(cx)$ respectively.
For all $k\in [m]$, let $U_{(k:m)}$ be the $k$-th order statistic from $m$ independent draws from $\unif$, i.e., $U_{(k:m)}$ is the $k$-th smallest value from $m$ independent draws.

\begin{fact}[{\protect\cite[Eqs. 8.2, 8.8]{ahsanullah2013introduction}}]\label{fact:os}
  For all $k\in [m]$, it holds that
  \begin{align*}
      \text{$\Ex\insquare{U_{(k:m)}} = \frac{k}{m+1}$ and $\mathrm{Var}\insquare{U_{(k:m)}} = \frac{k(m-k+1)}{(m+1)^2(m+2)}\leq \frac{1}{m+2}.$}
  \end{align*}
\end{fact}
\noindent By using linearity of expectation and \cref{fact:os} to compute $\Ex\insquare{\sum_{i=1}^{k} U_{(m-i+1:m)}}$, we get the following fact.
\begin{fact}\label{fact:sum_n}
  For all $k\in [m]$, it holds that
  $$\Ex\insquare{\sum_{i=1}^{k} U_{(m-i+1:m)}} = k\cdot \inparen{1-\frac{k+1}{2(m+1)}}.$$
\end{fact}

\begin{fact}[\protect{\cite[Example 5.1]{ahsanullah2013introduction}}]\label{fact:conditional_os}
  Let $U_{(1:m)},\dots,U_{(m:m)}$ and $V_{(1:m)},\dots,V_{(m:m)}$ be two sets of order statistics of the uniform distribution.
  For any fixed $k\in [m]$ and threshold $t\in (0,1)$,
  conditioned on the event $U_{(k:m)}\negsp =\negsp t$,
  \begin{itemize}[leftmargin=10pt]
    \item for all $\ell=k+1,k+2,\dots,m$, the distribution of $U_{(k+\ell:m)}$ is the same as the distribution of $t + (1-t) V_{(\ell:m-k)}$,
    \item for all $\ell=1,2,\dots,k-1$, the distribution of $U_{(\ell:m)}$ is the same as the distribution of $t V_{(\ell:k-1)}.$
  \end{itemize}
\end{fact}

\begin{fact}\label{fact:whp}
  Let $X$ be a  random variable with support on $[0,b]$ and $\evE,\evF$ be some events,
  then
  \begin{align*}
    \abs{\Ex[X] - \Ex[X\mid \evE]} \leq b\cdot(1-\Pr[\evE])
    \quad \text{and}\quad
    \abs{\Pr[\evF] - \Pr[\evF\mid \evE]} \leq 1-\Pr[\evE].
  \end{align*}
\end{fact}
\begin{proof}
  We have that
  \begin{align*}
    \Ex[X] & = \Ex[X\mid \evE] \Pr[\evE] + \Ex[X\mid \lnot\evE] (1-\Pr[\evE])\\
    &\leq \Ex[X\mid \evE] + b(1-\Pr[\evE]),\yesnum\label{eq:prelim:eq1}\\
    \Ex[X] & \geq \Ex[X\mid \evE] \Pr[\evE]\\
    &= \Ex[X\mid \evE]  - \Ex[X\mid \evE]\cdot (1-\Pr[\evE])\\
    &\geq \Ex[X\mid \evE]  - b(1-\Pr[\evE]).\yesnum\label{eq:prelim:eq2}
  \end{align*}
  Chaining inequalities in \cref{eq:prelim:eq1,eq:prelim:eq2}, we get $$\abs{\Ex[X] - \Ex[X\mid \evE]} \leq b\cdot(1-\Pr[\evE]).$$
  The equation $$\abs{\Pr[\evF] - \Pr[\evF\mid \evE]} \leq 1-\Pr[\evE],$$ follows by choosing $X'\coloneqq \mathbb{I}[\evF]$ and observing {that $X'\in [0,1]$, $\Pr[\evF]=\Ex[X']$ and $\Pr[\evF\mid \evE]=\Ex[X'\mid \evE].$}
\end{proof}
\begin{fact}\label{fact:sum_n_ub}
  For any set $S\subseteq [m]$, $U\in [0,n]$, and event $\evE$, if conditioned on $\evE$, $\tx$ selects at most $U$ items from $S$ (i.e., $\sum_{i\in S} \tx_i\leq U$), then
  \begin{align*}
    \Ex_w\insquare{\sum\nolimits_{i\in S} \tx_i w_i \given \evE}
    \leq U \cdot \inparen{ 1 - \frac{U}{2\abs{S}}  } + n\cdot \inparen{1-\Pr[\evE]}.
  \end{align*}
\end{fact}
\begin{proof}
  Let $v_1\geq v_2\geq \dots\geq v_{\abs{S}}$ be latent utilities of items in $\cup_{\sigma \in S} I_\sigma$.
  Here, $v_i$ is the $\inparen{\abs{S}-i}$-th order statistic of $\abs{S}$ draws from $\unif$.
  Hence, it is distributed as $U_{(\abs{S}-i:\abs{S})}$
  We have that
  \begin{align*}
    \Ex_w\insquare{\sum\nolimits_{i\in S} \tx_i w_i \given \evE}
    \quad
    &\leq\quad \Ex_w\insquare{\sum\nolimits_{i=1}^U v_i \given \evE}
    \tag{Using that $v_1\geq v_2\geq \dots\geq v_{\abs{S}}$ and conditioned on $\evE$, $\sum_{i\in S} \tx_i \leq U$}\\
    &\leq\quad \Ex_w\insquare{\sum\nolimits_{i=1}^U v_i} + n\cdot \inparen{1-\Pr[\evE]}
    \tag{Using \cref{fact:whp} and that $\sum\nolimits_{i=1}^U v_i\leq n$ as $U\leq n$}\\
    &=\quad U \cdot \inparen{ 1 - \frac{2U}{m \cdot \abs{S}}  } + n\cdot \inparen{1-\Pr[\evE]}.
    \tag{Using \cref{fact:sum_n}}
  \end{align*}
\end{proof}

\subsection{Proof of \cref{thm:ub}}\label{sec:proofof:thm:ub}
Recall that utility ratio is $$\ratio_{\unif}(L,\beta)\coloneqq \Ex_{w}\insquare{ \frac{ \sinangle{\tx,w} }{ \sinangle{\sx,w} } },$$
where $\tx$ is the selection picked by the constrained decision maker (as defined in Equation~\eqref{eq:css}) and $\sx$ is the selection with the highest latent utility (as defined in Equation~\eqref{eq:def_xstar}).
We are given the implicit bias parameters $0<\beta_1,\beta_2\leq 1$ and the guarantee that for some $\eta>0$ and $\rho>0$,
\begin{align*}
  \forall \sigma \in \zo^2,\quad
  \rho <\frac{\abs{I_\sigma}}{m} < 1-\rho
  \quad\text{and}\quad 
  \eta<\frac{n}{m}<1-\eta.
  \yesnum\label{asmp:1:app}
\end{align*}
Given $\eta>0$ and $\rho>0$, define
\begin{align*}
  \phi(\beta) \coloneqq  \inparen{\frac{\rho}{2\sqrt{2}}\cdot \min\inbrace{\eta,1-\eta}\cdot (1-\beta_1) \cdot(1-\beta_2)}^2.
\end{align*}
Our goal is to show that
\begin{align*}
  \max_{L_1,L_2\in \Z_{\geq 0}}\ratio_{\unif}(L,\beta)\leq
  1-\phi(\beta).\yesnum
\end{align*}
At a high level, our strategy is to compute $\Ex_w\insquare{\sinangle{\sx,w}}$ and $\Ex_w\insquare{\sinangle{\tx,w}}$. %
To do this, we give two optimization programs such that $\sinangle{\sx,w}$ and $\sinangle{\tx,w}$ are concentrated around the values of these programs.
We show that the value of the second program is strictly smaller than $\inparen{1-\phi(\beta)}$-times the value of the first program.
This suffices to prove the result because $\Ex_w\insquare{\sinangle{\tx,w}}$ is concentrated around the value of the second program and $\Ex_w\insquare{\sinangle{\sx,w}}$ is concentrated around the value of the first program.
More precisely, we divide the proof into four steps.

\smallskip\smallskip 

\paragraph{Step 1: Reduce computing \ $\sx$\ and\ $\tx$\ to a small number of variables.} If $\tx$ selects $k_\sigma$ items from $I_\sigma$ (for some $\sigma\in \zo^p$), then these are exactly the $k_\sigma$ items with the largest observed utility in $I_\sigma$.
Otherwise, we can swap some selected item $i\in I_\sigma$ with some unselected item with a higher observed utility in $I_\sigma$.
This swap would increase the observed utility without violating the lower-bound constraints.
Further, a property of the implicit bias model in Equation~\eqref{eq:mult_bias} is that: For any two items in the same intersection, the order of their latent utilities is the same as the order of their observed utilities.
Combined with the previous observation, this shows that: %
\begin{observation}\label{obs:wtx_picks_top_latent_util}
  Fix any $\sigma\in \zo^p$.
  If $\tx$ selects $k_\sigma\in \N$ items from $I_\sigma$, then these are exactly the $k_\sigma$ items with the largest latent utility in $I_\sigma$.
\end{observation}
\noindent Similarly, if $\sx$ selects $k_\sigma$ items from $I_\sigma$ (for some $\sigma\in \zo^p$), then these are exactly the $k_\sigma$ items with the largest latent utility.
Otherwise, we can swap some selected item $i\in I_\sigma$ with {some unselected item with a higher observed utility in $I_\sigma$.}
\begin{observation}\label{obs:sx_picks_top_latent_util}
  Fix any $\sigma\in \zo^p$.
  If $\sx$ selects $k_\sigma\in \N$ items from $I_\sigma$, then these are exactly the $k_\sigma$ items with the largest latent utility in $I_\sigma$.
\end{observation}
\noindent Hence, given the utilities $w$, to determine $\sx$ and $\tx$, it suffices to compute
\begin{align*}
  \forall \sigma\in \zo^p,\quad K^\star_\sigma\coloneqq \sum_{i\in I_\sigma} \sx_i \quad \text{and}\quad \wt{K}_\sigma\coloneqq \sum_{i\in I_\sigma} \tx_i.
\end{align*}
\noindent Henceforth, we fix $p\coloneqq 2$.
For each $\sigma\in \zo^2$, let $$\beta_\sigma \coloneqq \prod_{\ell\in [p]\colon \sigma_\ell=1}\beta_\ell.$$

\paragraph{Step 2: Compute ${\sinangle{\tx,w}}$ and ${\sinangle{\tx,\hw}}$ as function of $K$.}

\noindent  In this step we show that $K^\star$ and $\wt{K}$ are also sufficient to determine the utilities of $\sx$ and $\tx$ respectively:
\begin{lemma}\label{lem:opt_solution:app}
  Let $x$ be any selection that, for all $\sigma\in \zo^2$, selects top $K_\sigma$ items from $I_\sigma$ by observed utility.
  Where $x$ and $K$ are possibly random variables.
  With probability at least $1-O_\rho\sinparen{m^{-\frac{1}{4}}}$,
  \begin{align*}
    \sinangle{x, w} = f_1(K) \pm O_{\eta\rho}\sinparen{nm^{-\frac{1}{4}}}
    \quad\text{and}\quad
    \sinangle{x,\hw} = f_\beta(K) \pm O_{\eta\rho}\sinparen{nm^{-\frac{1}{4}}},
  \end{align*}
  where $f_1$ and $f_\beta$ are defined as follows:
  \begin{align*}
    \forall
    \gamma\in [0,1]^4\quad\text{and}\quad k\in \R_{\geq 0}^4,\quad
    f_\gamma(k)
    &\coloneqq \sum_\sigma \gamma_\sigma \abs{I_\sigma} \cdot \int_{1-\frac{k_\sigma}{\abs{I_\sigma}}}^1  x dx.
    \yesnum\label{eq:def_g:app:tmp}
  \end{align*}
\end{lemma}
\noindent Here, the function $f_\gamma$ depends on both the uniform distribution of utility and the multiplicative implicit bias model in Equation~\eqref{eq:mult_bias}.
We discuss other distributions and implicit bias models in \cref{sec:proofof:thm:gen_ub}.
An important property of $f_\gamma$ is that it is strongly concave, Lipschitz, and Lipschitz continuous:\footnote{A function $f\colon \R^n\to \R$ is said to be $L$-Lipschitz if for all $x,y\in \R$, $\abs{f(x)-f(y)}_2\leq L\norm{x-y}_2$
and $M$-Lipschitz continuous if for all $x,y\in \R$, $\norm{\nabla f(x)-\nabla f(y)}_2\leq M\norm{x-y}_2$.}
\begin{lemma}\label{lem:sc_lp}
  For all $\gamma\in [0,1]^4$, $f_\gamma$ is $\frac{\min_\sigma\gamma_\sigma}{(1-\rho)m}$-strongly concave, $\norm{\gamma}_2$-Lipschitz, and $\frac{1}{\rho m}$-Lipschitz continuous.
\end{lemma}
\noindent \cref{lem:sc_lp} is a special case of \cref{lem:sc_lp:gen}. The proof of \cref{lem:sc_lp:gen} appears in \cref{sec:proofof:lem:sc_lp:gen}.

  \paragraph{\bf Step 3: Express $K^\star$ and $\wt{K}$ as solutions of different optimization problems.}
  Using the strong concavity of $f_\gamma$ and \cref{lem:opt_solution:app}, in this step, we show that $K^\star$ and $\wt{K}$ are concentrated around the optimizers of the following optimization problems:\smallskip

  \smallskip

  {\begin{minipage}{.5\textwidth}
  \begin{tcolorbox}[enhanced,colback=white,grow to left by=-0.5cm,grow to right by=-0.5cm,bottom=0cm,top=-0.25cm]
    \begin{align*}
      &\argmax\nolimits_k \quad f_1(k),\yesnum\label{prog:1:app}\\
      &\qquad \st,\quad\textstyle \sum_\sigma k_\sigma=n,\\
      &\qquad\qquad\ \ \forall\sigma\in \zo^2,\quad  0\leq k_\sigma \leq \abs{I_\sigma}.
      \\\white{.}\\\white{.}
    \end{align*}
  \end{tcolorbox}
  \end{minipage}}%
  {\begin{minipage}{.5\textwidth}
  \begin{tcolorbox}[enhanced,colback=white,grow to left by=-0.4cm,grow to right by=-0.5cm,bottom=0cm,top=-0.25cm,left=-2pt]
    \begin{align*}
      &\argmax\nolimits_k \quad f_{\beta}(k),\yesnum\label{prog:2:app}\\
      &\qquad \st,\quad k_{10}+k_{11}\geq L_1,\yesnum\label{eq:prog2:lb1}\\
      &\qquad\qquad\ \ k_{01}+k_{11}\geq L_2,\yesnum\label{eq:prog2:lb2}\\
      &\qquad\qquad\ \ \textstyle \sum_\sigma k_\sigma=n,\\
      &\qquad\qquad\ \forall\sigma\in \zo^2,\  0\negsp\leq\negsp k_\sigma \negsp\leq\negsp \abs{I_\sigma}.
      \yesnum\label{eq:prog:2:grp_sz_constraint:app}
    \end{align*}
  \end{tcolorbox}
\end{minipage}}

\smallskip
\noindent In particular, we prove the following lemma:
\begin{lemma}\label{lem:conc_of_k:app}
  Let $s^\star,\wt{s}$ be the optimizers of Programs~\eqref{prog:1:app} and \eqref{prog:2:app} respectively.
  With probability at least $1-O_\rho\sinparen{m^{-\frac{1}{4}}}$, %
  \begin{align*}
    \forall\sigma\in \zo^2,\quad
    \abs{K^\star_\sigma - s^\star_\sigma} \leq {O_{\eta\rho}(nm^{-\frac{1}{8}})}
    \quad\text{and}\quad
    \abs{\wt{K}_\sigma - \wt{s}_\sigma} \leq {O_{\beta\eta\rho}(nm^{-\frac{1}{8}})}.
  \end{align*}
\end{lemma}
\noindent
Let $\evE$ be the event that the following hold:
\begin{align*}
  \sinangle{\sx, w} &= f_1(K^\star) \pm O_{\eta\rho}\sinparen{nm^{-\frac{1}{4}}},\quad\text{\white{and}}\quad
  \sinangle{\sx,\hw} = f_\beta(K^\star) \pm O_{\eta\rho}\sinparen{nm^{-\frac{1}{4}}},
  \yesnum\label{eq:def_evE1}\\
  \sinangle{\tx, w} &= f_1(\wt{K}) \pm O_{\eta\rho}\sinparen{nm^{-\frac{1}{4}}},\quad\ \ \hspace{0.5mm} \text{and}\quad\ \
  \sinangle{\tx,\hw} = f_\beta(\wt{K}) \pm O_{\eta\rho}\sinparen{nm^{-\frac{1}{4}}}.
  \yesnum\label{eq:def_evE2}
\end{align*}
From \cref{lem:opt_solution:app} we get that
\begin{align*}
  \Pr\insquare{\evE}\geq 1-O_\rho\sinparen{m^{-\frac{1}{4}}}.
  \yesnum\label{eq:prob_evE}
\end{align*}
Conditioned on $\evE$, \cref{eq:def_evE1,eq:def_evE2} hold, and hence:
\begin{align*}
  \abs{\sinangle{\sx, w}-f_1(s^\star)}
  &\leq \abs{f_1(K^\star) - f_1(s^\star) + O\sinparen{nm^{-\frac{1}{4}}}}.\\
  \intertext{Since $\norm{(1,1,1,1)}_2=2$, \cref{lem:sc_lp} implies that $f_1$ is $2$-Lipschitz, and hence}
  \abs{\Ex\insquare{\sinangle{\sx, w}}-f_1(s^\star)}
  &\leq 2\cdot \norm{K^\star-s^\star}_2 + {O_{\eta\rho}(nm^{-\frac{1}{8}})}\\
  &\leq {O_{\eta\rho}(nm^{-\frac{1}{8}})}.
  \tagnum{Using \cref{lem:conc_of_k:app}}\customlabel{eq:exp_in_k}{\theequation}
\end{align*}
Replacing $\sx$, $K^\star$, and $s^\star$ by $\tx$ in the above argument, by $\wt{K}$, and $\wt{s}$ we get that
\begin{align*}
  \abs{\Ex\insquare{\sinangle{\tx, w}}-f_1(\wt{s})}
  \leq {O_{\beta\eta\rho}(nm^{-\frac{1}{8}})}.
  \yesnum\label{eq:exp_in_k2}
\end{align*}

\smallskip
\noindent {\bf Step 4: Proof of Theorem~\ref{thm:ub}}

\noindent This step relies on the following lemma:
\begin{lemma}\label{lem:using_sc}
  Given $L_1,L_2\in \Z_{\geq 0}$, if $\snorm{\wt{s}-s^\star}_2\geq 2nm(1-\rho)\phi(\beta)$, then
  \begin{align*}
    \ratio_{\unif}(L,\beta) \leq 1 - \inparen{\frac{\rho}{3}\cdot \min\inbrace{\eta,1-\eta}\cdot (1-\beta_1) \cdot(1-\beta_2)}^2.
  \end{align*}
\end{lemma}

\noindent Fix any $L_1,L_2\in \Z_{\geq 0}$.
Using the gradient test, we can analytically solve \prog{prog:1:app} to get
\begin{align*}
  s^\star\coloneqq \inbrace{\abs{I_\sigma}\cdot \frac{n}{m}}_{\sigma\in \zo^2}.
  \yesnum\label{eq:val_k_star:app}
\end{align*}
We consider two cases depending on whether $\wt{s}$ satisfies any inequality in Equation~\eqref{eq:prog:2:grp_sz_constraint:app} with equality.
In either case, we prove that
\begin{align*}
  \snorm{s^\star-\wt{s}}_2^2\geq 2nm(1-\rho)\cdot \phi(\beta).
  \yesnum\label{eq:lb_on_difference:app}
\end{align*}
Since \cref{eq:lb_on_difference:app} holds, we can use \cref{lem:using_sc} to conclude that
\begin{align*}
  \ratio_{\unif}(L,\beta) \leq 1 - \inparen{\frac{\rho}{3}\cdot \min\inbrace{\eta,1-\eta}\cdot (1-\beta_1) \cdot(1-\beta_2)}^2.
\end{align*}
Finally, as the choice of $L_1,L_2\in \Z_{\geq 0}$ was arbitrary, \cref{thm:ub} follows.

\paragraph{\bf (Case A) $\wt{s}$ satisfies at least one inequality in Equation~\eqref{eq:prog:2:grp_sz_constraint:app} with equality:}
We further divide this into two cases:

\indent {\bf (Case A.1) $\exists\sigma\in \zo^2$, such that $\wt{s}_\sigma=0$:}
In this case, it holds that
\begin{align*}
  \snorm{s^\star-\wt{s}}_2^2
  \ \ &\Stackrel{\eqref{eq:val_k_star:app}}{\geq} \ \ \abs{\abs{I_\sigma}\frac{n}{m}}^2\\
  &> \ \  \eta^2\cdot \abs{I_\sigma}^2.\tag{Using that $n>\eta \cdot m $ due to \cref{asmp:1:app}}
\end{align*}

\indent {\bf (Case A.2) $\exists\sigma\in \zo^2$, such that $\wt{s}_\sigma=\abs{I_\sigma}$:}
In this case, it holds that
\begin{align*}
  \snorm{s^\star-\wt{s}}_2^2
  \ \ &\Stackrel{\eqref{eq:val_k_star:app}}{\geq} \ \ \abs{\abs{I_\sigma}\inparen{1-\frac{n}{m}}}^2\\
  &>\ \  (1-\eta)^2\cdot \abs{I_\sigma}^2.\tag{Using that $n<(1-\eta) \cdot m$ due to \cref{asmp:1:app}}
\end{align*}

\noindent Thus, in either case,
\begin{align*}
  \snorm{s^\star-\wt{s}}_2^2
  \geq \inparen{\rho m\max\inbrace{\eta, (1-\eta)}}^2.
  \yesnum\label{eq:case_anan:ub}
\end{align*}
Since $\phi(\beta)<\frac{\rho^2}{2\eta}\cdot \min\inbrace{\eta^2, (1-\eta)^2}$, Equation~\eqref{eq:lb_on_difference:app} follows.

\paragraph{\bf (Case B) $\wt{s}$ satisfies all inequalities in Equation~\eqref{eq:prog:2:grp_sz_constraint:app} with strict inequality:}
Since apart from Equation~\eqref{eq:prog:2:grp_sz_constraint}, \prog{prog:2:app} has only three other constraints, while $\wt{s}$ has four coordinates, it follows that there is some constant $t_0>0$ and a vector, namely $v\coloneqq (1,-1,-1,1)$ such that for all $-t_0\leq t\leq t_0$, $\wt{s}+tv$ is feasible for \prog{prog:2:app}.
Since $\wt{s}$ is the optimal solution of \prog{prog:2:app}, this implies that
\begin{align*}
  \inangle{\nabla f_\beta(\wt{s}), v} = 0.
  \yesnum\label{eq:zero_grad:app}
  \end{align*}
  Using the value of $s^\star$, we have
  \begin{align*}
    \inangle{\nabla f_\beta(s^\star), v}
    =\inparen{1-\frac{n}{m}}\cdot \inparen{1-\beta_1}\cdot \inparen{1-\beta_2}.
    \yesnum\label{eq:exp_grad:app}
  \end{align*}
  This is sufficient to show that, in this case, $\snorm{\nabla f_\beta(\wt{s}) - \nabla f_\beta(s^\star)}_2$ is large:
  \begin{align*}
    \snorm{\nabla f_\beta(\wt{s}) - \nabla f_\beta(s^\star)}_2
    &\geq \frac{1}{\norm{v}_2}\cdot \abs{\inangle{\nabla f_\beta(\wt{s}) - \nabla f_\beta(s^\star), v}}
    \tag{Using Cauchy-Schwarz inequality}\\
    &= \frac{1}{2}\inparen{1-\frac{n}{m}}\cdot \inparen{1-\beta_1}\cdot \inparen{1-\beta_2}.
    \tagnum{Using Equations~\eqref{eq:zero_grad} and \eqref{eq:exp_grad}}
    \customlabel{eq:lb_on_grad:app}{\theequation}
  \end{align*}
  Using \cref{lem:sc_lp}, in particular that $f_\beta$ is $\frac{1}{\rho m}$-Lipschitz continuous, this implies that
  \begin{align*}
    \snorm{s^\star-\wt{s}}_2^2
    \ \ &\geq\ \  \rho^2 m^2\cdot \snorm{\nabla f_\beta(\wt{s}) - \nabla f_\beta(s^\star)}_2
    \yesnum\label{eq:lip_const:ub}\\
    &\Stackrel{\eqref{eq:lb_on_grad:app}}{\geq}\ \  \frac{\rho^2 m^2}{4}\cdot \inparen{1-\frac{n}{m}}^2\cdot \inparen{1-\beta_1}^2\cdot \inparen{1-\beta_2}^2\\
    &\geq\ \  2mn(1-\rho)\cdot \phi(\beta).
  \end{align*}
  Thus, Equation~\eqref{eq:lb_on_difference:app} holds in this case also.

  \subsubsection{Proof of \cref{lem:opt_solution:app}}

  \begin{proof}[Proof of \cref{lem:opt_solution:app}]
    For each $\sigma\in \zo^2$, let $v^{\sexp{\sigma}}_1\geq v^{\sexp{\sigma}}_2\geq \dots\geq v^{\sexp{\sigma}}_{\abs{I_\sigma}}$ be latent utilities of items in $I_\sigma$.
    Here, $v^{\sexp{\sigma}}_i$ is the $\inparen{\abs{I_\sigma}-i}$-th order statistic of $\abs{I_\sigma}$ draws from $\unif$.
    Hence, it is distributed as $U_{(\abs{I_\sigma}-i:\abs{I_\sigma})}$
    From \cref{fact:os}, we have that for all $\sigma\in \zo^2$ and $i\in I_\sigma$,
    $$\mathrm{Var}\insquare{v^{\sexp{\sigma}}_i} = \mathrm{Var}\insquare{U_{(\abs{I_\sigma}-i:\abs{I_\sigma})}} \leq \frac{1}{\abs{I_\sigma}}.$$
    Since $\abs{I_\sigma}\geq \rho m$,
    $$\mathrm{Var}\insquare{v^{\sexp{\sigma}}_i} \leq \frac{1}{\rho m}.$$
    Using Chebyshev's inequality~\cite{motwani1995randomized}, we get:
    \begin{align*}
      \forall\ \sigma\in \zo^2,\ i\in I_\sigma,\quad \Pr\insquare{  \abs{v^{\sexp{\sigma}}_i - \Ex[v^{\sexp{\sigma}}_i]}
      \geq m^{-\frac14}  } \leq \rho^{-\frac{1}{2}}m^{-\frac12}. %
      \yesnum\label{eq:vi_conc:new}
    \end{align*}
    \noindent For each $i\in \N$, define
    \begin{align*}
      s(i)\coloneqq (i-1)\cdot m^{\frac34} + 1.
    \end{align*}
    For each intersection $\sigma\in \zo^2$, starting from the first order statistic, consider order statistics at intervals of $m^{\frac34}$:
    $$v^{\sexp{\sigma}}_{s(1)},\ \ v^{\sexp{\sigma}}_{s(2)},\ \dots.$$
    Let $\evG$ be the event that all of these order statistics are $m^{-\frac14}$-close to their mean:
    \begin{align*}
      \forall \sigma\in \zo^2,\ i\in \inbrace{1,2,\dots,\floor{\smash{m^{-\frac{3}{4}}\cdot\abs{I_\sigma}}}},\quad
      \abs{v^{\sexp{\sigma}}_{s(i)} - \Ex[v^{\sexp{\sigma}}_{s(i)}]} < m^{-\frac14}.
      \yesnum\label{eq:def_evG}
    \end{align*}
    Taking the union bound over all pairs of $\sigma\in \zo^2$ and $i$, by using Equation~\eqref{eq:vi_conc:new}, we get that
    \begin{align*}
      \Pr[\evG]\geq 1-4 m^{\frac14}\cdot \rho^{-\frac{1}{2}} m^{-\frac12}  = 1-O_\rho(m^{-\frac14}).\yesnum\label{eq:whpcg}
    \end{align*}

    \noindent Conditioned on $\evG$, we have
    \begin{align*}
      \sum_{\sigma\in \zo^2}{\sum_{i=1}^{K_{\sigma}} \gamma_\sigma v_j^{\sigma}}
      \quad
      &\leq\quad  \sum_{\sigma\in \zo^2} {\sum_{j=1}^{m^{-\frac34} K_{\sigma}} m^{\frac34}\cdot \gamma_\sigma v_{s(j)}^{\sigma}}
      \tag{$v_{i}^{\sigma}$ are arranged in non-increasing order}\\
      &\leq\quad \sum_{\sigma\in \zo^2} {\sum_{j=1}^{m^{-\frac34} K_{\sigma}} m^{\frac34}\cdot \gamma_\sigma\cdot \inparen{\Ex\insquare{v_{s(j)}^{\sigma}} + m^{-\frac14}} }
      \tag{Using that $\evG$ implies \cref{eq:def_evG}}\\
      &\leq\quad \sum_{\sigma\in \zo^2} {\sum_{i=1}^{K_{\sigma}} \gamma_\sigma\cdot \inparen{\Ex\insquare{v_{i}^{\sigma}} + \frac{1}{\abs{I_\sigma}}\cdot O\sinparen{m^{\frac14}}  + m^{-\frac14}} }
      \tag{Using that for any $\abs{i-j}\leq k$, $\abs{\Ex\sinsquare{v_i^\sigma}-\Ex\sinsquare{v_j^\sigma}}\leq \frac{1}{\abs{I_\sigma}}\cdot O\sinparen{k}$; see \cref{fact:os}}\\
      &=\quad \sum_{\sigma\in \zo^2} \sum_{i=1}^{K_{\sigma}} \gamma_\sigma \Ex\sinsquare{v_{i}^{\sigma}}
      + O_\rho(nm^{-\frac14}).
      \tag{Using that $\abs{I_\sigma}\geq \rho m$ (due to \cref{asmp:1:app}), $\gamma_\sigma\leq 1$, and $K_\sigma\leq n$}
      \intertext{From \cref{fact:os}: $\Ex\sinsquare{v_{i}^{\sigma}} = \frac{i}{\abs{I_\sigma}+1}\leq \frac{i}{\abs{I_\sigma}}$.
      Hence, $\sum_{i=1}^{K_\sigma}\Ex\sinsquare{v_{i}^{\sigma}}\leq \sum_{i=1}^{K_\sigma}\frac{i}{\abs{I_\sigma}}\leq \int_{1-\frac{K_\sigma}{\abs{I_\sigma}}}^1 \inparen{x+\frac{1}{\rho m}}dx$. This gives us:}
      \sum_{\sigma\in \zo^2}{\sum_{i=1}^{K_{\sigma}} \gamma_\sigma v_j^{\sigma}}
      \quad
      &\leq \quad \sum_{\sigma\in \zo^2} \int_{1-\frac{K_\sigma}{\abs{I_\sigma}}}^1 \inparen{x+\frac1m} dx + O_\rho(nm^{-\frac14})\\
      &\leq\quad f_\gamma(K) + O_{\eta\rho}(nm^{-\frac14}).
      \tagnum{Using that $n\geq \eta m$ (due to \cref{asmp:1:app}) and $K_\sigma\leq n$, hence $O_\rho(nm^{-\frac14})\gg \frac{K_\sigma}{m}$}
      \customlabel{eq:bound_on_obs_util:new}{\theequation}
    \end{align*}
    Similarly, we can show the following lower bound on $\sinangle{\tx,\hw}$ conditioned on $\evG$
    \begin{align*}
      \sum_{\sigma\in \zo^2}{\sum_{i=1}^{K_{\sigma}} \gamma_\sigma v_j^{\sigma}}
      &\geq  f_\gamma(K) - O_{\eta\rho}(nm^{-\frac14}).\yesnum\label{eq:bound_on_obs_util2:new}
    \end{align*}
    Substituting $\gamma=\beta$ in Equations~\eqref{eq:bound_on_obs_util:new} and \eqref{eq:bound_on_obs_util2:new}, and using that for all intersections $\sigma\in \zo^2$, $x$ selects top $K_\sigma$ items from $I_\sigma$ by observed utility, we get that conditioned on $\evG$
    \begin{align*}
      \sinangle{x,\hw} = \sum_{\sigma\in \zo^2}{\sum_{i=1}^{K_{\sigma}} \beta_\sigma v_j^{\sigma}} = f_\beta(K) \pm O_{\eta\rho}(nm^{-\frac14}).
      \end{align*}
      Further, substituting $\gamma=(1,1,1,1)$ in Equations~\eqref{eq:bound_on_obs_util:new} and \eqref{eq:bound_on_obs_util2:new}, and using that for all intersections $\sigma\in \zo^2$, $x$ selects top $K_\sigma$ items from $I_\sigma$ by observed utility, we get that conditioned on $\evG$
      \begin{align*}
        \sinangle{x, w} = \sum_{\sigma\in \zo^2}{\sum_{i=1}^{K_{\sigma}} v_j^{\sigma}} = f_1(K) \pm O_{\eta\rho}(nm^{-\frac14}).
      \end{align*}
    \end{proof}

    \subsubsection{Proof of \cref{lem:conc_of_k:app}}

    \begin{proof}[Proof of \cref{lem:conc_of_k:app}]
      Let $\wt{s}_R\in \R^4_{\geq 0}$ be the vector that rounds each coordinate of $\wt{s}$ to the nearest integer.
      We have
      \begin{align*}
        \norm{\wt{s} - \wt{s}_R}_2 \leq 2.
        \yesnum\label{eq:conc:1}
      \end{align*}
      Since $\norm{\beta}_2\leq 2$, from \cref{lem:sc_lp} we get that $f_\beta$ is $2$-Lipschitz.
      Combined with \cref{eq:conc:1}, this implies that
      \begin{align*}
        \abs{f(\wt{s}) - f(\wt{s}_R)}\leq 2\norm{\wt{s} - \wt{s}_R}_2\leq 4.
        \yesnum\label{eq:conc:3}
      \end{align*}
      Consider the selection $\tx_R$ that, for all $\sigma\in \zo^2$, selects top $K_\sigma$ items from $I_\sigma$ by observed utility.
      Let $\evF$ be the event that the following hold:
      \begin{align*}
        \sinangle{\tx_R, w} &= f_1(\wt{s}_R) \pm O_{\eta\rho}\sinparen{nm^{-\frac{1}{4}}},\quad\text{\white{and}}\quad
        \sinangle{\tx_R,\hw} = f_\beta(\wt{s}_R) \pm O_{\eta\rho}\sinparen{nm^{-\frac{1}{4}}},
        \yesnum\label{eq:def_evF1}\\
        \sinangle{\sx, w} &= f_1(K^\star) \pm O_{\eta\rho}\sinparen{nm^{-\frac{1}{4}}},\quad\text{\white{and}}\quad
        \sinangle{\tx,\hw} = f_\beta(\wt{K}) \pm O_{\eta\rho}\sinparen{nm^{-\frac{1}{4}}}.
        \yesnum\label{eq:def_evF2}
      \end{align*}
      From \cref{lem:opt_solution:app}, we get that
      \begin{align*}
        \Pr\insquare{\evF}\geq 1-O_\rho\sinparen{m^{-\frac{1}{4}}}.
        \yesnum\label{eq:prob_evF}
      \end{align*}
      Since $\tx$ maximizes the observed utility, we have
      \begin{align*}
        \sinangle{\tx,\hw} \geq \sinangle{\tx_R,\hw}.
      \end{align*}
      Hence, conditioned on $\evF$:
      \begin{align*}
        f_\beta(\wt{K}) \pm O_{\eta\rho}\sinparen{nm^{-\frac{1}{4}}}
        \geq f_\beta(\wt{s}_R) \pm O_{\eta\rho}\sinparen{nm^{-\frac{1}{4}}}
        \ \ \Stackrel{\eqref{eq:conc:3}}{\geq}\ \ f_\beta(\wt{s}) \pm O_{\eta\rho}\sinparen{nm^{-\frac{1}{4}}}.
      \end{align*}
      This implies that: Conditioned on $\evF$
      \begin{align*}
        \abs{f_\beta(\wt{K}) - f_\beta(\wt{s})} \leq O_{\eta\rho}\sinparen{nm^{-\frac{1}{4}}}.
        \yesnum\label{eq:conc:2}
      \end{align*}
      From \cref{lem:sc_lp}, we get that $f_\beta$ is $\frac{\beta_1\beta_2}{(1-\rho)m}$-strongly concave.
      This implies that: Conditioned on $\evF$
      \begin{align*}
        \abs{f_\beta(\wt{K}) - f_\beta(\wt{s})} \geq \frac{\beta_1\beta_2}{2(1-\rho)m}\cdot \norm{\wt{K}-\wt{s}}_2^2.
      \end{align*}
      Chaining this inequality with \cref{eq:conc:2}, we get that: Conditioned on $\evF$
      \begin{align*}
        \norm{\wt{K}-\wt{s}}_2
        &\leq \sqrt{\frac{2(1-\rho)m}{\beta_1\beta_2}\cdot O_{\eta\rho}\sinparen{nm^{-\frac{1}{4}}}}
        = O_{\beta\eta\rho}\sinparen{nm^{-\frac{1}{8}}}. \tag{Using that $n=O_\eta(m)$ due to \cref{asmp:1:app}}
      \end{align*}
      Replacing $\wt{K}, \wt{s}, \hw,$ and $\gamma=\beta$ in the above argument by $K^\star, \sx, w$, and $\gamma=(1,1,1,1)$ we get that: Conditioned on $\evF$
      \begin{align*}
        \norm{K^\star-s^\star}_2
        &\leq O_{\beta\eta\rho}\sinparen{nm^{-\frac{1}{8}}}.
      \end{align*}
    \end{proof}

    \subsubsection{Proof of \cref{lem:using_sc}}

    \begin{proof}[Proof of \cref{lem:using_sc}]
      Fix any $L_1,L_2\in \Z_{\geq 0}$.
      Suppose that
      \begin{align*}
        \snorm{\wt{s}-s^\star}_2^2\geq 2nm(1-\rho)\cdot \phi(\beta).
        \yesnum\label{lem:lb_on_diff}
      \end{align*}
      Since the feasible region of \prog{prog:1:app} is a superset of the feasible region of \prog{prog:2:app} and $s^\star$ maximizes $f_1$ over the feasible region of \prog{prog:1:app}, it follows that
      \begin{align*}
        f_1(\wt{s})\leq f_1(s^\star).
        \yesnum\label{eq:sc:1}
      \end{align*}
      Further, using that $f_1$ is $\frac{1}{(1-\rho)m}$-strongly concave (\cref{lem:sc_lp}), we have
      \begin{align*}
        \abs{f_1(\wt{s}) - f_1(s^\star)}
        &\ \ \geq\ \  \frac{1}{2(1-\rho)m} \snorm{\wt{s}-s^\star}_2^2
        \ \ \geq\ \  n\phi(\beta).\tagnum{Using \cref{lem:lb_on_diff}}\customlabel{eq:sc:2}{\theequation}
        \end{align*}
        Hence, substituting \cref{eq:sc:1} in Equation~\eqref{eq:sc:2}, we get
        \begin{align*}
          f_1(\wt{s})
          &\ \ \leq\ \   f_1(s^\star) - n\phi(\beta)
          \ \ \leq\ \  f_1(s^\star)\cdot \inparen{1 - \phi(\beta)}. %
          \tag{Using that $f_1(s^\star)\leq \sum_{\sigma}s^\star_\sigma = n$ (due to \prog{prog:1:app})}
        \end{align*}
        Substituting $\Ex\insquare{\sinangle{\tx, w}}$ and $\Ex\insquare{\sinangle{\sx, w}}$ from Equations~\eqref{eq:exp_in_k} and \eqref{eq:exp_in_k2} respectively, we get that:
        Conditioned on $\evE$
        \begin{align*}
          \Ex\insquare{\sinangle{\tx, w}} - O_{\eta\rho}(nm^{-\frac{1}{8}})
          \leq \Ex\insquare{\sinangle{\sx, w}}\cdot \inparen{1 - \phi(\beta)} + O_{\beta\eta\rho}(nm^{-\frac{1}{8}}).
        \end{align*}
        Using that, conditioned on $\evE$, $\Ex\insquare{\sinangle{\sx, w}}\geq f_1(s^\star)-O_{\eta\rho}(nm^{-\frac{1}{8}}) = \Theta(n)$, we get that:
        Conditioned on $\evE$
        \begin{align*}
          \Ex\insquare{\sinangle{\tx, w}}
          \leq \Ex\insquare{\sinangle{\sx, w}}\cdot \inparen{1 - \phi(\beta) + O_{\beta\eta\rho}(m^{-\frac{1}{8}}) }.
          \yesnum\label{eq:small_ratio}
        \end{align*}
        Then, we have
        \begin{align*}
          \ratio_{\unif}(L,\beta)
          &\coloneqq \Ex_{w}\insquare{ \frac{ \sinangle{\tx,w} }{ \sinangle{\sx,w} } }\\
          &\leq \Ex_{w}\insquare{ \frac{ \sinangle{\tx,w} }{ \sinangle{\sx,w} } \given \evE} + 1-\Pr\insquare{\evE}
          \tag{Using \cref{fact:whp}}\\
          &= \Ex_{w}\insquare{ \frac{ \sinangle{\tx,w} }{ \sinangle{\sx,w} } \given \evE} + O(m^{-\frac{1}{4}})
          \tag{Using \cref{eq:prob_evE}}\\
          &= 1 - \phi(\beta) + O_{\beta\eta\rho}(m^{-\frac{1}{8}}) + O_\rho(m^{-\frac{1}{4}})
          \tag{Using that $\evE$ implies \cref{eq:small_ratio}}\\
          &= 1 - \inparen{\frac{\rho}{2\sqrt{2}}\cdot \min\inbrace{\eta,1-\eta}\cdot (1-\beta_1) \cdot(1-\beta_2)}^2 + O_{\beta\eta\rho}(m^{-\frac{1}{8}}).
          \yesnum\label{eq:proof:app:eq:final}
        \end{align*}
        Pick {$m_0\geq \Omega\inparen{\inparen{\rho\eta\beta_1\beta_2}^8}$.}
        Then, $O_{\beta\eta\rho}(m^{-\frac{1}{8}}) \leq \frac{1}{72} \inparen{\rho\cdot \min\inbrace{\eta,1-\eta}\cdot (1-\beta_1) \cdot(1-\beta_2)}^2$.
        Substituting this in \cref{eq:proof:app:eq:final}, we get
        \begin{align*}
          \ratio_{\unif}(L,\beta) \leq 1 - \inparen{\frac{\rho}{3}\cdot \min\inbrace{\eta,1-\eta}\cdot (1-\beta_1) \cdot(1-\beta_2)}^2.
        \end{align*}
      \end{proof}

      \subsection{Proof of Theorem~\ref{thm:fullrecovery}}\label{sec:proofof:thm:fullrecovery}
      In this section, we give a proof of \cref{thm:fullrecovery}.
      Recall that $\tx$ is defined as  $\sx\coloneqq \argmax_x \inangle{x,w}$,
      and $\sx{}$ is defined as $\tx\coloneqq \argmax_{x\in C(L)} \inangle{x,\hw}$.
      We first prove \cref{coro:high_ratio_intersectional} (presented below).
      Then, \cref{thm:fullrecovery} follows as its corollary.
      \begin{lemma}[\textbf{Intersectional interventions recover near-optimal utility}]\label{coro:high_ratio_intersectional}
        Suppose that $n=\eta\cdot m$ for some fixed $\eta>0$.
        For all $\eps>0$, $\eta>0$, and $p\in \N$,
        there exists a threshold $m_0\in \N$,
        such that for all $m\geq m_0$ and group structures $G_1,\dots,G_p\subseteq [m]$,
        there exist intersectional constraints $L\in \Z_{\geq 0}^{2^p}$ such that for any
        implicit bias parameters $\beta_1,\dots,\beta_p\in (0,1]$ and continuous distributions $\cD$, %
        the optimal constrained selection $\tx\coloneqq \argmax_{x\in C(L)} \inangle{x,\hw}$ satisfies $\inangle{\tx, w}\geq (1-\eps) \cdot \max\nolimits_{x}\inangle{x,w}$
        with probability at least $1-\exp\sinparen{-\frac{\eps^2n}{p 4^p}}$. {Where the probability is over the draws of latent utilities $w$.}
      \end{lemma}

      \begin{proof}[Proof of \cref{coro:high_ratio_intersectional}.]
        We show that with high probability, for each intersection $\sigma\in \zo^p$, it holds that:
        \begin{align*}
          \frac{\sum_{i\in I_\sigma} \tx_i w_i}{\sum_{i\in I_\sigma} \sx_i w_i} > 1-\eps. \yesnum\label{eq:bound_with_eps_0}
        \end{align*}
        Then, the claimed bound follows as
        \begin{align*}
          \frac{\sinangle{\tx,w}}{\sinangle{\sx,w}} = \frac{\sum_{\sigma\in \zo^p}\sum_{i\in I_\sigma} \tx_i w_i}{\sum_{\sigma\in \zo^p}\sum_{i\in I_\sigma} \sx_i w_i}
          \ \Stackrel{\eqref{eq:bound_with_eps_0}}{>}\ \frac{\sum_{\sigma\in \zo^p}
          \sum_{i\in I_\sigma} \sx_i w_i}{\sum_{\sigma\in \zo^p}\sum_{i\in I_\sigma} \sx_i w_i}\cdot (1-\eps) > 1-\eps. \yesnum\label{eq:ratio_conditioned_on_ce}
        \end{align*}
        \noindent At a high level, our strategy is to find constraints such that $\tx$ selects a similar number of items from each intersection as $\sx{}$.
        This suffices to prove the result due to a property of the implicit bias model:
        If $\tx$ picks $K$ items from any intersection $I_\sigma$, then these are the $K$ items with the highest latent utility in $I_\sigma$ (see \cref{obs:wtx_picks_top_latent_util}).

        \paragraph{(\bf Picking similar number of items $\implies$ similar utility).}
        To see why this is sufficient, let $v_1,v_2,\dots,v_{|I_\sigma|}$ be the latent utilities of items of $I_\sigma$ ordered in non-increasing order of latent utility (or equivalently of observed utility).
        Further, let $N\in \Z_{\geq 0}$ and $\wt{N}\in \Z_{\geq 0}$ be the random variables that count the number of items selected from from intersection $I_\sigma$ by $\sx$ and $\tx$ respectively: %
        $N \coloneqq \sum\nolimits_{i\in I_\sigma}\sx_i$ and $\wt{N} \coloneqq \sum\nolimits_{i\in I_\sigma}\tx_i.$
        Using the above observation, we know $\sx{}$ and $\tx$ select the items corresponding to the latent utilities $v_1,v_2,\dots,v_{N}$ and $v_1,v_2,\dots,v_{\wt{N}}$ respectively.
        Then, we have
        \begin{align*}
          \frac{\sum_{i\in I_\sigma} \tx_i w_i}{\sum_{i\in I_\sigma} \sx_i w_i}
          &\geq \frac{\sum_{j=1}^{\wt{N}}v_j  }{ \sum_{i\in I_\sigma} \sx_i w_i } %
          = \frac{\sum_{j=1}^{\wt{N}}v_j  }{ \sum_{j=1}^{N}v_j}
          = \frac{\sum_{j=1}^{\wt{N}}v_j  }{  \sum_{j=1}^{\wt{N}}v_j + \sum_{j=\wt{N} + 1}^{N}v_j }.
        \end{align*}
        Using that for all $c\geq 0$, $\frac{x}{c+x}$ is an increasing function of $x$, we have
        \begin{align*}
          \frac{\sum_{i\in I_\sigma} \tx_i w_i}{\sum_{i\in I_\sigma} \sx_i w_i} &\geq \frac{ \wt{N}\cdot v_{{\wt{N}}}  }
          {  \wt{N}\cdot v_{{\wt{N}}}  + \sum_{j=\wt{N} + 1}^{N}v_j }%
          \geq \frac{ \wt{N}\cdot v_{{\wt{N}}}  }{  \wt{N}\cdot v_{{\wt{N}}}  + (N-\wt{N})\cdot v_{{\wt{N}}} }
          = \frac{ \wt{N}}{N}.\yesnum\label{eq:lowerbound_on_ratio}
        \end{align*}
        Thus if $\frac{ \wt{N}}{N}>1-\eps$ with high probability (for all $\sigma\in \zo^p$), then Equation~\eqref{eq:bound_with_eps_0} holds.

        \paragraph{\bf (A property of the implicit bias model).}
        Next, we discuss why the observation about the implicit bias model holds.
        Clearly, by the optimality of $\tx$, if it picks $K$ items from any intersection $I_\sigma$, then these are the $K$ items with the highest {\em observed utility} in $I_\sigma$.
        The observation holds because under the model of implicit bias in Equation~\eqref{eq:mult_bias},
        for two items in the same intersection, the order of their observed utility is the same as the order of their latent utility.
        To see this, note that for any two items $i,j\in I_\sigma$, $w_i>w_j$ if and only if $\hw_i = \beta_\sigma w_i > \beta_\sigma w_j = \hw_j$ because $\beta_\sigma>0$.
        (This is the only property of the implicit bias model we use in this proof.)

        It remains to pick constraints such that $\sx{}$ and $\tx$ pick a similar number of items from all intersections.
        For the remaining portion of the proof, fix an intersection $\sigma\in \zo^p.$
        Let $N_\sigma\in \Z_{\geq 0}$ be the random variables that counts the number of items selected by $\sx$ from intersection $I_\sigma$:
        $N_\sigma \coloneqq \sum\nolimits_{i\in I_\sigma}\sx_i$.
        The proof now proceeds in two parts:
        1) We estimate the value $N_\sigma$ with high probability.
        2) We choose the specific constraints, taking inspiration from this estimate.

        The following lemma shows that $N_\sigma$ the distribution of is concentrated.
        We give an overview of its proof here, and defer its full proof to \cref{sec:proofof:lem:1}.
        \begin{lemma}[\textbf{Estimate of $N_\sigma$ with high probability}]\label{lem:1}
          For any fixed $\sigma\in \zo^p$ and $\Delta\geq 2$, it holds that $\Ex[N_\sigma ] = |I_\sigma|\cdot \frac{n}{m}$ and $\Pr_{w }\insquare{ N_{\sigma} > \Ex[N_\sigma] + \Delta }\leq e^{-\frac{\Delta^2}{n}}.$
        \end{lemma}
        \begin{proof}[Proof overview.]
          Observe that $\sx{}$ picks the top $n$ items with the highest latent utility.
          The main idea is that, to count which items are selected, we only need to know the order of the items by latent utility.
          This allows us to consider an equivalent urn model, where items correspond to balls and the color of the balls denotes which intersection the item belongs to.
          Consider the process of picking $n$ balls from the urn.
          Here, $N_\sigma$ is equivalently defined as the number of balls of color $\sigma$ selected.
          It turns out that $N_\sigma$ is a hyper-geometric random variable.
          Using this fact, the result follows from standard properties of hyper-geometric random variables.
        \end{proof}
        \noindent Thus, the \cref{lem:1} shows that for all $\sigma\in \zo^p$, $N_\sigma$ is concentrated around $|I_\sigma|\cdot \frac{n}{m}$.
        The obvious strategy to ensure that $\tx$ picks a close to $N_\sigma$ items, might be to  set $L_\sigma$ to $|I_\sigma|\cdot \frac{n}{m}$ (for all $\sigma\in \zo^p$).
        However, this does not work.
        At a high level, this is because if $|I_\sigma|$ is small, then $\Ex[N_\sigma]$ is comparable to $\Delta$, so the concentration bound in \cref{lem:1} is weak.
        We can overcome this by setting a slightly larger lower bound for small intersections and a slightly smaller lower bound for large intersections.
        {Specifically, we claim that the following intersectional interventions suffice:}
        \begin{align*}
          \text{for all } \sigma \in \zo^p,\qquad L_\sigma \coloneqq \frac{\abs{I_\sigma}}{m}\cdot n\cdot (1-\eps)+2^{-p}{n\eps}.\footnotemark\yesnum\label{eq:inter_const_choice}
        \end{align*}
        \noindent
        To simplify notation set $\Delta\coloneqq 2^{-p}{n\eps}$.
        Let $\evE$ be the event that for all $\sigma\in \zo^p$, $N_\sigma\leq \Ex[N_\sigma] + \Delta$, i.e.,
        $\evE\coloneqq \bigcup_{\sigma\in \zo^p} \inparen{N_\sigma\leq |I_\sigma|\cdot \frac{n}{m} + \Delta}.$
        \footnotetext{If $L_\sigma>|I_\sigma|$, we update $L_\sigma = |I_\sigma|$.}
        Assume that event $\evE$ occurs.
        Since $\tx$ picks at least $L_\sigma$ items from $I_\sigma$, we have
        \begin{align*}
          \frac{\sum_{i\in I_\sigma} \tx_i w_i}{\sum_{i\in I_\sigma} \sx_i w_i}\quad
          &\Stackrel{\eqref{eq:lowerbound_on_ratio}}{\geq} \quad \frac{ L_\sigma}{N_\sigma}\\
          &\geq\quad  \frac{ L_\sigma}{|I_\sigma|\cdot \frac{n}{m} + \Delta}\\
          &\Stackrel{\eqref{eq:inter_const_choice}}{=}\quad
          \frac{ \frac{\abs{I_\sigma}}{m} n\cdot (1-\eps) + \Delta }{|I_\sigma|\cdot \frac{n}{m} + \Delta}\\
          &\geq\quad  1- \frac{ 2^p n^{-1}\Delta }{1 +  m\Delta (n\cdot |I_\sigma|)^{-1} }\\
          &\geq\quad  1- \frac{2^p n^{-1}\Delta }{1 +  n^{-1}\Delta}. %
        \end{align*}
        Choose $m_0\coloneqq \frac{1}{\eps\eta}\cdot 2^{p+1} $.
        Substituting the value of $\Delta$ in the above equation, we get
        $$\frac{\sum_{i\in I_\sigma} \tx_i w_i}{\sum_{i\in I_\sigma} \sx_i w_i} > 1-\eps.$$
        Now the lemma follows by \Eqref{eq:ratio_conditioned_on_ce} and showing that the event $\evE$ occurs with high probability.
        To see this, observe that using the union bound and \cref{lem:1} we have
        \begin{align*}
          \Pr[\evE]
          \qquad
          &\Stackrel{\rm\cref{lem:1}}{\geq}\qquad  1- 2^p e^{-\frac{\Delta^2}{n}}\\
          &\Stackrel{\ensuremath{(\Delta\coloneqq n\eps 2^{-p})}}{=}\qquad 1-2^{p}\cdot e^{-n\eps^2 2^{-2p}}\\
          &= 1-\exp\inparen{-\frac{\eps^2 n}{p 4^{p}}}.\yesnum\label{eq:prob_of_event_e}
        \end{align*}

      \end{proof}
      \begin{proof}[Proof of \cref{thm:fullrecovery} (assuming \cref{coro:high_ratio_intersectional})]
        Set 
        \begin{align*}
            m_0 \coloneqq \Theta\inparen{\frac{p4^p}{\eta\eps^2}\ln{\frac1\eps}}.
            \yesnum
            \label{eq:def_m0:lb}
        \end{align*}
        Consider the constraints for $\eps_0=\frac\eps2$ from \cref{thm:fullrecovery}, and the event $\evE$ for $\eps_0$ from the proof of \cref{thm:fullrecovery}.
        Then from \cref{thm:fullrecovery}, we know conditioned on the event $\evE$,
        \begin{align*}
          R(L,\beta)\coloneqq \frac{\sinangle{\tx,w}}{\sinangle{\sx,w}}>1-\eps_0.
        \end{align*}
        Thus, it follows that $$\Ex_{w}\insquare{ R(L,\beta) \mid \evE} > 1-\eps_0.$$ %
        {Since, $R(L,\beta)\in [0,1]$, we can use \cref{fact:whp}. This gives us}
        \begin{align*}
          \abs{\Ex\nolimits_{w}\insquare{ R(L,\beta) } - \Ex\nolimits_{w}\insquare{ R(L,\beta) \mid \evE}}
          &\leq 1-\Pr\insquare{\evE}\\
          &\leq e^{-\frac{\eps_0^2}{p 4^{p}} n}\\
          &\Stackrel{}{\leq} \frac{\eps}2.\tag{Using \cref{eq:def_m0:lb}}
        \end{align*}
        {Combining the above two equations we get the required result: $\ratio_{\cD}(L,\beta)\coloneqq \Ex_{w}\insquare{ R(L,\beta) } > 1-\eps.$}
      \end{proof}

      \subsubsection{Proof of \cref{lem:1}}\label{sec:proofof:lem:1}
      \begin{proof}
        \sx{} picks the $n$ items with the highest latent utility, i.e., it picks the top $n$ order statistics.
        Thus, $N_\sigma$ counts the number of the top $n$ order statistics which fall in $I_\sigma$.

        \paragraph{Equivalent urn model.}
        Pick $n$ items without replacement and color them blue.
        Color the remaining items red.
        Let $N_\sigma$ be the number of blue balls in $I_\sigma$.

        \noindent Suppose $N_{\sigma}=j$.
        Then number of ways of coloring $j$ out of $|I_\sigma|$ balls in $I_\sigma$ is $\binom{\abs{I_\sigma}}{j}$, and the number of ways of coloring $n-j$ remaining blue balls is $\binom{m-\abs{I_\sigma}}{n-j}$.
        Therefore it follows that:
        \begin{align}
          \Pr[ N_{\sigma}= j  ] = \frac{\binom{\abs{I_\sigma}}{j}\binom{m-j}{n-j}}{\binom{m}{n}}.\label{eq_prob_mk}
        \end{align}

        \begin{mdframed}
          Given numbers $N, K, n $, for an hyper-geometric random variable, $HG$, we have: $\forall k:\: \max(0,n+K-N)\leq k \leq \min(K,n)$
          \begin{align}
            \Pr[HG = k] \coloneqq \frac{\binom{K}{k}\binom{N-K}{n-k}}{\binom{N}{n}}.\label{eq_hyper}
          \end{align}
        \end{mdframed}

        \noindent Thus, one can observe that $N_{\sigma}$ is a hyper-geometric random variable.
        From the well-known properties of the hyper-geometric distribution~\cite{hush2005concentration}, we have that:
        \begin{align}
          \Ex[\ N_{\sigma} \ ] &= |I_\sigma|\cdot \frac{n}{m},\\
          \Pr\left[ N_{\sigma} \geq \Ex[N_\sigma]+ \Delta\right] &\stackrel{\text{\cite{hush2005concentration}}}{\leq} e^{-2(\Delta^2-1)\gamma},
        \end{align}
        where $\gamma \coloneqq \max \big(\frac{1}{|I_\sigma|+1}+\frac{1}{n-|I_\sigma|+1}, \frac{1}{n+1}+\frac{1}{m-n+1}\big) \geq \frac{1}{n+1}$.
        Since for all $\Delta\geq 2$ and $n\geq 0$, $2\frac{\Delta^2-1}{n+1}\geq \frac{\Delta^2}{n}$, we get that
        \begin{align*}
          \Pr\left[ N_{\sigma} \geq \Ex[N_\sigma]+ \Delta\right] \leq e^{-\frac{\Delta^2}{n}}.
        \end{align*}
      \end{proof}

      \renewcommand{\bb}{\ensuremath{\beta_{\rm max}}}

      \subsection{Proof of Proposition~\ref{prop:89}}\label{sec:proofof:prop:89}
      Recall that the utility ratio is defined as
      \begin{align*}
        \ratio(L,\beta)\coloneqq \Ex_{w}\insquare{ \frac{ \sinangle{\tx,w} }{ \sinangle{\sx,w} } },
      \end{align*}
      where $\sx\coloneqq \argmax_x \inangle{x,w}$ is a function of $w$, and $\tx\coloneqq \argmax_{x\in C(L)} \inangle{x,\hw}$ is a function of $w$,$L$, and $\beta$.
      Our goal is to construct an explicit family of groups $G_1$ and $G_2$ and a value of $n$, parameterized by $m$, such that for all $0<\beta_1,\beta_2<1$ and large enough $m$ it holds that:
      \begin{align*}
        \max_{L_1,L_2\in \Z_{\geq 0}}\ratio_\cD(L,\beta)\leq \frac89 + \frac{3}2\cdot \max\inbrace{\beta_1,\beta_2}.
        \end{align*}
        More precisely, we show that
        \begin{align*}
          \max_{L_1,L_2\in \Z_{\geq 0}}\ratio_\cD(L,\beta)\leq \frac89 + \frac{4\bb}{3}\cdot \max\inbrace{\beta_1,\beta_2} + O(m^{-\frac13}).
          \yesnum\label{eq:to_prove:89}
          \end{align*}
          \cref{prop:89} follows from \cref{eq:to_prove:89} by choosing $m_0\geq \Omega\inparen{\frac{1}{\max\inbrace{\beta_1,\beta_2}^3}}$.

          Towards proving \cref{eq:to_prove:89}, we consider the following family:
          \begin{align*}{}
            n = \frac{m}{2}
            \qquad \text{ and }\qquad
            \forall\sigma\in \zo^2,\quad I_\sigma = \frac{m}{4}.
            \yesnum\label{eq:bad_example_construction_2}
          \end{align*}
          Without loss of generality we assume that $L_1\geq L_2.	$
          We consider three cases depending on the values of $L_1$ and $L_2$:
          \begin{enumerate}
            \item {\bf (Case A)} $L_1\leq \frac{m}{4}$ and $L_2\leq \frac{m}{4}$,
            \item {\bf (Case B)} $L_1> \frac{m}{4}$ and $L_2\leq \frac{m}{4}$, and
            \item {\bf (Case C)} $L_1> \frac{m}{4}$ and $L_2> \frac{m}{4}$.
          \end{enumerate}
          \noindent In each case, we prove that \cref{eq:to_prove:89}, for any $L_1,L_2\geq 0$:
          \begin{align*}
            \Ex\insquare{\inangle{\tx, w}} \leq n\inparen{\frac{2}{3} + \bb} + O(1).
            \yesnum\label{eq:proved_bound:89}
          \end{align*}
          Using this one can prove \cref{eq:to_prove:89} as follows:
          Let $\evJ$ be the event that
          \begin{align*}
            \abs{\sinangle{\sx,w} - \Ex[\sinangle{\sx,w} ]} \leq O(nm^{-\frac13}).
            \yesnum\label{eq:def_cj_2}
          \end{align*}
          From \cref{lem:xtr_conc} we have that
          \begin{align*}
            \Pr[\evJ]\geq 1-O(m^{-\frac13}).
            \yesnum\label{eq:whpcj_2}
          \end{align*}
          Using \cref{fact:whp}, \cref{eq:whpcj_2}, the fact that the utility ratio takes values between 0 and 1:
          \begin{align*}
            \Ex\insquare{ \frac{ \sinangle{\tx,w} }{ \sinangle{\sx,w} }}
            &\leq  \Ex\insquare{ \frac{ \sinangle{\tx,w} }{ \sinangle{\sx,w} } \given \evJ } + O(m^{-\frac13})\\
            &\leq \frac{ \Ex\insquare{\sinangle{\tx,w} \given \evJ } }{ \Ex\insquare{\sinangle{\sx,w}} -  O(nm^{-\frac13})}   + O(m^{-\frac13})
            \tag{Using \cref{eq:def_cj_2}}\\
            &\leq \frac{ \Ex\insquare{\sinangle{\tx,w} } + O(nm^{-\frac13}) }{ \Ex\insquare{\sinangle{\sx,w}} -  O(nm^{-\frac13})}   + O(m^{-\frac13})
            \tag{Using \cref{fact:whp}, \cref{eq:whpcj_2}, the fact that $\sinangle{\tx,w}\leq \sum_{i=1}^m \tx_i = n$}\\
            &\leq \frac{ \frac{2n}{3} + n\bb + O(1) + O(nm^{-\frac13}) }{ n\inparen{1-\frac{n}{2m+2}} -  O(nm^{-\frac13})}   + O(m^{-\frac13})
            \tag{Using \cref{lem:xtr_conc} and \cref{eq:proved_bound:89}}\\
            &\leq \frac{ \frac{2n}{3} + n\bb }{ n\inparen{1-\frac{n}{2m}}}   + O_\eta(m^{-\frac13})
            \tag{Using that $\eta\leq \frac{n}{m}\leq 1-\eta$}\\
            &\leq \frac{8}{9} + \frac{4\bb}{3} + O_\eta(m^{-\frac13}).
            \tag{Using \cref{eq:bad_example_construction_2}}
          \end{align*}

          \subsubsection{\bf Additional Facts and Notation}
          \noindent
          From \cref{obs:wtx_picks_top_latent_util}, we know that if $\tx$ selects $k_\sigma\in \N$ items from $I_\sigma$, then these are exactly the $k_\sigma$ items with the largest latent utility in $I_\sigma$.
          Hence, given the utilities $w$, to determine $\sx$ and $\tx$, it suffices to compute
          \begin{align*}
            \forall \sigma\in \zo^2,\quad K_\sigma\coloneqq \sum_{i\in I_\sigma} \tx_i.
          \end{align*}
          Here, $K\in \N^4$ is a random variable.
          If it $K$ was deterministic and known, then since latent utilities are drawn independently from $\cU$, we can use \cref{fact:sum_n} to compute $\Ex\insquare{\sinangle{\tx,w}}$.
          Our strategy is to prove that $K$ satisfies certain linear inequalities with high probability.
          Then, using \cref{fact:sum_n_ub}, we upper bound $\Ex\insquare{\sinangle{\tx,w}}$. %
          Towards this, consider the following event:
          \begin{definition}
            Let $\evF$ be the event that:
            For each $\sigma\in \zo^2$, the intersection $I_\sigma$ has at least $\frac{m}{4}\cdot (1-3\bb)$ items with latent utility higher than $2\bb$.
          \end{definition}
          \noindent
          We show that $\evF$ holds with high probability (\cref{fact:whpcf_2}), and that conditioned on $\evF$, $K$ does not satisfy certain linear-inequalities (as specified in \cref{lem:imposs_res}).
          \begin{lemma}\label{fact:whpcf_2}
            It holds that $\Pr[\evF]\geq 1-4\exp\inparen{-\frac{m}{12}\cdot \bb^2\cdot  (1-2\bb)}$.
          \end{lemma}
          \begin{proof}
            For all $i\in [m]$, let $Z_i\in \zo$ be the indicator random variable that $w_i>2\bb$, i.e., $Z_i\coloneqq \mathbb{I}\insquare{w_i>2\bb}.$
            Since for all $i\in [m]$, $w_i$ is drawn independently from $\unif$, we have
            \begin{align*}
              \Pr[Z_i=1]=1-2\bb
              \quad\text{and}\quad
              \forall i,j\in [m],i\neq j,\quad
              \text{$Z_i$ and $Z_j$ are independent.}
              \yesnum\label{eq:probZ_2}
            \end{align*}
            We prove the following:
            For each $\sigma\in \zo^2$, it holds that
            \begin{align*}
              \Pr\insquare{ \sum\nolimits_{i\in I_\sigma} Z_i   < \frac{m}{4}\cdot (1-3\bb)   }
              \leq \exp\inparen{\frac{m\bb^2 (1-2\bb)}{12}}.
              \yesnum\label{eq:whp:89}
            \end{align*}
            Then, the lemma follows by a union bound over intersections $\sigma\in \zo^2$.
            We prove \cref{eq:whp:89} as follows:
            \begin{align*}
              \Pr\insquare{ \sum\nolimits_{i\in I_\sigma} Z_i   < \frac{m}{4}\cdot (1-3\bb)   }
              \ \
              &\Stackrel{}{\leq}\ \ \Pr\insquare{ \sum\nolimits_{i\in I_\sigma} Z_i < \frac{m}{4}\cdot (1-2\bb)\cdot (1-\bb)}
              \tag{Using that $\bb\geq 0$}\\
              &\Stackrel{}{=}\ \ \Pr\insquare{ \sum\nolimits_{i\in I_\sigma} Z_i < (1-\bb)\cdot \Ex\insquare{\sum\nolimits_{i\in I_\sigma} Z_i}}
              \tag{Using linearity of expectation, \cref{eq:probZ_2}, and \cref{eq:bad_example_construction_2}}\\
              &\leq \ \  \exp\inparen{\frac{1}{3}\cdot \bb^2 \cdot \Ex\insquare{\sum\nolimits_{i\in I_\sigma} Z_i}}\\
              &\leq \ \  \exp\inparen{\frac{m}{12}\cdot \bb^2 \cdot (1-2\bb)}.
              \tag{Using linearity of expectation, \cref{eq:probZ_2}, and \cref{eq:bad_example_construction_2}}
            \end{align*}
          \end{proof}
          \begin{lemma}\label{lem:imposs_res}
            Conditional on $\evF$  none of the following events can hold:
            \begin{enumerate}
              \item {\bf (Event 1):} $K_{00}<\frac{m}{4}\cdot\inparen{1-3\bb}$ and $K_{01}>0$ and $K_{01}+K_{11}>L_1.$
              \item {\bf (Event 2):} $K_{00}<\frac{m}{4}\cdot\inparen{1-3\bb}$ and $K_{10}>0$ and $K_{10}+K_{11}>L_2.$
              \item {\bf (Event 3):} $K_{00}<\frac{m}{4}\cdot\inparen{1-3\bb}$ and $K_{01}>0$ and $K_{10}>0$  and $K_{11}<\frac{m}{4}.$
            \end{enumerate}
          \end{lemma}
          \begin{proof}
            Assume that $\evF$ occurs.

            \smallskip
            \noindent {\bf (Event 1)} Since $\evF$ occurs and $K_{00}<\frac{m}{4}\cdot \inparen{1-2\bb}$, we can pick a $j^\star\in I_{00}$ such that
            \begin{align*}
              w_{j^\star}> 2\bb \quad\text{and}\quad \tx_{j^\star} = 0.\yesnum\label{eq:jstar_2}
            \end{align*}
            Consider any $i^\star\in I_{01}$ such that $\tx_{i^\star}=1$.
            Construct a selection $\overline{x}\in \zo^m$ by swapping $i^\star$ and $j^\star$ in $\tx$.
            Thus, compared to $\tx$, $\overline{x}$ picks one less element from $G_1$ and the same number of elements from $G_2$.
            In particular, if $K_{01}+K_{11}>L_1$, then $\overline{x}$ satisfies the lower bound constraints.
            Further, we can see that $\overline{x}$ has a higher observed utility than $\tx$ as follows:
            \begin{align*}
              \inangle{\overline{x}, \hw} = \sinangle{\tx,\hw}-\hw_{i^\star}+\hw_{j^\star} \Stackrel{\eqref{eq:jstar_2}}{>}  \sinangle{\tx,\hw}-\hw_{i^\star}+2\bb \Stackrel{}{>} \sinangle{\tx,\hw}+\bb.
            \end{align*}
            In the last step we use $\hw_{i^\star}\coloneqq \beta_1 w_{i^\star}\leq \beta_1\leq \bb$.
            We have a contradiction since $\overline{x}$ satisfies the lower bound constraint and has higher observed utility than $\tx$.

            \smallskip
            \noindent {\bf (Event 2)}
            The proof for this case follows by replacing $L_1$, $\beta_1$, $K_{10}$, and $K_{01}$ by $L_2$, $\beta_2$, $K_{01}$, and $K_{10}$ respectively in the proof for Event 1.

            \smallskip
            \noindent {\bf (Event 3)} Since $\evF$ occurs and $K_{00}<\frac{m}{4}\cdot \inparen{1-2\bb}$, we can pick a $j_1\in I_{00}$ such that
            \begin{align*}
              w_{j_1}> 2\bb \quad\text{and}\quad \tx_{j_1} = 0.\yesnum\label{eq:jstar_3}
            \end{align*}
            Consider any $i_1\in I_{01}$ and $i_2\in I_{10}$ such that $\tx_{i_0}=1$ and $\tx_{i_2}=1$ (these exist as $K_{01}>0$ and $K_{10}>0$ ).
            Further, consider any $j_2\in I_{11}$ such that $\tx_{i_0}=0$ (this exists as $K_{11}<\frac{m}{4}$).
            Construct a selection $\overline{x}\in \zo^m$ by swapping $i_1$ and $i_2$ for $j_1$ and $j_2$ in $\tx$.
            Thus, compared to $\tx$, $\overline{x}$ the same number of items from $G_1$ and $G_2$.
            In particular, then $\overline{x}$ satisfies the lower bound constraints.
            Further, we can see that $\overline{x}$ has a higher observed utility than $\tx$ as follows:
            \begin{align*}
              \inangle{\overline{x}, \hw} = \sinangle{\tx,\hw}-\hw_{i_1}-\hw_{i_2}+\hw_{j_1}+\hw_{j_2} \Stackrel{\eqref{eq:jstar_3}}{>}
              \sinangle{\tx,\hw}-\beta_1-\beta_2+2\bb+0 \Stackrel{}{\geq} \sinangle{\tx,\hw}+\bb.
            \end{align*}
            Where we use $\hw_{i_1}\coloneqq \beta_1 w_{i_1}\leq \beta_1\leq \bb$ (similarly for $i_2$).
            We have a contradiction since $\overline{x}$ satisfies the lower bound constraint and has higher observed utility than $\tx$.
          \end{proof}

          \subsubsection{Step 1: Prove \cref{eq:proved_bound:89} for Case A}

          \noindent Recall that in Case A, $L_1\leq \frac{m}{4}$ and $L_2\leq \frac{m}{4}$.
          This step relies on the following lemma.
          \begin{lemma}\label{lem:thm_2_case1}
            If both $L_1\leq \frac{m}{4}$ and $L_2\leq \frac{m}{4}$,
            then conditioned on $\evF$, $K_{00} \geq \frac{m}{4}\cdot(1-3\bb).$
          \end{lemma}
          \noindent The proof of \cref{lem:thm_2_case1} appears at the end of this section.
          Since $\tx$ selects $n$ items, we have that $\sum_{\sigma\in \zo^2}K_\sigma=n$.
          Using this and \cref{lem:thm_2_case1}, it follows that:
          Conditioned on $\evF$
          \begin{align*}
            K_{01}+K_{10}+K_{11}=n-K_{00}\leq \frac{m}{4}\cdot (1+3\bb).
            \yesnum\label{eq:bound_on_sum}
          \end{align*}
          Further, since $I_{00}=\frac{m}{4}$ (by \cref{eq:bad_example_construction_2}), it follows that
          \begin{align*}
            K_{00}\leq\frac{m}{4}.
            \yesnum\label{eq:bound_on_sum:b:new}
          \end{align*}
          Then, we have:
          \begin{align*}
            \Ex\insquare{\sinangle{\tx,w}\given \evF}
            &\quad =\quad \Ex\insquare{  \sum\nolimits_{i\in I_{00}} \tx_i w_i \given \evF}+\Ex\insquare{ \sum\nolimits_{i\in I_{01}\cup I_{10}\cup I_{11} } \tx_i w_i \given \evF}
            \tag{Using linearity of expectation}\\
            \intertext{Using \cref{eq:bound_on_sum:b:new}, we can apply \cref{fact:sum_n_ub} with $S\coloneqq I_{00}$, $U\coloneqq \frac{m}{4}$, and $\evF$.
            On substituting $\abs{I_{00}}=\frac{m}{4}$ and $\Pr[\evF]$ from \cref{fact:whpcf_2}, in the result we get:}
            &\ \ \Stackrel{}{=} \ \
            \frac{m}{4}\inparen{1-\frac{\frac{m}{4}+1}{\frac{m}2+2}}
            +\Ex\insquare{\sum\nolimits_{i\in I_{01}\cup I_{10}\cup I_{11} } \tx_i \given \evF}
            + ne^{-\Omega(m\cdot \bb^2)}
            \intertext{Similarly, using \cref{eq:bound_on_sum}, to apply \cref{fact:sum_n_ub} with $S\coloneqq I_{00}$, $U\coloneqq \frac{m}{4}$, and $\evF$, we get:}
            &\quad \Stackrel{}{\leq}\quad   \frac{m}{4}\inparen{1-\frac{\frac{m}{4}+1}{\frac{m}2+2}}
            + ne^{-\Omega(m\cdot \bb^2)}
            + \frac{m}{4}\cdot (1+3\bb)\cdot \inparen{1-\frac{\frac{m}{4}\cdot(1+3\bb)+1}{\frac{3m}2+2}}\\
            &\quad \Stackrel{}{=} \quad  \frac{m}{4}\inparen{1-\frac{1}{2}}
            + \frac{m}{4}\cdot (1+3\bb)\cdot \inparen{1-\frac{1+3\bb}{6}} + O(1)\\
            &\quad \Stackrel{}{\leq} \quad  \frac{m}{4}\inparen{1-\frac{1}{2}}
            + \frac{m}{4}\cdot \inparen{\frac{5}{6}+2\bb}+O(1)
            \tag{Using $\bb\geq 0$}\\
            &\quad \Stackrel{}{=}\quad   n\cdot \inparen{\frac{2}{3}+\bb}+O(1).
            \tag{Using \cref{eq:bad_example_construction_2}}
          \end{align*}
          \begin{proof}[Proof of \cref{lem:thm_2_case1}]
            Suppose, toward a contradiction, that $\evF$ holds and $K_{00} < \frac{m}{4}\cdot (1-3\bb)$.
            Then from Event 1 in \cref{lem:imposs_res} we get that either $K_{01}+K_{11}=L_1$ or $K_{01}=0$.
            Similarly, from Event 2 in \cref{lem:imposs_res} we get that either $K_{10}+K_{11}=L_2$ or $K_{10}=0$.
            These give us four cases:
            \paragraph{{\bf (Case I)} $K_{01}=K_{10}=0$:}
            Since $K_{01}=K_{10}=0$, we have that $K_{11} = n  - K_{00}$.
            Hence, if  $K_{00} < \frac{m}{4}\cdot(1-3\bb)$, then $K_{11}>\frac{m}{4}\cdot (1+3\bb).$
            Which is a contradiction because $\abs{I_{11}}=\frac{m}{4}.$

            \paragraph{{\bf (Case II)} $K_{01}+K_{11}=L_1$ and $K_{10}=0$:}
            Since $n = \sum_{\sigma} K_\sigma$, in this case we have
            $n \leq K_{00} + L_1$.
            If $K_{00} < \frac{m}{4}\cdot(1-3\bb)$, then $n < \frac{m}{4}\cdot(1-3\bb) + L_1$.
            Further, in Case A, we know that $L\leq \frac{m}{4}$.
            Hence, $n < \frac{m}2 -\frac{m\bb}4$.
            This is a contradiction since $n=\frac{m}{2}$ (see \cref{eq:bad_example_construction_2}).

            \paragraph{{\bf (Case III)} $K_{10}+K_{11}=L_2$ and $K_{01}=0$:}
            This follows by replacing $L_1$ by $L_2$ in the proof of Case II.

            \paragraph{{\bf (Case IV)} $K_{01}+K_{11}=L_1$ and $K_{10}+K_{11}=L_2$:}
            If $K_{10}=0$, then this is a special case of Case II.
            Hence, we must have $K_{10}>0$.
            Similarly, we must have $K_{01}>0$.
            Since $K_{10}+K_{11}=L_2$ and $K_{10}>0$, $K_{11}<L_2$.
            Further, as $L_2\leq \frac{m}{4}$, it follows that $K_{11}<\frac{m}{4}$
            Thus, we have a contradiction due to Event 3 in \cref{lem:imposs_res}.
          \end{proof}

          \subsubsection{\bf Step 2: Prove \cref{eq:proved_bound:89} for Case B}

          \noindent Recall that in Case B, $L_1 > \frac{m}{4}$ and $L_2\leq \frac{m}{4}$.
          This step relies on the following lemma.
          \begin{lemma}\label{lem:thm_2_case2}
            If both $L_1 > \frac{m}{4}$ and $L_2\leq \frac{m}{4}$,
            then conditioned on $\evF$, $K_{00} \geq (n-L_1)\cdot (1-3\bb).$
          \end{lemma}
          \noindent The proof of \cref{lem:thm_2_case2} appears at the end of this section.
          Since $\tx$ selects $n$ items, we have that $\sum_{\sigma\in \zo^2}K_\sigma=n$.
          Using this and \cref{lem:thm_2_case2}, it follows that:
          Conditioned on $\evF$
          \begin{align*}
            K_{01}+K_{10}+K_{11}
            =n-K_{00}
            &\leq L_1+3\bb\cdot (n-L_1).
            \yesnum\label{eq:step_2:new}
          \end{align*}
          Since $K_{10}\geq 0$, it follows that
          \begin{align*}
            K_{01}+K_{11}\leq L_1+3\bb\cdot (n-L_1).\yesnum\label{eq:bound_on_sum_2}
          \end{align*}
          Further, substituting $K_{01}+K_{11}\geq L_1$ in \cref{eq:step_2:new}, we get
          \begin{align*}
            K_{10}\leq 3\bb\cdot (n-L_1).\yesnum\label{eq:bound_on_sum_2_2}
          \end{align*}
          Finally, as $K_{01}+K_{11}\leq L_2$ and $K_{10}\geq 0$, we have
          \begin{align*}
            K_{00} = n-K_{01}-K_{10}-K_{11}\leq n-L_1.\yesnum\label{eq:bound_on_sum_2_3}
          \end{align*}
          Then, we have
          \begin{align*}
            \Ex\insquare{\sinangle{\tx,w}\given \evF}
            \quad
            &=\quad \Ex\insquare{  \sum\nolimits_{i\in I_{00} } \tx_i w_i \given \evF}+\Ex\insquare{ \sum\nolimits_{i\in I_{10} } \tx_i w_i \given \evF}
            +\Ex\insquare{ \sum\nolimits_{i\in I_{01}\cup I_{11} } \tx_i w_i \given \evF}\\
            \intertext{
            Using \cref{eq:bound_on_sum_2,eq:bound_on_sum_2_2,eq:bound_on_sum_2_3} respectively, we can apply \cref{fact:sum_n_ub} for each of the three terms in the RHS of the above equation.
            On substituting $\abs{I_{00}}=\frac{m}{4}$ and $\Pr[\evF]$ from \cref{fact:whpcf_2}, in the result we get:
            }
            \Ex\insquare{\sinangle{\tx,w}\given \evF}
            &\Stackrel{}{\leq} \qquad (n-L_1)\cdot \inparen{1-\frac{(n-L_1)+1}{\frac{m}2+2}}
            +3\bb\cdot (n-L_1)\cdot \inparen{1-\frac{3\bb\cdot (n-L_1)+1}{2(\frac{m}2+1)}}\\
            &\qquad + \quad (L_1+3\bb\cdot (n-L_1))\inparen{1-\frac{L_1+3\bb\cdot (n-L_1)+1}{\frac{m}2+2}} + ne^{-\Omega(m\cdot \bb^2)}\\
            &=\qquad (n-L_1)\cdot \inparen{1-2\cdot \frac{n-L_1}{m}}
            + 3\bb\cdot (n-L_1)\cdot \inparen{1-\frac{3\bb\cdot (n-L_1)}{m}}\\
            &\qquad + (L_1+3\bb\cdot (n-L_1))\inparen{1-2\cdot \frac{L_1+3\bb\cdot (n-L_1)}{m}}+O(1)\\
            &\Stackrel{}{=} \qquad -(3+4\bb+12\bb^2)\cdot \frac{L_1^2}{2n} + 2(1+\bb+6\bb^2)\cdot L_1 - 6\bb^2 n +  O(1).
            \yesnum\label{eq:tmp:89}
          \end{align*}
          The above equation is a strongly-concave and quadratic function of $L_1$.
          It achieves its maximum at
          $$L_1 = \frac{2n\cdot (1+\bb+6\bb^2)}{3+4\bb+12\bb^2}.$$
          Substituting this in \cref{eq:tmp:89}, we get
          \begin{align*}
            \Ex\insquare{\sinangle{\tx,w}\given \evF}\ \
            &\leq\ \  \frac{m}{3}\inparen{\frac{1+2\bb+4\bb^2}{1+\frac43\bb+4\bb^2}}\\
            &\leq \ \ \frac{m}{3}\inparen{1+\frac{2\bb}{3}}\\
            &\Stackrel{\eqref{eq:bad_example_construction_2}}{=}\ \  n\inparen{\frac23+\frac{4\bb}{9}}.
          \end{align*}
          \begin{proof}[Proof of \cref{lem:thm_2_case2}]
            Suppose, toward a contradiction, that $\evF$ holds and $K_{00} < (n-L_1)\cdot (1-3\bb).$
            Hence,
            \begin{align*}
              K_{00}
              &\ \ <\ \  (n-L_1)\cdot  (1-3\bb)\yesnum\label{eq:used_lat}\\
              &\ \ \leq\ \  \inparen{n-\frac{m}{4}}\cdot (1-3\bb)
              \tag{Using that $L_1\geq \frac{m}{4}$}\\
              &\ \ \Stackrel{\eqref{eq:bad_example_construction_2}}{=}\ \  \frac{m}{4}\cdot(1-3\bb).\yesnum\label{eq:used_lat_2}
            \end{align*}
            Then from Event 1 in \cref{lem:imposs_res} we get that either $K_{01}+K_{11}=L_1$ or $K_{01}=0$.
            Similarly, from Event 2 in \cref{lem:imposs_res} we get that either $K_{10}+K_{11}=L_2$ or $K_{10}=0$.
            These give us four cases:
            \paragraph{{\bf (Case I)} $K_{01}=K_{10}=0$:}
            Since $\tx$ satisfies the lower bounds, we have $K_{11}  + K_{01} \geq L_1$.
            Using that $K_{01}=0$ in this Case I and $L_1 > \frac{m}{4}$ in Case B, we get $K_{11} > \frac{m}{4}.$
            This is a contradiction as $\abs{I_{11}}=\frac{m}{4}$ and $K_{11}\leq \abs{I_{11}}$.

            \paragraph{{\bf (Case II)} $K_{01}+K_{11}=L_1$ and $K_{10}=0$:}
            In this case, we have that $K_{00}+L_1+0=n$.
            Hence, $K_{00}=n-L_1$.
            This is a contradiction to \Eqref{eq:used_lat}.

            \paragraph{{\bf (Case III)} $K_{10}+K_{11}=L_2$ and $K_{01}=0$:}
            Same proof as Case I.

            \paragraph{{\bf (Case IV)} $K_{01}+K_{11}=L_1$ and $K_{10}+K_{11}=L_2$:}
            Suppose $K_{10}=0$, then this case is a special case of Case II.
            But since Case II results in a contradiction, we must have $K_{10}>0$.
            Similarly, if $K_{01}=0$, this case reduces to Case III.
            But since Case III results in a contradiction, we must have $K_{01}>0$.
            Finally, as $K_{10}+K_{11}=L_2$ and $K_{10}>0$, $K_{11}<L_2$.
            Since $L_2\leq \frac{m}{4}$, $K_{11}<\frac{m}{4}$.
            Thus, we have a contradiction due to Event 3 in \cref{lem:imposs_res}.
          \end{proof}

          \subsubsection{\bf Step 3: Prove \cref{eq:proved_bound:89} for Case C}

          \noindent Recall that in Case C, $L_1 > \frac{m}{4}$ and $L_2>\frac{m}{4}$.
          This step relies on the following lemma.
          \begin{lemma}\label{lem:thm_2_case3}
            If both $L_1 > \frac{m}{4}$ and $L_2> \frac{m}{4}$,
            then conditioned on $\evF$, $K_{00} \geq \inparen{n-L_1-L_2+\frac{m}{4}}\cdot (1-3\bb).$
          \end{lemma}
          \noindent The proof of \cref{lem:thm_2_case3} appears at the end of this section.
          Since $\tx$ selects $n$ items, we have that $\sum_{\sigma\in \zo^2}K_\sigma=n$.
          Using this and \cref{lem:thm_2_case2}, it follows that:
          Conditioned on $\evF$,
          \begin{align*}
            K_{01}+K_{10}+K_{11}
            &=n-K_{00}\\
            &\leq L_1+L_2-\frac{m}{4}+3\bb\cdot\inparen{n-L_1-L_2+\frac{m}{4}}.
            \yesnum\label{eq:tmp:tmp:new}
          \end{align*}
          Since $K_{10}\geq 0$,
          \begin{align*}
            K_{01}+K_{11}\leq L_1+L_2-\frac{m}{4}+3\bb\cdot\inparen{n-L_1-L_2+\frac{m}{4}}.
          \end{align*}
          Further, substituting $K_{01}+K_{11}\geq L_1$ in \cref{eq:tmp:tmp:new}, we get:
          \begin{align*}
            K_{10}
            &\leq L_2-\frac{m}{4}+3\bb\cdot\inparen{n-L_1-L_2+\frac{m}{4}}\\
            &< L_2-\frac{m}{4}+\frac{3m\bb}4.
            \tagnum{Using that $L_1,L_2>\frac{m}{4}$}\customlabel{eq:bound_on_sum_3_2}{\theequation}
          \end{align*}
          By symmetry, we have
          \begin{align*}
            K_{01}\leq L_1-\frac{m}{4}+\frac{3m\bb}4. \yesnum\label{eq:bound_on_sum_3}
          \end{align*}
          Further, we have
          \begin{align*}
            K_{00}
            &= n-K_{01}-K_{10}-K_{11}\\
            &\leq n-L_1-L_2+K_{11} \tag{Using that $K_{10}+K_{11}\geq L_1$ and $K_{01}+K_{11}\geq L_2$}\\
            &\leq \frac{3m}{4}-L_1-L_2.\tagnum{Using that $n=\frac{m}{2}$ and $K_{11}\leq \abs{I_{11}}=\frac{m}{4}$ due to \cref{eq:bad_example_construction_2}}
            \customlabel{eq:bound_on_sum_3_3}{\theequation}
          \end{align*}
          Also, because $\abs{I_{11}}=\frac{m}{4}$ (\cref{eq:bad_example_construction_2}) and $K_{11}\leq \abs{I_{11}}$, it follows that
          \begin{align*}
            K_{11}\leq \frac{m}{4}.
            \yesnum\label{eq:bound_on_sum_3_4}
          \end{align*}
          Then, we have:
          \begin{align*}
            \Ex\insquare{\sinangle{\tx,w}\given \evF}
            \quad &=\quad \Ex\insquare{  \sum\nolimits_{i\in I_{00} } \tx_i \given \evF}
            +\Ex\insquare{ \sum\nolimits_{i\in I_{10} } \tx_i \given \evF}
            +\Ex\insquare{ \sum\nolimits_{i\in I_{01}} \tx_i \given \evF}
            +\Ex\insquare{ \sum\nolimits_{i\in I_{11} } \tx_i \given \evF}.
          \end{align*}
          Using Equations~\eqref{eq:bound_on_sum_3}, \eqref{eq:bound_on_sum_3_2}, \eqref{eq:bound_on_sum_3_3}, and \eqref{eq:bound_on_sum_3_4} respectively, we can apply \cref{fact:sum_n_ub} for each of the four terms in the RHS of the above equation.
          On substituting $\abs{I_{\sigma}}=\frac{m}{4}$ and $\Pr[\evF]$ from \cref{fact:whpcf_2}, in the result we get:
          \begin{align*}
            \Ex\insquare{\sinangle{\tx,w}\given \evF}
            \quad &\leq\quad
            \inparen{\frac{3m}{4}-L_1-L_2}\cdot \inparen{1- \frac{{\frac{3m}{4}-L_1-L_2}}{\frac{m}{2}+2}}\\
            &\qquad+ \inparen{L_1-\frac{m}{4}+\frac{3m\bb}4}\cdot \inparen{1-\frac{{L_1-\frac{m}{4}+\frac{3m\bb}4}}{\frac{m}{2}+2}}\\
            &\qquad+ \inparen{L_2-\frac{m}{4}+\frac{3m\bb}4}\cdot \inparen{1-\frac{{L_2-\frac{m}{4}+\frac{3m\bb}4}}{\frac{m}{2}+2}}\\
            &\qquad+ {\frac{m}{4}}\cdot \inparen{1-\frac{1}{2}} + ne^{-\Omega(m\cdot \bb^2)}.
          \end{align*}
          Hence, we have
          \begin{align*}
            \Ex\insquare{\sinangle{\tx,w}\given \evF}
            \quad &\leq\quad
            \inparen{\frac{3m}{4}-L_1-L_2}\cdot \inparen{1-2\cdot \frac{{\frac{3m}{4}-L_1-L_2}}{m}}\\
            &\qquad+ \inparen{L_1-\frac{m}{4}+\frac{3m\bb}4}\cdot \inparen{1-2\cdot \frac{{L_1-\frac{m}{4}+\frac{3m\bb}4}}{m}}\\
            &\qquad+ \inparen{L_2-\frac{m}{4}+\frac{3m\bb}4}\cdot \inparen{1-2\cdot \frac{{L_2-\frac{m}{4}+\frac{3m\bb}4}}{m}}\\
            &\qquad+ {\frac{m}{4}}\cdot \inparen{1-\frac{1}{2}} + O(1).
            \yesnum\label{eq:bound_case_c:new}
          \end{align*}
          This is a quadratic and concave expression in $L_1$ and $L_2$.
          Maximizing it over $L_1,L_2\in \R$, we get that the maximum is achieved at:
          \begin{align*}
            L_1=L_2=\frac{m}{4}\cdot (1-\bb).
          \end{align*}
          Substituting these values in \cref{eq:bound_case_c:new}, we get the following upper bound on $\Ex\insquare{\sinangle{\tx,w}\given \evF}$:
          \begin{align*}
            \Ex\insquare{\sinangle{\tx,w}\given \evF}
            \ \ &\leq\ \  \frac{m}{4}\cdot \inparen{1-6\bb^2}\\ 
            &\leq\ \  \frac{n}{2}. \tag{Using that $n=\frac{m}{2}$, due to \cref{eq:bad_example_construction_2}, and $\bb\geq 0$}
          \end{align*}

          \begin{proof}[Proof of \cref{lem:thm_2_case3}]
            Suppose, toward a contradiction, that $\evF$ holds and $K_{00} < \inparen{n-L_1-L_2+\frac{m}{4}}\cdot (1-3\bb).$
            Hence,
            \begin{align*}
              K_{00}
              &\ \ <\ \  \inparen{n-L_1-L_2+\frac{m}{4}}\cdot(1-3\bb)\yesnum\label{eq:used_lat_3}\\
              &\ \ <\ \  \inparen{n-\frac{m}{4}}\cdot(1-3\bb)
              \tag{Using that $L_1,L_2>\frac{m}{4}$}\\
              &\ \ \Stackrel{\eqref{eq:bad_example_construction_2}}{=}\ \  \frac{m}{4}\cdot(1-3\bb).\yesnum\label{eq:used_lat_4}
            \end{align*}
            Then, from Event 1 in \cref{lem:imposs_res} we get that either $K_{01}+K_{11}=L_1$ or $K_{01}=0$.
            Similarly, from Event 2 in \cref{lem:imposs_res} we get that either $K_{10}+K_{11}=L_2$ or $K_{10}=0$.
            These give us four cases:
            \paragraph{{\bf (Case I)} $K_{01}=K_{10}=0$:}
            Since $\tx$ satisfies the lower bound constraints, $K_{11}  + K_{01} \geq L_1$.
            Further, $K_{01}= 0$, we have $K_{11}\geq L_1$.
            Finally, because in Case C, $L_1 > \frac{m}{4}$, we have $K_{11}>\frac{m}{4}$.
            This is a contradiction as $\abs{I_{11}}=\frac{m}{4}$ and $K_{11}\leq \abs{I_{11}}.$

            \paragraph{{\bf (Case II)} $K_{01}+K_{11}=L_1$ and $K_{10}=0$:}
            Substituting that $K_{01}+K_{11}=L_1$ and $K_{10}=0$, in $\sum_{\sigma\in \zo^2}K_\sigma=n$, we get $K_{00}=n - L_1.$
            However, because $L_2>\frac{m}{4}$ in Case C and $\bb\geq 0$, \Eqref{eq:used_lat_3} implies that $K_{00}<\inparen{n-L_1}$.
            Thus, we have a contradiction

            \paragraph{{\bf (Case III)} $K_{10}+K_{11}=L_2$ and $K_{01}=0$:}
            Same proof as Case I.

            \paragraph{{\bf (Case IV)} $K_{01}+K_{11}=L_1$ and $K_{10}+K_{11}=L_2$:}
            Suppose $K_{10}=0$, then this is a special case of Case II.
            But since Case II results in a contradiction, we must have $K_{10}>0$.
            Similarly, if $K_{10}=0$,then this is a special case of Case II.
            However, since Case III results in a contradiction, we must have $K_{01}>0$.
            Now, as  $K_{10}+K_{11}=L_2$ and $K_{10}>0$, we have $K_{11}<L_2$.
            If $K_{11}<\frac{m}{4}$,
            then we have a contradiction due to Event 3 of \cref{lem:imposs_res}.
            Hence, it must hold that  $K_{11}\geq \frac{m}{4}$.
            As $K_{11}\leq \abs{I_\sigma}$ and $\abs{I_\sigma}=\frac{m}{4}$, it must  hold that $K_{11}= \frac{m}{4}$.
            But, this implies that $K_{00}  = n-L_1-L_2+\frac{m}{4}.$ This is a contradiction to \cref{eq:used_lat_3}.
          \end{proof}

          \newcommand{\orSubscript}[2]{\text{\tiny${\hbox{\tiny$($}}#1\negsp-\negsp #2{\hbox{\tiny$:$}}  #1{\hbox{\tiny$)$}}$}}

          \subsubsection{Proof of \cref{lem:xtr_conc}}
          \noindent We note that the proof of \cref{lem:xtr_conc} does not depend on the family of groups chosen. %

          \begin{lemma}\label{lem:xtr_conc}
            $\Ex\insquare{\sinangle{\sx,w}}\negsp =\negsp  n\inparen{1-\frac{n+1}{2(m+1)}}
            \text{ and }
            \Pr\insquare{\abs{\sinangle{\sx,w}- \Ex[\sinangle{\sx,w}]}\negsp \geq\negsp O(n^{\frac23}) }\negsp  \leq\negsp  O(m^{-\frac13}).$
            {
            }
          \end{lemma}

          \begin{proof}
            Recall that $U_{(k:m)}$ denotes the $k$-th order statistic of $m$ independent draws from $\unif$.
            Let $U_{(k:m)}$ correspond to the $k$-th smallest latent utility.
            Since $\sx$ maximizes the latent utility, it must pick items corresponding to $U_{(m:m)}, U_{(m-1:m)}, \cdots, U_{(m-n+1:m)}$.
            Thus,
            \begin{align*}
              \Ex[\sinangle{\sx,w}] = \Ex\insquare{\sum\nolimits_{i=1}^n U_{(m-i+1:m)}}
              \ \ \quad\Stackrel{\rm\cref{fact:sum_n}}{=}\quad\ \
              n\inparen{1-\frac{n+1}{2(m+1)}}.\yesnum\label{eq:expect_lat_util}
            \end{align*}
            It remains to show that the distribution of $\sinangle{\sx,w}$ is concentrated.

            For some $\delta>0$ (fixed later),
            let $\evH$ be the event: $\sinangle{\sx,w} \leq (1+\delta) \sinparen{ \Ex\insquare{ \sinangle{\sx,w}} + \frac12 n m^{-\frac13} }$.
            Fix any $v\in \Ex[U_{(m-n:m)}]\pm m^{-\frac13}$.
            Substituting the value of $\Ex[U_{(m-n:m)}]$ from \cref{fact:os} and using \cref{eq:expect_lat_util}, it follows that
            \begin{align*}
              \frac{n(1+v)}{2}\in
              \Ex[\sinangle{\sx,w}]\pm  \frac{1}{2}nm^{-\frac13}.
              \yesnum\label{eq:new5}
            \end{align*}
            Let $\evE$ be the event that $\abs{ U_{(m-n:m)} - \Ex[U_{(m-n:m)}] }\leq m^{-\frac13}$.
            From Chebyshev's inequality~\cite{motwani1995randomized}  and
            $\mathrm{Var}[U_{(m-n:m)}]\leq \frac{1}{m+2}$ (\cref{fact:os}), we get $\Pr[\evE] \geq 1 - m^{-\frac23}$.
            Thus, applying \cref{fact:whp} using $0\leq \sinangle{\sx,w}\leq n$,
            \begin{align*}
              \abs{ \Ex\sinsquare{ \sinangle{\sx,w} \given \evE} - \Ex\sinsquare{ \sinangle{\sx,w}}}
              \ \leq \
              n m^{-\frac23}.
              \yesnum\label{eq:new2}
            \end{align*}
            Then
            \begin{align*}
              \Pr\insquare{ \evH  \given U_{(m-n:m)}=v}\quad\ \
              &=\quad\ \  \Pr\insquare{ \sinangle{\sx,w} \leq (1+\delta) \inparen{ \Ex\insquare{ \sinangle{\sx,w}} + \frac12 \cdot n m^{-\frac23} }  \given U_{(m-n:m)}=v}\\
              &
              \geq\quad\ \  \Pr\insquare{ \sinangle{\sx,w} \leq (1+\delta) \inparen{ \frac{n(1+v)}{2} }  \given U_{(m-n:m)}=v}\tag{Using \cref{eq:new5}}\\
              &
              \Stackrel{}{=}\quad\ \  \Pr\insquare{ \inparen{  \sum_{i=1}^n \inparen{v+(1-v)U'_{(i:n)}} } \leq (1+\delta) \inparen{ \frac{n(1+v)}{2} }  \given U_{(m-n:m)}=v}
              \tag{Using \cref{fact:conditional_os}}\\
              &
              \geq \quad\ \ \Pr\insquare{ (1-v) \sum_{i=1}^n U'_{(i:n)}  \leq (1+\delta) \inparen{ \frac{n(1 - v)}{2} }  \given U_{(m-n:m)}=v}\tag{Using that $nv\delta\geq 0$}\\
              &
              =\quad\ \  \Pr\insquare{ \sum_{i=1}^n U'_{(i:n)}  \leq (1+\delta){ \frac{n}{2} }}.
              \tag{Using that $v\leq 1$}
            \end{align*}
            Here, $0\leq \sum_{i=1}^n U'_{(i:n)}\leq n$ is the sum of $n$ independent draws from $\unif$.
            Thus, from Hoeffding's inequality~\cite{motwani1995randomized}
            \begin{align*}
              \Pr\insquare{ \evH  \given w_{(m-n:m)}=v} \geq 1 - \exp\inparen{ -2n\delta^2 }.
              \yesnum\label{eq:new3}
            \end{align*}

            \noindent Then
            \begin{align*}
              \Pr\insquare{ \evH  }
              & \  =\ \ \int_{v=0}^1\negsp\negsp \Pr\insquare{ \evH  \given w_{(m-n:m)}=v} \cdot \Pr[w_{(m-n:m)}=v] dv\\
              & \  \geq\ \ \int_{v=\Ex[w_{(m-n:m)}] - m^{-\frac13}}^{\Ex[w_{(m-n:m)}] + m^{-\frac13} }
              \Pr\insquare{ \evH  \given w_{(m-n:m)}=v} \cdot \Pr[w_{(m-n:m)}=v] dv\\
              & \  \Stackrel{}{\geq}\ \ \inparen{1 - \exp\inparen{ -2n\delta^2 }}\cdot \int_{v=\Ex[w_{(m-n:m)}] - m^{-\frac13}}^{\Ex[w_{(m-n:m)}] + m^{-\frac13} }
              \Pr[w_{(m-n:m)}=v] dv\tag{Using \cref{eq:new3}}
             \end{align*}
             \begin{align*}
              & \  =\ \ \inparen{1 - \exp\inparen{ -2n\delta^2 }}\cdot \Pr[\evE]\\
              & \  \geq\ \ 1 - O(m^{-\frac23})
            \end{align*}
            where we choose, e.g., $\delta\coloneqq n^{-\frac13}$, and use that $n=\eta\cdot m$ for fixed a constant $\eta>0$.
            Similarly, we can show that $$\Pr\insquare{  \sinangle{\sx,w} \geq \Ex\insquare{ \sinangle{\sx,w}} - O\sinparen{  n^{\frac23}  }  } \geq 1 - O(m^{-\frac23}).$$
            Combining the two bounds, we get the required result.
          \end{proof}

          \subsection{Proof of Theorem~\ref{thm:gen_ub}}\label{sec:proofof:thm:gen_ub}
          The proof of \cref{thm:gen_ub} is similar to the proof of \cref{thm:ub}.
          Instead of repeating the entire proof we highlight how the proof of \cref{thm:gen_ub} differs from the proof of \cref{thm:ub}.

          Recall that for each intersection $\sigma\in \zo^p$, we have a strictly increasing function $b_\sigma\colon \R_{\geq 0}\to \R_{\geq 0}$ that defines the observed utilities as follows:
          For each item $i\in I_\sigma$ its observed utility is:
          \begin{align*}
            \hw_i\coloneqq b_{\sigma}(w_{i}),
            \yesnum\label{eq:gen_model:app}
          \end{align*}
          where $w_i$ is the item's latent utility.
          The latent utilities of each item are drawn independently from some distribution $\cD$ over non-negative numbers with cumulative distribution function $F(\cdot)$ and probability density function $\mu_\cD(\cdot)$.
          We assume that the bias functions and distribution $\cD$ satisfy the following assumption:
          \begin{assumption}\label{asmp:2:app}
            There is are constants $c,d>0$ such that
            the probability density function of $\cD$, $\mu_\cD$, is differentiable and for all $x$ in the support of $\cD$ and takes values between $c$ and $\frac{1}{c}$ (i.e., $c \leq \mu_\cD(x)\leq \frac{1}{c}$), and
            the bias function for each intersection (i.e., $b_\sigma$ for each $\sigma\in \zo^p$) is  differentiable, takes value at most  $\frac{1}{c}$ (i.e., $b_\sigma(x)\leq \frac{1}{c}$), and has a derivative between $d$ and $\frac{1}{c}$ (i.e., $d\leq b_\sigma'(x)\leq \frac{1}{c}$) for each $x$ in the support of $\cD$.
            \end{assumption}
            \noindent
            Given $\eta>0$ and $\rho>0$, define
            \begin{align*}
              \phi_b \coloneqq \inparen{c^{4}\rho\cdot\min\inbrace{\eta,1-\eta}\cdot  (b_{00}-b_{10}-b_{01}+b_{11})\circ \inparen{1-\frac{n}{m}}}^2.
            \end{align*}
            In the $p=2$ case, under \cref{asmp:2:app}, we need to show that for all $L_1\geq 0$ and $L_2\geq 0$:
            \begin{align*}
              \ratio_\cD(L,\beta)\leq
              1 - \phi_b.
            \end{align*}
            \paragraph{Overview of the proof of \cref{thm:gen_ub}}
            We follow the same strategy as the proof of \cref{thm:ub}.
            The main difference is in the proof of \cref{lem:opt_solution:app}.
            The proof of \cref{lem:opt_solution:app} uses closed form expressions for the expectation and variance of order statistics for arbitrary distributions.
            However, unlike the case with the uniform distribution, we do not have estimates or concentration inequalities for order statistics of distribution $\cD$.
            Under \cref{asmp:2:app}, we prove concentration bounds for order statistics of distribution $\cD$.
            This in turn, changes the definition of $f_\gamma$ to an expression dependent on the distribution $\cD$ and the bias functions $b_\sigma$.
            We show that a version of \cref{lem:sc_lp} holds for the new definition of $f_\gamma$.
            Then, the rest of the proof is identical to the proof of \cref{thm:ub}, but using the new definition of $f_\gamma$.

            \paragraph{Additional notation} Let $M$ be the maximum value in the support of $\cD$ and $M_b$ be the maximum of $b_\sigma(x)$ over all $x$ in the support of $\cD$ and all intersections $\sigma\in \zo^2$:
            \begin{align*}
              M \coloneqq  \sup_{x\in{\rm supp}(\cD)} x
              \qquad\text{and}\qquad
              M_b \coloneqq  \sup_{x\in{\rm supp}(\cD)}\max_{\sigma\in \zo^2} b_\sigma(x).
            \end{align*}
            Let $d$ and $D$ be the maximum and minimum values of the PDF of $\cD$ on its support:
            \begin{align*}
              d \coloneqq \inf_{x\in{\rm supp}(\cD)} \cD(x)
              \qquad\text{and}\qquad
              D \coloneqq \sup_{x\in{\rm supp}(\cD)} \cD(x).
            \end{align*}
            Under \cref{asmp:2:app}, we the following bounds on $M,M_b,d,$ and $D$
            \begin{align*}
              M, M_b\leq \frac{1}c
              \qquad\text{and}\qquad
              c \leq d\leq D\leq \frac{1}c.
              \yesnum\label{eq:bounds:gen}
            \end{align*}

            \paragraph{\bf Step 1: Reduce computing $\sx$\ and $\tx$ to a small number of variables.} This step only requires the property that: For any two items in the same intersection, the order of their latent utilities is the same as the order of their observed utilities.
            This holds because $b_\sigma$ are strictly increasing functions.
            Hence, \cref{obs:wtx_picks_top_latent_util,obs:sx_picks_top_latent_util} hold.

            \paragraph{Step 2: Compute ${\sinangle{\tx,w}}$ and ${\sinangle{\tx,\hw}}$ as function of $K$.} In this step, we change the definition of $f_\gamma$.
            In particular, instead of $f_1$ and $f_\beta$ we consider the following functions:
            \begin{align*}
              \forall
              k\in \R_{\geq 0}^4,\quad
              f_1(k)
              &\coloneqq \sum_\sigma \abs{I_\sigma} \cdot \int_{z_\sigma(k)}^{z_\sigma(0)}  x d\cD(x),
              \quad
              f_b(k)
              \coloneqq \sum_\sigma \abs{I_\sigma} \cdot \int_{z_\sigma(k)}^{z_\sigma(0)}  b_\sigma(x) d\cD(x),
              \yesnum
              \label{eq:def_f:app:extn}
            \end{align*}
            where $z_\sigma$ is
            \begin{align*}
              \forall\sigma\in \zo^2,\quad
              z_\sigma(k)\coloneqq F^{-1}\inparen{1-\frac{k_\sigma}{\abs{I_\sigma}}}.
            \end{align*}
            For the above definition, a version of \cref{lem:opt_solution:app}, namely, \cref{lem:opt_solution:app:gen} holds.
            If we have concentration inequalities for the order-statistics of $\cD$, then the proof of \cref{lem:opt_solution:app:gen} is analogous to the proof of \cref{lem:opt_solution:app}.
            We prove a concentration bound on the order statistics of $\cD$ under \cref{asmp:2:app}.
            This is the main addition in the proof of \cref{lem:opt_solution:app:gen} compared to the proof of \cref{lem:opt_solution:app}.

            The differences in the statement of \cref{lem:opt_solution:app:gen} compared to \cref{lem:opt_solution:app} are that: (1) the $O_{\rho}(\cdot)$ and $O_{\eta\rho}(\cdot)$ terms change to $O_{(\sfrac{\rho}{M})}(\cdot)$ and $O_{(\sfrac{\eta\rho}{M})}(\cdot)$, and (2) it $f_\beta$ changes to $f_b$.
            \begin{lemma}\label{lem:opt_solution:app:gen}
              Let $x$ be any selection that, for all $\sigma\in \zo^2$, selects top $K_\sigma$ items from $I_\sigma$ by observed utility.
              Where $x$ and $K$ are possibly random variables.
              With probability at least $1-O_{(\sfrac{\rho}{M})}\sinparen{m^{-\frac{1}{4}}}$,
              \begin{align*}
                \sinangle{x, w} = f_1(K) \pm O_{(\sfrac{\eta\rho}{M})}\sinparen{nm^{-\frac{1}{4}}}
                \quad\text{and}\quad
                \sinangle{x,\hw} = f_b(K) \pm O_{(\sfrac{\eta\rho}{M})}\sinparen{nm^{-\frac{1}{4}}}.
              \end{align*}
            \end{lemma}
            \noindent
            The proof of \cref{lem:opt_solution:app:gen} appears in \cref{sec:proofof:lem:opt_solution:app:gen}.

            We require $f_1$ and $f_b$ to be strongly concave, Lipschitz, and Lipschitz continuous, for appropriate parameters, to use the proof of \cref{thm:ub}.
            We show that $f_1$ and $f_b$ remain strongly concave, Lipschitz, and Lipschitz continuous, but the specific parameters change as described in the following lemma:\footnote{A function $f\colon \R^n\to \R$ is said to be $L$-Lipschitz if for all $x,y\in \R$, $\abs{f(x)-f(y)}_2\leq L\norm{x-y}_2$
            and $M$-Lipschitz continuous if for all $x,y\in \R$, $\norm{\nabla f(x)-\nabla f(y)}_2\leq M\norm{x-y}_2$.}
            \begin{lemma}\label{lem:sc_lp:gen}
              It holds that:
              \begin{itemize}\label{lem:sc_lp:ext}
                \item $f_1$ is $\frac{1}{(1-\rho)mD}$-strongly concave, $2M$-Lipschitz, and $\frac{1}{\rho m d}$-Lipschitz continuous.
                \item $f_b$ is $\frac{c}{(1-\rho)mD}$-strongly concave, $2M_b$-Lipschitz, and $\frac{1}{\rho m cd}$-Lipschitz continuous.
              \end{itemize}
              \end{lemma}
              \noindent
              \cref{lem:opt_solution:app:gen} can be proved by computing the gradient and Hessian of $f_1$ and $f_b$.
              And then, bounding their norms using \cref{asmp:2:app}.
              The proof of \cref{lem:sc_lp:gen} appears in \cref{sec:proofof:lem:sc_lp:gen}.

              \paragraph{\bf Step 3: Express $K^\star$ and $\wt{K}$ as solutions of different optimization problems.} The only change in the proof of \cref{lem:conc_of_k:app} is that it uses the new values of the strong concavity and Lipschitz constants from \cref{lem:sc_lp:gen}.
              This changes the $O_{\eta\rho}(\cdot)$ and $O_{\beta\eta\rho}(\cdot)$ terms to $O_{(\sfrac{\eta\rho}{M})}(\cdot)$ and $O_{(\sfrac{\eta\rho}{M_b})}(\cdot)$

              \paragraph{\bf Step 4: Proof of Theorem~\ref{thm:gen_ub}.} The only change in the proof of \cref{lem:using_sc} is that it uses the new values of the strong concavity and Lipschitz constants from \cref{lem:sc_lp:gen} and changes the threshold $m_0$ to a value such that
              $${m_0^\frac{1}{8}\geq \Omega\inparen{\frac{\rho\eta c}{MM_b} \cdot (b_{00}-b_{10}-b_{01}+b_{11})\circ \inparen{F^{-1}\inparen{1-\frac{n}{m}}} }.}$$
              This gives us a new version of \cref{lem:using_sc}, namely \cref{lem:using_sc:gen}.
              This version has two changes:
              First, the term $2nm(1-\rho)\cdot \phi(\beta)$ changes to $D\cdot 2nm(1-\rho)\phi_b$.
              Second, the upper bound on $\ratio_\cD(L,\beta)$ changes to the upper bound in \cref{thm:gen_ub}.
              \begin{lemma}\label{lem:using_sc:gen}
                Given $L_1,L_2\in \Z_{\geq 0}$, if $\snorm{\wt{s}-s^\star}_2\geq 2Dnm(1-\rho)\phi_b$, then
                \begin{align*}
                  \ratio_\cD(L,\beta) \leq 1 - \frac{d^2c^2}{DM_b^2}\cdot \inparen{\rho\cdot\min\inbrace{\eta,1-\eta}\cdot  (b_{00}-b_{10}-b_{01}+b_{11})\circ \inparen{F^{-1}\inparen{1-\frac{n}{m}}}}^2.
                \end{align*}
              \end{lemma}
              \noindent
              Now the rest of the proof follows by using \cref{lem:using_sc:gen,lem:using_sc:gen} and observing that \cref{asmp:2:app}, $$\frac{d^2c^2}{DM_b^2}\geq c^7.$$

              \noindent
              In particular, using \cref{lem:using_sc:gen} and \cref{eq:bounds:gen}, it suffices to show that for all $L_1,L_2\in \Z_{\geq 0}$:
              \begin{align*}
                \snorm{s^\star-\wt{s}}_2^2\geq 2D nm(1-\rho)\cdot \phi_b.
                \yesnum\label{eq:lb_on_difference:app:gen}
              \end{align*}
              Since the new definition of $f_1$ is concave, \cref{prog:1:app} remains convex.
              Hence, we can analytically solve \prog{prog:1:app} using the gradient test.
              This gives us that:
              \begin{align*}
                s^\star\coloneqq \inbrace{\abs{I_\sigma}\cdot \frac{n}{m}}_{\sigma\in \zo^2}.
                \yesnum\label{eq:val_k_star:app:extn}
              \end{align*}
              In particular, the value of $s^\star$ does not change compared to the proof of \cref{thm:ub}.
              Fix any $L_1,L_2\in \Z_{\geq 0}$.

              \smallskip
              \noindent {\bf (Case A) $\wt{s}$ satisfies at least one inequality in Equation~\eqref{eq:prog:2:grp_sz_constraint:app} with equality:}
              This case remains the same, and from \cref{eq:case_anan:ub}, we have that:
              $$\snorm{s^\star-\wt{s}}_2^2
              \geq \inparen{\rho m\max\inbrace{\eta, (1-\eta)}}^2.$$
              Since $\phi_b<\frac{\rho^2}{2\eta}\min\inbrace{\eta^2, (1-\eta)^2}$, Equation~\eqref{eq:lb_on_difference:app:gen} follows.\\

              \smallskip
              \noindent {\bf (Case B) $\wt{s}$ satisfies all inequalities in Equation~\eqref{eq:prog:2:grp_sz_constraint:app} with strict inequality:}
              The proof of \cref{eq:zero_grad:app} is the same, except the following two changes:
              First, \cref{eq:exp_grad:app} changes to
              \begin{align*}
                \inangle{\nabla f_b(s^\star), v}
                = (b_{00}-b_{10}-b_{01}+b_{11})\circ \inparen{F^{-1}\inparen{1-\frac{n}{m}}}.
                \yesnum
                \label{eq:exp_grad:app:gen}
              \end{align*}
              Here, we used the fact that the gradient of $f_b$ is:
              \begin{align*}
                \forall\sigma\in \zo^2,\quad
                \frac{d f_b(k)}{d k_\sigma}
                =
                b_{\sigma}\inparen{F^{-1}_\cD\inparen{1-\frac{k_{\sigma}}{\abs{I_{\sigma}}}}}.
              \end{align*}
              Second, the Lipshitz constant in \cref{eq:lip_const:ub} changes to $\frac{c}{(1-\rho)mD}$.
              After these changes, following the same argument as before implies that:
              \begin{align*}
                \snorm{s^\star-\wt{s}}_2^2
                &\geq\ \  2Dmn(1-\rho)\cdot \phi_b.
              \end{align*}
              Thus, \cref{lem:using_sc:gen} implies that Equation~\eqref{eq:lb_on_difference:app:gen} holds in this case also.

              \subsubsection{Proof of \cref{lem:opt_solution:app:gen}}\label{sec:proofof:lem:opt_solution:app:gen}
              \begin{proof}
                For each $\sigma\in \zo^2$, let $v^{\sexp{\sigma}}_1\geq v^{\sexp{\sigma}}_2\geq \dots\geq v^{\sexp{\sigma}}_{\abs{I_\sigma}}$ be latent utilities of items in $I_\sigma$.
                \noindent For each $i\in \N$, define
                \begin{align*}
                  s(i)\coloneqq (i-1)\cdot m^{\frac34} + 1.
                \end{align*}
                Starting from $s(1)=1$, for $i=1,2,\dots$ values at an interval of $m^{\frac34}$.
                For each intersection $\sigma\in \zo^2$, starting from the $v^{\sexp{\sigma}}_{s(1)}$, consider $v^{\sexp{\sigma}}_{s(2)}$, then $v^{\sexp{\sigma}}_{s(3)}$, and so on,
                i.e., consider
                $$v^{\sexp{\sigma}}_{s(1)},\ \ v^{\sexp{\sigma}}_{s(2)},\ \dots$$
                Since we pick every $m^{\frac34}$-th value in each intersection and each intersection $\sigma$ has $\abs{I_\sigma}\leq m$ values, we select $O(m^{\frac{1}{4}})$ values from each intersection $\sigma\in \zo^2$.
                Since there are constant number of intersections, we select $O(m^{\frac{1}{4}})$ order statistics in total.
                Let $\evG$ be the event that all of these order statistics are $m^{-\frac14}$-close to their mean, i.e., $\evG$ is the event that:
                \begin{align*}
                  \forall \sigma\in \zo^2,\ j\in \N,\ \  \st,\ \ s(j)\leq \abs{I_\sigma},\qquad
                  \abs{{v^{\sexp{\sigma}}_{s(j)}} - \Ex\insquare{{v^{\sexp{\sigma}}_{s(j)}}}} < m^{-\frac14}.
                  \yesnum\label{eq:def_evG:extn}
                \end{align*}
                We prove the following three lemmas.
                \begin{lemma}\label{lem:prob_evG:extn}
                  $\Pr[\evG]\geq 1-O_{c\rho}\inparen{m^{-\frac{1}{4}}}$.
                \end{lemma}
                \begin{lemma}\label{lem:exp_os:extn}
                  For all $\sigma\in \zo^2$ and all $i\in \inbrace{1,2,\dots,\abs{I_\sigma}}$, it holds that
                  \begin{align*}
                    \Ex\insquare{v^{\sexp{\sigma}}_{\ell}}
                    &\in F^{-1}\inparen{ 1-\frac{\ell}{\abs{I_\sigma}}} \pm O_{c\rho}\sinparen{m^{-\frac{1}{4}}}.
                  \end{align*}
                \end{lemma}
                \begin{lemma}\label{lem:close_mean}
                  For all $\sigma\in \zo^2$ and all $i,j\in \inbrace{1,2,\dots,\abs{I_\sigma}}$, it holds that
                  $$\abs{\Ex\insquare{  v^{\sexp{\sigma}}_{i} } - \Ex\insquare{v^{\sexp{\sigma}}_{j}} }\leq \frac{1}{\abs{I_\sigma}} \cdot O_{c}\inparen{\abs{i-j}} + O_{c\rho}\inparen{m^{-\frac{1}{4}}}.$$
                \end{lemma}
                \noindent {When $\cD$ is the uniform distribution on $[0,1]$, the special cases of the \cref{lem:prob_evG:extn,lem:exp_os:extn,lem:close_mean} are either well known facts or have short proofs.
                Hence, they were not stated as explicit lemmas in the proof of \cref{thm:ub}.
                When $\cD$ is a non-uniform distribution proving these lemmas requires proving concentration bounds on the order-statistics from $\cD$.
                We present the concentration bounds and the proofs of \cref{lem:prob_evG:extn,lem:exp_os:extn,lem:close_mean} in \cref{sec:proofof:lem:prob_evG:extn,sec:proofof:lem:exp_os:extn,sec:proofof:lem:close_mean}.}

                Conditioned on $\evG$, we have
                \begin{align*}
                  \sum_{\sigma\in \zo^2}{\sum_{i=1}^{K_{\sigma}} b_\sigma\sinparen{v_j^{\sigma}}}
                  \quad
                  &\leq\quad  \sum_{\sigma\in \zo^2} \sum_{j=1}^{m^{-\frac34} K_{\sigma}} m^{\frac34}\cdot b_\sigma\sinparen{v_{s(j)}^{\sigma}}
                  \tag{$v_{i}^{\sigma}$ are arranged in non-increasing order}\\
                  &\leq\quad \sum_{\sigma\in \zo^2} {\sum_{j=1}^{m^{-\frac34} K_{\sigma}} m^{\frac34}\cdot
                  b_\sigma\inparen{\Ex\insquare{v_{s(j)}^{\sigma}} + m^{-\frac14}} }
                  \tag{Using that $\evG$ implies \cref{eq:def_evG:extn}}\\
                  &\leq\quad \sum_{\sigma\in \zo^2} {\sum_{i=1}^{K_{\sigma}} b_\sigma\inparen{\Ex\insquare{v_{i}^{\sigma}} + \frac{1}{\abs{I_\sigma}}\cdot O\sinparen{m^{\frac14}}  + m^{-\frac14}} }
                  \tag{Using \cref{lem:close_mean}}\\
                  &\leq\quad \sum_{\sigma\in \zo^2} {\sum_{i=1}^{K_{\sigma}} b_\sigma\inparen{\Ex\insquare{v_{i}^{\sigma}} + O_\rho(m^{-\frac14})  } }
                  \tag{Using that $\abs{I_\sigma}\geq \rho m$}\\
                  &\leq\quad \sum_{\sigma\in \zo^2} \sum_{i=1}^{K_{\sigma}} b_\sigma\inparen{\Ex\insquare{v_{i}^{\sigma}} } + \sum_{\sigma\in \zo^2} \sum_{i=1}^{K_{\sigma}} O_{c\rho}(m^{-\frac14})
                  \tag{Using that for all $x\in {\rm supp}(\cD)$, $0\leq b_\sigma'(x)\leq \frac{1}{c}$ due to \cref{asmp:2}}\\
                  &=\quad \sum_{\sigma\in \zo^2} \sum_{i=1}^{K_{\sigma}} b_\sigma\inparen{\Ex\sinsquare{v_{i}^{\sigma}}}
                  + O_\rho(nm^{-\frac14}).
                  \tag{Using that $K_\sigma\leq n$}
                \end{align*}
                Next, using \cref{lem:exp_os:extn}, we get that
                \begin{align*}
                  \sum_{\sigma\in \zo^2}{\sum_{i=1}^{K_{\sigma}} b_\sigma\sinparen{v_j^{\sigma}}}
                  \quad
                  &= \sum_{\sigma\in \zo^2} \sum_{i=1}^{K_{\sigma}} b_\sigma\inparen{F^{-1}\inparen{\frac{\abs{I_\sigma}-i}{\abs{I_\sigma}+1} }}
                  + O_\rho(nm^{-\frac14})\\
                  &\leq \sum_{\sigma\in \zo^2} \sum_{i=1}^{K_{\sigma}} b_\sigma\inparen{ F^{-1}\inparen{\frac{\abs{I_\sigma}-i}{\abs{I_\sigma}} }}
                  + O_{\rho c}(m^{-1})
                  + O_\rho(nm^{-\frac14})
                  \tag{Using that $\abs{I_\sigma}\geq \rho m$ and for all $x\in {\rm supp}(\cD)$, $b_\sigma'(x),\inparen{F^{-1}}'(x)\leq \frac{1}{c}$  due to \cref{asmp:2}}
                 \end{align*}
                 \begin{align*}
                  \white{.}\qquad\qquad\qquad &\leq \sum_{\sigma\in \zo^2} \sum_{i=1}^{K_{\sigma}} b_\sigma\inparen{ F^{-1}\inparen{\frac{\abs{I_\sigma}-i}{\abs{I_\sigma}} }}
                  + O_{\rho c}(nm^{-\frac14})
                  \tag{Using that $K_\sigma\leq n$}\\
                  &\leq \sum_{\sigma\in \zo^2} \int_{x=\abs{I_\sigma}-K_\sigma}^{\abs{I_\sigma}} b_\sigma\inparen{ F^{-1}\inparen{\frac{x}{\abs{I_\sigma}} +\frac{1}{\rho m} }} dx
                  + O_{\rho c}(nm^{-\frac14})
                  \tag{Using that $\abs{I_\sigma}\geq \rho m$}\\
                  &\leq \sum_{\sigma\in \zo^2} \int_{x=\abs{I_\sigma}-K_\sigma}^{\abs{I_\sigma}} \inparen{b_\sigma\inparen{ F^{-1}\inparen{\frac{x}{\abs{I_\sigma}}}} + O_{\rho c}\inparen{m^{-1}}} dx
                  + O_{\rho c}(nm^{-\frac14})
                  \tag{Using that $K_\sigma\leq n$, $\abs{I_\sigma}\geq \rho m$, and for all $x\in {\rm supp}(\cD)$, $b_\sigma'(x),\inparen{F^{-1}}'(x)\leq \frac{1}{c}$ due to \cref{asmp:2} }\\
                  &\leq \sum_{\sigma\in \zo^2} \int_{x=\abs{I_\sigma}-K_\sigma}^{\abs{I_\sigma}} b_\sigma\inparen{ F^{-1}\inparen{\frac{x}{\abs{I_\sigma}}}} dx
                  + O_{\rho c}(nm^{-\frac14})
                  \tagnum{Using that $n\geq \eta m$, hence $O_\rho(nm^{-\frac14})\gg \frac{K_\sigma}{m}$}\\
                  &= f_b(K)
                  + O_{\rho c}(nm^{-\frac14}).
                  \tagnum{Using \cref{eq:def_f:app:extn}}
                  \customlabel{eq:bound_on_obs_util:new:extn}{\theequation}
                \end{align*}
                Similarly, we can show the following lower bound on $\sinangle{\tx,\hw}$ conditioned on $\evG$
                \begin{align*}
                  \sum_{\sigma\in \zo^2} \sum_{i=1}^{K_{\sigma}} b_\sigma\sinparen{v_j^{\sigma}}
                  &\geq  f_b(K) - O_{\eta\rho}(nm^{-\frac14}).\yesnum\label{eq:bound_on_obs_util2:new:extn}
                \end{align*}
                Substituting $\gamma=\beta$ in Equations~\eqref{eq:bound_on_obs_util:new:extn} and \eqref{eq:bound_on_obs_util2:new:extn}, and using that for all intersections $\sigma\in \zo^2$, $x$ selects top $K_\sigma$ items from $I_\sigma$ by observed utility, we get that conditioned on $\evG$
                \begin{align*}
                  \sinangle{x,\hw} = \sum_{\sigma\in \zo^2}{\sum_{i=1}^{K_{\sigma}} b_\sigma\sinparen{v_j^{\sigma}}} = f_\beta(K) \pm O_{\eta\rho}(nm^{-\frac14}).
                  \end{align*}
                  Further, choosing $\beta_\sigma$ to be the identify function in Equations~\eqref{eq:bound_on_obs_util:new:extn} and \eqref{eq:bound_on_obs_util2:new:extn}, and using that for all intersections $\sigma\in \zo^2$, $x$ selects top $K_\sigma$ items from $I_\sigma$ \mbox{by observed utility, we get that conditioned on $\evG$}
                  \begin{align*}
                    \sinangle{x, w} = \sum_{\sigma\in \zo^2}{\sum_{i=1}^{K_{\sigma}} v_j^{\sigma}} = f_1(K) \pm O_{\eta\rho}(nm^{-\frac14}).
                  \end{align*}
                \end{proof}

                \subsubsection{Proof of \cref{lem:sc_lp:ext}}\label{sec:proofof:lem:sc_lp:gen}
                \begin{proof}[Proof of \cref{lem:sc_lp:ext}]
                  The gradient of $f_b$ is:
                  \begin{align*}
                    \forall \sigma\in \zo^2,\quad
                    \inparen{\nabla f_b(x)}_\sigma
                    = b_\sigma\inparen{F^{-1}_\cD\inparen{1-\frac{k_\sigma}{\abs{I_\sigma}}}}.
                  \end{align*}
                  Since $F^{-1}_\cD\inparen{1-\frac{k_\sigma}{\abs{I_\sigma}}}$ is in the support of $\cD$,
                  $b_\sigma\inparen{F^{-1}_\cD\inparen{1-\frac{k_\sigma}{\abs{I_\sigma}}}}$ is upper bounded by $M_b$.
                  Hence,
                  \begin{align*}
                    \forall \sigma\in \zo^2,\quad
                    \inparen{\nabla f_b(x)}_\sigma
                    \leq M_b.
                  \end{align*}
                  It follows that
                  \begin{align*}
                    \norm{\nabla f_b(x)}_2 \leq 2M_b.
                  \end{align*}
                  Thus, $f_b$ is $2M_b$-Lipshitz.
                  The Hessian of $f_b$ is:
                  \begin{align*}
                    \forall \sigma\in \zo^2,\quad
                    \inparen{\nabla^2 f_b(x)}_{\sigma\sigma}
                    &= -\frac{1}{\abs{I_\sigma}} \cdot \frac{b_\sigma'\inparen{F^{-1}_\cD\inparen{1-\frac{k_\sigma}{\abs{I_\sigma}}}}}{\mu_\cD\inparen{F^{-1}_\cD\inparen{1-\frac{k_\sigma}{\abs{I_\sigma}}}}},\\
                    \forall \sigma,\tau\in \zo^2, \sigma\neq \tau,\quad
                    \inparen{\nabla^2 f_b(x)}_{\sigma\tau}
                    &= 0.
                  \end{align*}
                  Hence, the Hessian is a diagonal matrix.
                  We can upper bound its diagonal entries using that $b'_\sigma(x)\geq c$ and $\mu_\cD(x)\leq D$
                  for all $x\in {\rm supp}(\cD)$ and $\abs{I_\sigma}\leq (1-3\rho)m$ (this holds as for all $\tau\in \zo^2$ $\tau\neq \sigma$, $\abs{I_\tau} \geq \rho m$):
                  \begin{align*}
                    \forall \sigma\in \zo^2,\quad
                    \inparen{\nabla^2 f_b(x)}_{\sigma\sigma}
                    &\leq -\frac{c}{(1-3\rho)mD} %
                    \leq -\frac{c}{(1-\rho)mD}.
                  \end{align*}
                  It follows that $f_b$ is $\frac{c}{(1-\rho)mD}$-strongly concave.
                  Similarly, we can lower bound the diagonal entries of the Hessian by using that $b'_\sigma(x)\leq \frac{1}{c}$ and $\mu_\cD(x)\geq d$
                  for all $x\in {\rm supp}(\cD)$ and $\abs{I_\sigma}\geq \rho m$:
                  \begin{align*}
                    \forall \sigma\in \zo^2,\quad
                    \inparen{\nabla^2 f_b(x)}_{\sigma\sigma}
                    &\geq -\frac{1}{\rho m dc}. %
                  \end{align*}
                  Hence, $f_b$ is $\frac{1}{\rho m dc}$-Lipshitz continuous.

                  \smallskip

                  Replacing $f_b$, $M_b$, $b_\sigma(x)$, and $c$ by $f_1$, $M$, $x$, and $1$ in the above proof
                  it follows that $f_1$ is $2M$-Lipshitz, $\frac{1}{(1-\rho)mD}$-strongly concave, and $\frac{1}{\rho m d}$-Lipshitz continuous.
                \end{proof}

                \subsubsection{Proof of \cref{lem:prob_evG:extn}}\label{sec:proofof:lem:prob_evG:extn}
                \begin{proof}[Proof of \cref{lem:prob_evG:extn}]
                  Fix any intersection $\sigma\in \zo^2$ and select an index $j\in \N$ such that $s(j)\leq \abs{I_\sigma}$.
                  Here, $v^{\sexp{\sigma}}_{s(j)}$ is the $\inparen{\abs{I_\sigma}-s(j)}$-th order statistic of $\abs{I_\sigma}$ draws from $\cD$.
                  Using this, it can be shown that $F\sinparen{v^{\sexp{\sigma}}_i}$ is the $\inparen{\abs{I_\sigma}-i}$-th order statistic of $\abs{I_\sigma}$ draws from $\unif$.
                  It suffices to prove the following concentration inequality:
                  \begin{align*}
                    \Pr\insquare{  \abs{v^{\sexp{\sigma}}_{s(j)} - \Ex\insquare{v^{\sexp{\sigma}}_{s(j)}}}
                    \geq {\Theta_c\inparen{m^{-\frac14}}}  } \leq O_{\rho c}{\inparen{m^{-\frac12}}}.
                    \yesnum\label{eq:to_prove:extn}
                  \end{align*}
                  If we prove \cref{eq:to_prove:extn}, then lemma follows by taking a union bound over all pairs of $\sigma\in \zo^2$ and select an $j\in \N$ such that $s(j)\leq \abs{I_\sigma}$.
                  Here, we use the fact that there are $O(m^{\frac{1}{4}})$ such pairs of $\sigma$ and $j$.

                  At a high level, the proof of this concentration inequality follows because $F\sinparen{v^{\sexp{\sigma}}_i}$ is concentrated around its mean and $F^{-1}$ has a bounded gradient (by \cref{asmp:2}).
                  Since $\sigma\in \zo^2$ and $j\in \N$ are fixed, to simplify the notation we drop the super-script and sub-script from $v^{\sexp{\sigma}}_{s(j)}$ in the rest of this proof.
                  That is, we use $v$ to denote $v^{\sexp{\sigma}}_{s(j)}$.

                  As argued earlier, $F\sinparen{v}$ is distributed as $U_{(\abs{I_\sigma}-s(j):\abs{I_\sigma})}$.
                  From \cref{fact:os}, we have that for all $\sigma\in \zo^2$ and $i\in I_\sigma$,
                  $$\mathrm{Var}\insquare{F\sinparen{v}} = \mathrm{Var}\insquare{U_{(\abs{I_\sigma}-i:\abs{I_\sigma})}} \leq \frac{1}{\abs{I_\sigma}}.$$
                  Since $\abs{I_\sigma}\geq \rho m$,
                  $$\mathrm{Var}\insquare{F\sinparen{v}} \leq \frac{1}{\rho m}.$$
                  Hence, using the Chebyshev's inequality~\cite{motwani1995randomized}, we get:
                  \begin{align*}
                    \Pr\insquare{  \abs{F\sinparen{v} - \Ex\insquare{F\sinparen{v}}}
                    \geq m^{-\frac14}  } \leq \rho^{-\frac{1}{2}}m^{-\frac12}. %
                    \yesnum\label{eq:vi_conc:new:extn}
                  \end{align*}
                  Let $\evJ$ be the event that $\abs{F\sinparen{v} - \Ex\insquare{F\sinparen{v}}}\leq  m^{-\frac14}$.
                  Thus, by \cref{eq:vi_conc:new:extn},
                  \begin{align*}
                    \Pr[\evJ]\geq 1-\rho^{-\frac{1}{2}}m^{-\frac12}.
                    \yesnum\label{eq:probJ:extn}
                  \end{align*}
                  Conditioned on $\evJ$, we have that:
                  \begin{align*}
                    F(v) \in \Ex\insquare{F(v)} \pm m^{-\frac{1}{4}}.
                  \end{align*}
                  Applying $F^{-1}$ on both sides (using that $F^{-1}$ is an increasing and continuous) we get that:
                  Conditioned on $\evJ$
                  \begin{align*}
                    v &\in F^{-1}\inparen{\Ex\insquare{F(v)} \pm m^{-\frac{1}{4}}}.
                    \yesnum\label{eq:1:extn}
                  \end{align*}
                  By \cref{asmp:2}, we have that $\mu_\cD(x)\leq \frac{1}{c}$ for all $x\in {\rm supp}(\cD)$.
                  Hence, in particular, for all $x\in  \Ex\insquare{F(v)} \pm m^{-\frac{1}{4}}$, $\mu_\cD(v)\leq \frac{1}{c}$.
                  Equivalently, $F'(x) \leq \frac{1}{c}$ for all $x\in  \Ex\insquare{F(v)} \pm m^{-\frac{1}{4}}$.
                  This implies that $(F^{-1})'(x) \leq \frac{1}{c}$ for all $x\in  \Ex\insquare{F(v)} \pm m^{-\frac{1}{4}}$.
                  Using this, \cref{eq:1:extn}, and the fact that $F^{-1}$ is increasing we get that: Conditioned on $\evJ$
                  \begin{align*}
                    v &\in F^{-1}\inparen{\Ex\insquare{F(v)}} \pm O_c\inparen{m^{-\frac{1}{4}}}
                    \yesnum\label{eq:intv_bound4:extn}
                  \end{align*}
                  Hence, it must hold that
                  \begin{align*}
                    \Ex\insquare{v\given \evJ} &\in F^{-1}\inparen{\Ex\insquare{F(v)}} \pm O_c\inparen{m^{-\frac{1}{4}}}.
                    \yesnum\label{eq:intv_bound1:extn}
                  \end{align*}
                  But from \cref{fact:whp} and \cref{eq:probJ:extn}, we have that
                  \begin{align*}
                    \Ex\insquare{v}
                    &\in \Ex\insquare{v\given \evJ} \pm \rho^{-\frac{1}{2}}m^{-\frac12}.
                    \yesnum\label{eq:intv_bound2:extn}
                  \end{align*}
                  Chaining the inequalities in \cref{eq:intv_bound1:extn,eq:intv_bound2:extn} we get that
                  \begin{align*}
                    \Ex\insquare{v}
                    &\in F^{-1}\inparen{\Ex\insquare{F(v)}} \pm O_{c\rho}\inparen{m^{-\frac{1}{4}}}.
                    \yesnum\label{eq:intv_bound6:extn}
                  \end{align*}
                  Equivalently, that
                  \begin{align*}
                    F^{-1}\inparen{\Ex\insquare{F(v)}}
                    &\in \Ex\insquare{v}  \pm O_{c\rho}\inparen{m^{-\frac{1}{4}}}.
                    \yesnum\label{eq:intv_bound3:extn}
                  \end{align*}
                  Chaining the inequalities in \cref{eq:intv_bound4:extn,eq:intv_bound3:extn}, we get:
                  Conditioned on $\evJ$
                  \begin{align*}
                    v &\in \Ex\insquare{v} \pm O_{c\rho}\inparen{m^{-\frac{1}{4}}}.
                  \end{align*}
                  Thus, \cref{eq:to_prove:extn} follows from the above equation and \cref{eq:probJ:extn}.
                \end{proof}
                \subsubsection{Proof of \cref{lem:exp_os:extn}}\label{sec:proofof:lem:exp_os:extn}
                \begin{proof}[Proof of \cref{lem:exp_os:extn}]
                  Fix any intersection $\sigma\in \zo^2$.
                  Since $\sigma$ is fixed, we drop the super-script from $v^{\sexp{\sigma}}_{i}$.
                  That is, we use $v_{i}$ to denote $v^{\sexp{\sigma}}_{i}$.
                  The proof of this lemma is similar to the proof of \cref{lem:prob_evG:extn}.
                  In particular, we borrow the following fact from \cref{lem:prob_evG:extn}, which is proved in \cref{eq:intv_bound6:extn}:
                  For any $\ell\in \inbrace{1,2,\dots,\abs{I_\sigma}}$:
                  \begin{align*}
                    \Ex\insquare{v_{\ell}}
                    &\in F^{-1}\inparen{\Ex\insquare{F\sinparen{v_{\ell}}}} \pm O_{c\rho}\sinparen{m^{-\frac{1}{4}}}.
                    \yesnum\label{eq:b3:extn}
                  \end{align*}
                  Since $v_{\ell}$ is the $\inparen{\abs{I_\sigma}-\ell}$-th order statistic of $\abs{I_\sigma}$ draws from $\cD$, it follows that $F\sinparen{v_\ell}$ is the $\inparen{\abs{I_\sigma}-\ell}$-th order statistic of $\abs{I_\sigma}$ draws from $\unif$.
                  Hence, using \cref{fact:os}, it follows that
                  \begin{align*}
                    \Ex\insquare{F\sinparen{v_{\ell}}}
                    = \frac{\abs{I_\sigma}-\ell}{\abs{I_\sigma}+1}
                    &= 1-\frac{\ell}{\abs{I_\sigma}} \pm O_\rho(m^{-1}).
                    \tag{Using that $\abs{I_\sigma}\geq \rho m$}
                  \end{align*}
                  Applying $F^{-1}$ to both sides, we get that
                  \begin{align*}
                    F^{-1}\inparen{ \Ex\insquare{F\sinparen{v_{\ell}}}   }
                    &= F^{-1}\inparen{ 1-\frac{\ell}{\abs{I_\sigma}} \pm O_{\rho}(m^{-1})}.
                    \yesnum\label{eq:b1:extn}
                  \end{align*}
                  \noindent By \cref{asmp:2}, we have that $\mu_\cD(x)\leq \frac{1}{c}$ for all $x\in {\rm supp}(\cD)$.
                  This implies that
                  \begin{align*}
                    \forall x\in {\rm supp}(\cD),\quad
                    (F^{-1})'(x) \leq \frac{1}{c}.
                    \yesnum\label{eq:b5:extn}
                  \end{align*}
                  Using this, \cref{eq:b1:extn}, and the fact that $F^{-1}$ is increasing we get that:
                  \begin{align*}
                    F^{-1}\inparen{ \Ex\insquare{F\sinparen{v_{\ell}}}   }
                    &= F^{-1}\inparen{ 1-\frac{\ell}{\abs{I_\sigma}}} \pm O_{c\rho}(m^{-1}).
                    \yesnum\label{eq:b2:extn}
                  \end{align*}
                  Combining the above equation with \cref{eq:b3:extn}, it follows that
                  \begin{align*}
                    \Ex\insquare{v_{\ell}}
                    &\in F^{-1}\inparen{ 1-\frac{\ell}{\abs{I_\sigma}}} \pm O_{c\rho}\sinparen{m^{-\frac{1}{4}}}.
                  \end{align*}
                  \vspace{-5mm}
                  
                \end{proof}

                \subsubsection{Proof of \cref{lem:close_mean}}\label{sec:proofof:lem:close_mean}
                \begin{proof}
                  Fix any intersection $\sigma\in \zo^2$.
                  Since $\sigma$ is fixed, we drop the super-script from $v^{\sexp{\sigma}}_{i}$.
                  That is, we use $v_{i}$ to denote $v^{\sexp{\sigma}}_{i}$.
                  By \cref{asmp:2}, we have that $\mu_\cD(x)\leq \frac{1}{c}$ for all $x\in {\rm supp}(\cD)$.
                  This implies that
                  \begin{align*}
                    \forall x\in {\rm supp}(\cD),\quad
                    (F^{-1})'(x) \leq \frac{1}{c}.
                    \yesnum\label{eq:c1:extn}
                  \end{align*}
                  \noindent We have
                  \begin{align*}
                    \abs{\Ex\insquare{  v_{i} } - \Ex\insquare{v_{j}} }
                    &\leq
                    \abs{ F^{-1}\inparen{ 1-\frac{i}{\abs{I_\sigma}}} - F^{-1}\inparen{ 1-\frac{j}{\abs{I_\sigma}}} }
                    + O_{c\rho}\inparen{m^{-\frac{1}{4}}}
                    \tag{Using \cref{lem:exp_os:extn}}\\
                    &\leq
                    \abs{F^{-1}\inparen{ 1-\frac{j}{\abs{I_\sigma}}} - F^{-1}\inparen{ 1-\frac{j}{\abs{I_\sigma}}} + \frac{1}{\abs{I_\sigma}}\cdot  O_{c}\inparen{\abs{i-j}} }
                    + O_{c\rho}\inparen{m^{-\frac{1}{4}}}
                    \tag{Using \cref{eq:c1:extn} and the fact that $F^{-1}$ is a non-decreasing function}\\
                    &\leq
                    \frac{1}{\abs{I_\sigma}} \cdot O_{c}\inparen{\abs{i-j}}
                    + O_{c\rho}\inparen{m^{-\frac{1}{4}}}.
                  \end{align*}
                  
                  \vspace{-5mm}
                  
                \end{proof}

                \vspace{-5mm}

                \subsection{Proof of Proposition~\ref{prop:robustness_ub_complement}}\label{sec:proofof:prop:robustness_ub_complement}
                In this section, we prove \cref{prop:robustness_ub_complement}, which complements \cref{thm:gen_ub} when the implicit bias functions satisfy the additive relation $b_{00}(x)+b_{11}(x)=b_{10}(x)+b_{01}(x)$ for all $x\in {\rm supp}(\cD)$.
                \begin{proposition}\label{prop:robustness_ub_complement}
                  Suppose the fraction of selected candidates is between $\eta$ and $1-\eta$ for some constant $\eta>0$ (i.e., $\eta <\frac{n}{m}<1-\eta$) and all intersections have size at least $\rho m$ (i.e., for all $\sigma\in \zo^2$,  $\abs{I_\sigma}\geq \rho m$).
                  Under the implicit bias model in Equation~\eqref{eq:gen_model}, for any continuous distribution $\cD$ with non-negative support and set of strictly increasing functions $b_\sigma\colon \R_{\geq 0}\to \R_{\geq 0}$, where there is one function for {each intersection $\sigma\in\zo^2$, if \cref{asmp:2} and the following equation holds:}
                  \begin{align*}
                    (b_{00}-b_{10}-b_{01}+b_{11})\circ \inparen{F_\cD^{-1}\inparen{1-\frac{n}{m}}} = 0
                    \yesnum\label{asmp:3}
                  \end{align*}
                  then there is a threshold $m_0\in \N$, such that when the number of items are larger than this threshold, $m\geq m_0$,
                  there are non-intersectional lower bounds $L_1,L_2\geq 0$
                  such that the utility ratio is at least $1-\eps$, i.e., $$\ratio_\cD(L,\beta)\geq {1 - \eps}.$$
                \end{proposition}

                \noindent We claim that the following non-intersection constraints satisfy the claim in \cref{prop:robustness_ub_complement}:
                \begin{align*}
                  L_1 = \frac{\abs{G_1}}{m}
                  \quad\text{and}\quad
                  L_2 = \frac{\abs{G_2}}{m}.
                \end{align*}
                The proof borrows results from the proof of \cref{thm:gen_ub}.
                Beyond those results, the main step is to show that for these values of $L_1$ and $L_2$, if \cref{asmp:3} holds then $s^\star=\wt{s}$.
                Where $s^\star$ and $\wt{s}$ are the optimizers of Programs~\eqref{prog:1:app} and \eqref{prog:2:app} respectively.
                If $s^\star=\wt{s}$, then the proof can be completed as follows:
                Let
                \begin{align*}
                  \delta\coloneqq \frac{c\eta\rho}{MM_b}.
                  \yesnum
                  \label{eq:def_delta:app}
                \end{align*}
                Let $\evH$ be the event that the following hold:
                \begin{align*}
                  \sinangle{\sx, w} &= f_1(K^\star) \pm O_{\delta}\sinparen{nm^{-\frac{1}{4}}},\quad\text{\white{and}}\quad
                  \sinangle{\sx,\hw} = f_\beta(K^\star) \pm O_{\delta}\sinparen{nm^{-\frac{1}{4}}},
                  \yesnum\label{eq:def_evH1}\\
                  \sinangle{\tx, w} &= f_1(\wt{K}) \pm O_{\delta}\sinparen{nm^{-\frac{1}{4}}},\quad\ \ \hspace{0.5mm} \text{and}\quad\ \
                  \sinangle{\tx,\hw} = f_\beta(\wt{K}) \pm O_{\delta}\sinparen{nm^{-\frac{1}{4}}}.
                  \yesnum\label{eq:def_evH2}
                \end{align*}
                From \cref{lem:opt_solution:app:gen} we get that
                \begin{align*}
                  \Pr\insquare{\evH}\geq 1-O_{\delta}\sinparen{m^{-\frac{1}{4}}}.
                  \yesnum\label{eq:prob_evH}
                \end{align*}
                Conditioned on $\evH$, \cref{eq:def_evH1,eq:def_evH2} hold, and hence:
                \begin{align*}
                  \abs{\sinangle{\sx, w}-f_1(s^\star)}
                  &\leq \abs{f_1(K^\star) - f_1(s^\star) + O_\delta\sinparen{nm^{-\frac{1}{4}}}}\\
                  &\leq 2M\cdot \norm{K^\star-s^\star}_2 + {O_{\delta}(nm^{-\frac{1}{8}})}.
                  \tag{Using that $f_1$ is $2M$-Lipschitz proved in \cref{lem:sc_lp:gen}}
                \end{align*}
                \cref{lem:conc_of_k:app} also holds in for the implicit bias model in \cref{eq:gen_model} and distributions $\cD$; see \cref{sec:proofof:thm:gen_ub}. Using it, it follows that: Conditioned on $\evH$
                \begin{align*}
                  \abs{\sinangle{\sx, w}-f_1(s^\star)}
                  &\leq {O_{\delta}(nm^{-\frac{1}{8}})}.\yesnum\label{eq:exp_in_k:gen}
                \end{align*}
                Replacing $\sx$, $K^\star$, and $s^\star$ by $\tx$ in the above argument, by $\wt{K}$, and $\wt{s}$ we get that: Conditioned on $\evH$
                \begin{align*}
                  \abs{{\sinangle{\tx, w}}-f_1(\wt{s})}
                  \leq {O_{\delta}(nm^{-\frac{1}{8}})}.
                  \yesnum\label{eq:exp_in_k2:gen}
                \end{align*}
                Since $s^\star=\wt{s}$, using the triangle inequality with \cref{eq:exp_in_k:gen,eq:exp_in_k2:gen} implies that: Conditioned on $\evH$
                \begin{align*}
                  \abs{{\sinangle{\tx, w}}-{\sinangle{\sx, w}}} \leq {O_{\delta}(nm^{-\frac{1}{8}})}.
                  \yesnum\label{eq:lb_comp:app}
                \end{align*}
                Then we can prove the result as follows:
                \begin{align*}
                  \ratio_{\unif}(L,\beta)
                  &\coloneqq \Ex_{w}\insquare{ \frac{ \sinangle{\tx,w} }{ \sinangle{\sx,w} } }\\
                  &\geq \Ex_{w}\insquare{ \frac{ \sinangle{\tx,w} }{ \sinangle{\sx,w} } \given \evH} - O_\delta(m^{-\frac{1}{4}})
                  \tag{Using \cref{fact:whp} annd \cref{eq:prob_evH}}\\
                  &\geq \Ex_{w}\insquare{ \frac{\sinangle{\sx,w} -  O_{\delta}(nm^{-\frac{1}{8}})}{ \sinangle{\sx,w} } \given \evH} - O_\delta(m^{-\frac{1}{4}})
                  \tag{Using \cref{eq:lb_comp:app}}\\
                  &= 1 - O_{\delta}(m^{-\frac{1}{8}}).
                  \tag{Using that $f_1(K^\star)\geq \frac{n}{4}$ and $\evH$ implies \cref{eq:def_evH1}}
                \end{align*}
                Thus, for any value $m_0$ such that
                $$m_0^\frac{1}{8}\geq  \Omega\inparen{ \frac{1}{\eps\delta} } \ \ \Stackrel{\eqref{eq:def_delta:app}}{=} \ \  \Omega\inparen{ \frac{MM_b}{\eps c\eta\rho} },$$
                it holds that $\ratio_{\unif}(L,\beta) \geq 1-\eps.$

                \smallskip
                \noindent{\bf Proof that $s^\star=\wt{s}$.}
                Recall that $s^\star$ is:
                \begin{align*}
                  s^\star \coloneqq \inbrace{\abs{I_\sigma}\cdot \frac{n}{m}}_{\sigma\in \zo^2}.
                  \yesnum\label{eq:val_k_star:app:lb}
                \end{align*}
                One can verify that $s^\star$ is feasible for Program~\eqref{prog:2:app}.
                Moreover, $s^\star$ satisfies the constraints~\eqref{eq:prog2:lb1} and \eqref{eq:prog2:lb2} with equality.
                Note that Program~\eqref{prog:2:app} is concave, its feasible region is a subset of $\inbrace{k\in \R^4\colon \sum_\sigma k_\sigma=n}.$
                and
                $s^\star$ satisfies constraints~\eqref{eq:prog2:lb1} and \eqref{eq:prog2:lb2} with equality.
                These imply that $s^\star$ is optimal for Program~\eqref{prog:2:app} if and only if $\nabla f_b(s^\star)$ has no component in the linear-subspace of $\inbrace{k\in \R^4\colon \sum_\sigma k_\sigma=0}$ which is perpendicular to constraints~\eqref{eq:prog2:lb1} and \eqref{eq:prog2:lb2}.
                One can show that the corresponding linear-space is one dimensional and contains the vector: $v\coloneqq (1,-1,-1,1)$.
                Hence, $s^\star$ is optimal for Program~\eqref{prog:2:app} if and only if
                \begin{align*}
                  \inangle{\nabla f_b(s^\star), v} = 0.
                \end{align*}
                Substituting the value of $\nabla f_b(s^\star)$ from \cref{eq:exp_grad:app:gen}, we get
                \begin{align*}
                  \inangle{\nabla f_b(s^\star), v}
                  &= (b_{00}-b_{10}-b_{01}+b_{11})\circ \inparen{F_\cD^{-1}\inparen{1-\frac{n}{m}}}
                  = 0. \tag{Using \cref{asmp:3}}
                \end{align*}

\section{Conclusions, Limitations, and Future Work}

In this work we studied the set selection problem in the presence of implicit bias where each item may belong to any intersectional socially-salient groups.
Focusing on lower-bound constraints and randomly (and independently) drawn utilities, we first showed that {no} non-intersectional constraints achieve near-optimal utility.
We then presented  intersectional constraints that recover almost all utility for arbitrarily chosen biases for each intersectional group.
Thus, the {advantage of intersectional constraints} is substantial over a dimension-by-dimension approach.

Our work raises several questions.
While we show that the upper bound on the expected utility ratio in the non-intersectional case is 8/9 when $\beta$'s tend to zero {and utilities have the uniform distribution}, we do not know if this 8/9 is tight.
Moreover, it would be interesting to investigate if this 8/9  bound can go further down {for the distribution of utility changes or} as the number of groups $p$ increases.
Absent constraints, we know from \cref{Prop:range} that a ratio of $1/2$ can always be guaranteed for the uniform distribution.

More generally, our analysis suffers from several limitations, also sometimes found in closely related literature.
First, the assumption of independent random draws of utilities is critical but may not be realistic in practice.
Second, bias is assumed to be exogenous and due to the lack of a temporal element {in bias,} interventions cannot be evaluated in terms of their long-term equilibrium effects.
Third, and relating to the previous point, the interaction between affirmative action via setting lower-bound constraints and the behavior of individuals or groups is not taken into account.
To this end, the recently dropped law suit by the U.S. Department of Justice against Yale University\footnote{\url{https://www.nytimes.com/2021/02/03/us/yale-admissions-affirmative-action.html}}, that accused admission to be biased against Asian-American and white American, exemplifies how by prioritizing some groups others necessarily feel or are deprioritized.

Moreover, lower-bound constraints -- or affirmative action --  is just one intervention to curb the adverse effects of implicit bias.
Any effective approach against implicit bias must consider a diverse set interventions, including information campaigns and (re-)structured evaluations, complemented by increased accountability and transparency to enforce and guide these interventions.
Here, our work also raises the question of how policy makers can incorporate intersectionality {in other interventions and the potential advantages of doing so.}

The presence of bias and discrimination unquestionably remains one of the pressing issues of our society.
Implicit bias -- vis-\`a-vis explicit discrimination -- is too often viewed as an inevitable byproduct of human nature.
To holistically address how policies can help rather than hinder the persistence of unjust outcomes it is important to unpack many connected questions: Which groups merit protection and when? How does affirmative action regarding one group affect another? Should interventions aim at reducing the existence of implicit bias or is it sufficient if they counter its effects?

\paragraph{Acknowledgments.}
 This project is supported in part by an NSF Award (CCF-2112665).
 We would like to thank Jean-Paul Carvalho, Elisa Celis, and Patrick Loiseau for useful discussions.

\newpage

\bibliographystyle{plain}
\bibliography{bib-v1.bib}

\appendix

\newpage

\section{Properties of the Utility Ratio}\label{sec:app_utility_ratio}
To compare the efficacy of the constraints across different distributions, we would like $\ratio_{\cD}(L,\beta)$ to take similar values for different choices of $\cD$.
Recall that we focus on distributions with a non-negative support.
Clearly, for these $\ratio_{\cD}(L,\beta)$ lies between 0 and 1.
This is because $\sinangle{\tx,w}$ is always non-negative, and it is at most $\sinangle{\sx,w}$.
However, it is possible that its range, while a subset of $[0,1]$, changes significantly between distributions.
The next proposition shows that this is not true.
\begin{proposition}[\textbf{Invariance of range}]\label{Prop:range}\label{claim:range_of_r}
  For all continuous distributions $\cD$ and implicit bias parameters $\beta_1,\dots \beta_p\in (0,1]$, {absent any constraints (i.e., $L=0$), the range of $\ratio_{\cD}(L,\beta)$ is
  $\left(\frac{\Ex_{v\sim \cD}[v]}{\sup(\supp(\cD))}, 1\right]$.}
\end{proposition}
\noindent As a special case, the above proposition shows that the range of the utility ratio is invariant across all distributions with the same mean and support.
Its proof is straightforward, and appears in \cref{sec:proofof:claim:range_of_r}.
Second, we have the following property.
\begin{proposition}[\textbf{Invariance to scaling}]\label{claim:invariance_of_r}
  Suppose $\cD^\prime$ is a scaled version of $\cD$.\footnote{Consider $w\sim \cD$. $\cD^\prime$ is the distribution of $w\cdot \gamma$ for some $\gamma>0$.}
  For all implicit bias parameters $\beta_1,\dots \beta_p\in (0,1]$ and constraints $L$,
  $\ratio_{\cD}(L,\beta) =
  \ratio_{\cD^\prime}(L,\beta).$
\end{proposition}
\noindent \cref{claim:invariance_of_r} shows that the utility ratio, $\ratio_{\cD}(L,\beta)$ is invariant to scaling of the distribution.
Its proof is also straightforward, and appears in \cref{sec:proofof:claim:invariance_of_r}

\subsection{Proof of \cref{claim:range_of_r}}\label{sec:proofof:claim:range_of_r}
\begin{proof}
  Recall that $\sx{}$ is the selection maximizing the latent utility subject to the constraint $\sum_{i=1}^n x_i= n$.
  Let $\hx$ be the selection maximizing the observed utility subject to the constraint $\sum_{i=1}^n x_i= n$.
  By the optimality of \sx{}, we have
  \begin{align*}
    \forall w,\ \ \sinangle{\sx{}, w} \geq \inangle{\tx, w} \quad
    &\iff \forall w, \quad \frac{\inangle{\tx, w}}{\sinangle{\sx{}, w}} \leq 1\quad\\
    &\ \implies
    \qquad\quad \ratio(L,\beta) = \Ex_w\insquare{\frac{\inangle{\tx, w}}{\sinangle{\sx{}, w}}} \leq 1.
  \end{align*}
  \noindent It remains to prove the lower bound on $\ratio(L,\beta)$.
  To simplify the notation, for each intersection $\sigma\in \zo^p$ define
  \begin{align*}
    \beta_{\sigma}\coloneqq \prod\nolimits_{\ell\in [p]: \sigma_\ell=1}\beta_\ell.
  \end{align*}
  Then, let $\sigma_1,\sigma_2,\dots,\sigma_{2^p}\in \zo^p$ be the order of the intersections by non-increasing order of $\beta_\sigma$.
  We construct a selection $\overline{x}$ which selects fixed number of items from each intersection and satisfies $\inangle{\overline{x}, w} \leq \inangle{\hx, w}$.
  \newcommand{\hS}{\widehat{S}}
  Construct $\overline{x}$ by first picking from $I_{\sigma_1}$ in non-decreasing order of latent utility, then from $I_{\sigma_2}$ in non-decreasing order of latent utility, and so on, until $n$ items have been selected.
  Let $\hS$ be the set of items picked by $\hx$ and $\overline{S}$ be the set of items picked by $\overline{x}.$
  We claim that any item in $\hS\backslash \overline{S}$ has an equal or larger latent utility than any item in $\overline{S}\backslash \hS$.
  To see this, consider $i\in \hS\backslash \overline{S}$ and $j\in \overline{S}\backslash \hS$ and suppose, toward a contradiction, that
  $w_i < w_j.$
  Suppose $i\in I_{\sigma_\ell}$ and $j\in I_{\sigma_k}$.
  If $k>\ell$, then $\overline{x}$ would have picked $i$ before $j$, and hence, $i\in \overline{S}$.
  Which contradicts the fact that $i\in \hS\backslash \overline{S}$.
  Hence, we must have that $k\leq \ell$, and hence
  \begin{align*}
    \beta_{\sigma_k}\geq \beta_{\sigma_\ell}.
  \end{align*}
  \noindent Further, since $\hx$ selected $i$ but not $j$, we must have
  $\hw_i > \hw_j.$
  Using this, we have
  \begin{align*}
    \hw_i
    &= \beta_{\sigma_\ell} w_i > \beta_{\sigma_k} w_j= \hw_j
  \end{align*}
  Since $\beta_{\sigma_k}\geq \beta_{\sigma_\ell}$ this implies that
  \begin{align*}
    w_i > w_j.
  \end{align*}
  We have a contradiction.
  Now, we can show that $\inangle{\overline{x}, w} \leq \inangle{\hx, w}$, as follows:
  \begin{align*}
    \inangle{\overline{x}, w}
    &= \sum_{i\in \overline{S}} w_i
    = \sum_{i\in \overline{S} \cap \hS } w_i + \sum_{i\in \overline{S}\backslash \hS } w_i.
  \end{align*}
  Since both $\overline{x}$ and $\hx$ select $n$ items, we have $|\overline{S}\backslash \hS|=|\hS\backslash \overline{S}|.$
  Using this and the fact that any item in $\hS\backslash \overline{S}$ has a higher latent utility than any item in $\overline{S}\backslash \hS$, we get:
  \begin{align*}
    \inangle{\overline{x}, w} &=\sum_{i\in \overline{S} \cap \hS } w_i + \sum_{i\in \overline{S}\backslash \hS } w_i\\
    &\leq  \sum_{i\in \overline{S} \cap \hS } w_i + \sum_{i\in \hS\backslash\overline{S}} w_i\\
    &=  \sum_{i\in  \hS} w_i\\
    &=  \inangle{\hx, w}.\yesnum\label{eq:appendix_new}
  \end{align*}
  Suppose that for each $\sigma\in \zo^p$, $\overline{x}$ picks $k_\sigma\in \N$ items from $I_\sigma$.
  By the construction of $\overline{x}$, we know that $\inbrace{k_\sigma}_\sigma$ are fixed constants.
  For each intersection $\sigma$, where $\overline{x}$ picks at least one item, we have
  \begin{align*}
    \frac{\sum_{i\in I_\sigma} \overline{x}_i w_i}{ \sum_{i\in I_\sigma} \overline{x}_i } = \frac{\sum_{i\in I_\sigma} \overline{x}_i w_i}{ k_\sigma } \geq \frac1{|I_\sigma|}\sum_{i\in I_\sigma} w_i.\yesnum\label{eq:app_new}
  \end{align*}
  Using this we have
  \begin{align*}
    \Ex\insquare{\inangle{\hx, w}} \quad &\Stackrel{\eqref{eq:appendix_new}}{\geq} \quad  \Ex\insquare{\inangle{\overline{x}, w}}\\
    &=\quad \sum_{\sigma: |\overline{S}\cap I_\sigma|>0}\Ex\insquare{\sum_{i\in I_\sigma} \overline{x}_i w_i}\\
    &\Stackrel{\eqref{eq:app_new}}{=}\quad  \sum_{\sigma: |\overline{S}\cap I_\sigma|>0}\Ex\insquare{k_\sigma\cdot \frac{\sum_{i\in I_\sigma} w_i}{|I_\sigma|}}\\
    &=\quad \sum_{\sigma: |\overline{S}\cap I_\sigma|>0}k_\sigma\cdot{ \frac{\sum_{i\in I_\sigma} \Ex_{v\sim \cD}\insquare{v}}{|I_\sigma|}} \tag{Each $w_i$ is drawn independently from $\cD$}\\
    &=\quad \Ex_{v\sim \cD}\insquare{v} \cdot \sum_{\sigma: |S\cap I_\sigma|>0}k_\sigma \\
    &=\quad \Ex_{v\sim \cD}\insquare{v} \cdot n.\yesnum\label{eq:appendix_eq_2}
  \end{align*}
  Further, we have
  \begin{align*}
    \inangle{\sx{}, w} =\sum\nolimits_{i=1}^m{\sx_i w_i}
    \leq \sum\nolimits_{i=1}^m{\sx_i \cdot \sup(\supp(\cD))}
    = n\cdot \sup(\supp(\cD)).\yesnum\label{eq:appendix_eq_3}
  \end{align*}
  Combining Equations~\eqref{eq:appendix_eq_2} and \eqref{eq:appendix_eq_3} we get
  \begin{align*}
    \ratio(L,\beta) = \Ex_w\insquare{\frac{\inangle{\hx, w}}{\sinangle{\sx{}, w}}}
    \ \ \quad \Stackrel{\eqref{eq:appendix_new},\eqref{eq:appendix_eq_3}}{\geq} \quad \ \  \Ex\insquare{\frac{\inangle{\overline{x}, w}}{ n\cdot \sup(\supp(\cD))   }}
    \ \ \Stackrel{\eqref{eq:appendix_eq_2}}{\geq}\ \  \Ex\insquare{\frac{\Ex_{v\sim \cD}\insquare{v} \cdot n}{ n\cdot \sup(\supp(\cD))   }}
    = {\frac{\Ex_{v\sim \cD}\insquare{v} }{ \sup(\supp(\cD))   }}.
  \end{align*}

\end{proof}

\vspace{-4mm}

\subsection{Proof of \cref{claim:invariance_of_r}}\label{sec:proofof:claim:invariance_of_r}
\begin{proof}
  Let $\gamma>0$ be such that $\cD^\prime$ be the distribution of $w\cdot \gamma$ where $w$ is drawn from $\cD$.
  Then
  \begin{align*}
    \ratio_{\cD^\prime}(L,\beta)
    &= \Ex_{w^\prime\sim \cD^\prime}\insquare{ \frac{ \sinangle{\tx,w^\prime} }{ \sinangle{\sx,w^\prime} } }
    =  \Ex_{w\sim \cD}\insquare{ \frac{ \sinangle{\tx,w\cdot \gamma} }{ \sinangle{\sx,w\cdot \gamma} } }
    \ \Stackrel{(\gamma>0)}{=} \ \Ex_{w\sim \cD}\insquare{ \frac{ \sinangle{\tx,w}}{ \sinangle{\sx,w}} }
    = \ratio_{\cD}(L,\beta).
  \end{align*}
\end{proof}

\end{document}